\documentclass{tlp}
\usepackage{amsmath}
\usepackage{changebar}
\usepackage{ifthen}
\usepackage{graphicx,amssymb}
\usepackage{latexsym}
\usepackage{xspace}
\usepackage{url}
\usepackage[center]{caption}
\usepackage{times}
\usepackage{enumerate}
\usepackage{epsfig}
\usepackage{empheq}
\usepackage{qtree}
\usepackage{xyling2}
\usepackage{dsfont}
\usepackage[usenames,dvipsnames]{xcolor}
\usepackage{tikz}
\usetikzlibrary{arrows,shapes}
\usepackage{tkz-graph}

\newcommand{\grap}[2]{\ensuremath{\langle #1,#2 \rangle}}

\newcommand{\expa}{\ensuremath{\mathit{exp}}}
\newcommand{\unexpa}{\ensuremath{\mathit{unexp}}}
\newcommand{\update}{\ensuremath{\mathit{update}}}
\newcommand{\succe}{\ensuremath{\mathit{succ}}}

\newcommand{\ct}{\ensuremath{\mathsc{ct}}}
\newcommand{\st}{\ensuremath{\mathsc{st}}}
\newcommand{\rl}{\ensuremath{\mathsc{rl}}}

\newcommand{\bl}{\ensuremath{bl}}

\DeclareMathSymbol{\FORALL}   {\mathord}{symbols}{"38}
\DeclareMathSymbol{\EXISTS}   {\mathord}{symbols}{"39}
\DeclareMathSymbol{\SUCHTHAT} {\mathbin}{symbols}{"01}
\DeclareMathAlphabet{\mathsc} {OT1}{cmr}{m}{sc}

\def\Exists#1#2{{\EXISTS#1}: #2}
\DeclareSymbolFont{AMSb}{U}{msb}{m}{n}
\DeclareMathSymbol{\N}{\mathbin}{AMSb}{"4E}
\DeclareSymbolFontAlphabet{\mathbb}{AMSb}

\newcommand{\set}[1]{\ensuremath{\{#1\}}}
\newcommand{\mset}[1]{\ensuremath{\{\mathit{#1}\}}}

\newcommand{\setmin}[2]{\ensuremath{#1\!\setminus\!#2}}

\newcommand{\clos}[1]{\ensuremath{\mathit{clos}(#1)}}

\newcommand{\SHOQ}{$\mathcal{SHOQ}$}
\newcommand{\ALCHOQ}{$\mathcal{ALCHOQ}$}
\newcommand{\SHQ}{$\mathcal{SHQ}$}

\newcommand{\SHIQ}{\ensuremath{\mathcal{SHIQ}}\xspace}

\newcommand{\SHIQD}{\ensuremath{\mathcal{SHIQ}(\mathbf{D})}\xspace}

\newcommand{\ALC}{\ensuremath{\mathcal{ALC}}\xspace}
\newcommand{\ALCH}{\ensuremath{\mathcal{ALCH}}\xspace}

\newcommand{\Ind}{\mathbf{I}}
\newcommand{\I}{\mathcal{I}}
\newcommand{\Int}{\mathcal{I}}
\newcommand{\DeltaI}{{\Delta^{\I}}}
\newcommand{\C}{\mathbf{C}}
\newcommand{\sqs}{\sqsubseteq}
\newcommand{\trans}{\mathsf{Trans}}
\newcommand{\sqsast}{\mbox{$\sqs\!\!\!\!\!{\scriptstyle {}^\ast}\,$}}

\newcommand{\Lc}[1]{\ensuremath{\mathcal{L}(#1)}}

\newcommand{\roo}{\varepsilon}

\newenvironment{program}{\[\begin{array}{rll}}{\end{array}\]}

\newcommand{\tsrule}[2]{\ensuremath{\mathit{#1} &\gets& \mathit{#2}\\}}

\newcommand{\prule}[2]{\ensuremath{\mathit{#1}\gets\mathit{#2}}}
\newcommand{\psrule}[2]{\ensuremath{\mathit{#1\hspace{-0.1cm}}\gets\hspace{-0.1cm}\mathit{#2}}}
\newcommand{\nprule}[3]{\ensuremath{\mathit{#1}: \mathit{#2} \gets \mathit{#3}}}

\newenvironment{programn}{\[\begin{array}{rrll}}{\end{array}\]}
\newenvironment{programm}{\normalsize \[\begin{array}{llll}}{\end{array}\] \normalsize}

\newcommand{\nsrule}[3]{\ensuremath{\mathit{#1\hspace{-0.1cm}:}&\hspace{-0.4cm}\mathit{#2}&\hspace{-0.3cm}\gets&\hspace{-0.35cm}\mathit{#3} \\}}
\newcommand{\nrule}[3]{\ensuremath{\mathit{#1:}&\mathit{#2}&\gets&\mathit{#3} \\}}

\newcommand{\naf}[1]{\ensuremath{\mathit{not}~ #1}}
\newcommand{\NAF}{\ensuremath{\mathit{not}}}

\newcommand{\pred}[1]{{#1}}

\newcommand{\ground}[2]{{#1}_{#2}}

\newcommand{\preds}[1]{\ensuremath{\mathit{preds}(#1)}}

\newcommand{\upreds}[1]{\ensuremath{upreds(#1)}}
\newcommand{\bpreds}[1]{\ensuremath{bpreds(#1)}}

\newcommand{\lits}[1]{\ensuremath{lits(#1)}}

\newcommand{\atoms}[1]{\ensuremath{atoms(#1)}}

\newcommand{\vars}[1]{\ensuremath{\mathit{vars}(#1)}}

\newcommand{\exptime}{\textsc{exptime}}

\newcommand{\np}{\textsc{np}}

\newcommand{\nexptimenp}[1]{\ensuremath{\textsc{nexptime}^{\textsc{np}}}}

\newcommand{\nexptime}[1]{\ensuremath{\textsc{nexptime}}}

\newcommand{\xnexptime}[1]{\ensuremath{{#1}\mbox{-}\textsc{nexptime}}}

\newcommand{\CS}{\ensuremath{\mathit{CS}}}

\newcommand{\numberrestless}[3]{(\leq #1~#2.#3)}

\newcommand{\numberrestgreater}[3]{(\geq #1~#2.#3)}

\newcommand{\posi}[1]{\ensuremath{{#1}^{+}}}
\newcommand{\nega}[1]{\ensuremath{{#1}^{-}}}

\newcommand{\lit}[1]{\ensuremath{\mathit{#1}}}

\newcommand{\axiom}[2]{\ensuremath{\mathit{#1} \sqs \mathit{#2}}}

\newcommand{\taxiom}[2]{\ensuremath{\mathit{#1} &\sqs &\mathit{#2}\\}}

\newenvironment{knowb}{\[\begin{array}{rll}}{\end{array}\]}

\newcommand{\cts}[1]{\ensuremath{\mathit{cts}{(#1)}}}

\newcommand{\dlrom}{\ensuremath{\mathcal{DLRO}^{-\set{\leq}}}}

\newenvironment{programxy}{\begin{array}{rll}}{\end{array}}

\newcommand{\dlpluslog}{\ensuremath{\mathcal{DL}\text{+}\mathit{log}}}

\newcommand{\EF}{\ensuremath{\mathit{EF}}}
\newcommand{\ef}{\ensuremath{\mathit{ef}}}

\newcommand{\F}{\ensuremath{\mathit{F}}}
\newcommand{\ES}{\ensuremath{\mathit{ES}}}

\newtheorem{definition}{Definition}
\newtheorem{example}{Example}

\newtheorem{lemma}{Lemma}
\newtheorem{claim}{Claim}
\newtheorem{proposition}{Proposition}
\newtheorem{corollary}{Corollary}
\newcommand{\qed}{\nobreak \ifvmode \relax \else
      \ifdim\lastskip<1.5em \hskip-\lastskip
      \hskip1.5em plus0em minus0.5em \fi \nobreak
      \vrule height0.75em width0.5em depth0.25em\fi}
\submitted{30 November 2009}
\revised{30 May 2011, 21 September 2011}
\accepted{4 October 2011}

\begin{document}
\title[Reasoning with FoLPs and f-hybrid
KBs]{Reasoning with Forest Logic Programs and f-hybrid Knowledge Bases\footnote{A preliminary version of this paper appeared in the
proceedings of the  \emph{European Semantic Web Conference 20009
(ESWC2009). We extended that paper with detailed examples, a more detailed description of the algorithm and of the fragment of f-hybrid knowledge bases, a detailed characterisation of simple FoLPs, as well as with proofs for all theorems.}
\cite{feier+heymans-HybridReasoningForestLogicPrograms:09}.}}

\author[Cristina Feier and Stijn Heymans]{CRISTINA FEIER, STIJN HEYMANS\thanks{This work is partially supported by the Austrian Science Fund (FWF) under the projects P20305 and P20840, and by the European Commission under the project OntoRule (IST-2009-231875).}\\
Knowledge-Based Systems Group, Institute of Information Systems\\
Vienna University of Technology \\
Favoritenstrasse 9-11, A-1040 Vienna, Austria \\
\email{\{feier,heymans\}@kr.tuwien.ac.at}}
\maketitle

\begin{abstract}
Open Answer Set Programming (OASP) is an undecidable framework for integrating ontologies and rules. Although several decidable fragments of OASP have been identified, few reasoning procedures exist. In this article, we provide a sound, complete, and terminating algorithm for satisfiability checking w.r.t. Forest Logic Programs (FoLPs), a fragment of OASP where rules have a tree shape and allow for inequality atoms and constants. The algorithm establishes a decidability result for FoLPs. Although believed to be decidable, so far only the decidability for two small subsets of  FoLPs, local FoLPs and acyclic FoLPs, has been shown. We further introduce f-hybrid knowledge bases, a hybrid framework where \SHOQ{} knowledge bases and forest logic programs co-exist, and we show that reasoning with such knowledge bases can be reduced to reasoning with forest logic programs only. We note that f-hybrid knowledge bases do not require the usual (weakly) DL-safety of the rule component, providing thus a genuine alternative approach to current integration approaches of ontologies and rules.
\end{abstract}


 \section{Introduction}
\label{sec:intro}

Integrating Description Logics (DLs) with rules for the Semantic Web has received considerable attention. Such approaches for combining rules and ontologies are \emph{Description Logic Programs} \cite{grosof}, \emph{DL-safe rules} \cite{motik}, \dlpluslog{} \cite{rosati-kr2006}, \emph{dl-programs} \cite{eiter-ai2008}, \emph{Description Logic Rules} \cite{Krotzsch+Rudolph+Hitzler-DLRules:2008}, and Open Answer Set Programming (OASP) \cite{heymans-acm2008}. OASP combines attractive features from the DL and the Logic Programming (LP) world: an open domain semantics from the DL side allows for stating generic knowledge, without the need to mention actual constants, and a rule-based syntax from the LP side supports nonmonotonic reasoning via \emph{negation as failure}. Concretely, Open Answer Set Programming is an extension of (unsafe) function-free Answer Set Programming \cite{gelfond88stable} with open domains, i.e., the syntax remains the same, the semantics is still stable-model based, but programs are interpreted w.r.t. open domains, i.e., non-empty arbitrary domains which extend the Herbrand universe.
\begin{example}
Consider the following program:
\begin{program}
\tsrule{fail(X)}{\naf{pass(X)}}
\tsrule{pass(john)}{}
\end{program}
 Although the predicate $fail$ is not satisfiable under the ordinary answer set semantics -- the only answer set being $\{pass(john)\}$ -- it is satisfiable under the open answer set semantics.  If one considers, for example, the universe $\mset{\lit{john},x}$, with $x$ some individual which does not belong to the Herbrand universe, there is an open answer set $\mset{pass(john),fail(x)}$ which satisfies $fail$.
\end{example}

Open Answer Set Programming is undecidable. One way to obtain decidable fragments is to impose syntactical restrictions while carefully safe-guarding enough expressiveness for integrating rule- and ontology-based knowledge. Such restrictions typically ensure the \textit{tree-model property}: predicates are either unary or binary, and if a unary predicate $p$ is satisfiable then there is a model which can be seen as a labeled tree such that: each node of the tree is labeled with a set of unary predicates, the label of the root includes $p$, and each arc is labeled with a set of binary predicates.

Such a restriction led to \emph{Conceptual Logic Programs (CoLPs)} \cite{heymans-amai2006} which are able to simulate reasoning in the DL \SHQ{}. CoLPs make use only of unary and binary predicates and disallow the presence of constants in programs. They also impose some constraints on the shape of rules: unary and binary rules are tree-shaped rules which have as head a single unary atom and binary atom, respectively. The tree-like structure of rules refers to the chaining pattern of rule variables: one variable can be seen as the root of a tree and the others as successors of the root such that for every arc in the tree there is a positive binary literal in the body which connects the two corresponding variables. Inequalities between `successor' variables can also appear in the body of such a rule; we will refer to the set of literals in the body of a rule formed only with the help of the `root' variable as the `local part' of the rule and to the remaining part of the rule body as the `successor part' of the rule. Constraints, i.e., rules with empty head, are also allowed, but their body also has to be tree-shaped, so that they can be simulated via unary rules. Another type of rules which can appear in CoLPs are so-called \emph{free rules} which have one of the following shapes: $\prule{a(X)\lor \naf{a(X)}}{}$ or $\prule{f(X,Y)\lor \naf{f(X,Y)}}{}$, where $a$ is a unary predicate and $f$ is a binary predicate. Conceptual Logic Programs were proved to be decidable by a reduction of satisfiability checking to checking non-emptiness of two-way alternating tree automata \cite{heymans-amai2006}.
\begin{example} The following program $P$ is a CoLP which describes the fact that somebody is happy if she meets a friend who is happy or an enemy who is unhappy, and somebody is unhappy if she meets an enemy who is happy or a friend who is not happy. This is expressed by means of four unary tree-shaped rules, $r_1$-$r_4$, each of these rules having $X$ as the root variable and $Y$ as the successor of $X$. Furthermore, somebody is happy if she has at least two different friends: rule $r_5$ captures this knowledge in a tree-style fashion, $X$ being the root of a tree, and $Y$ and $Z$ its distinct successors (expressed by the inequality in the body of the rule). The binary predicates $sees$, $friend$, and $enemy$ are free predicates, i.e., they are defined only via free rules. The last two rules are constraints which disallow that somebody is friend and enemy with the same person, or that somebody is at the same time both happy and unhappy.

\begin{programm}
\label{colp}
\nsrule{r_1}{happy(X)}{sees(X,Y),friend(X,Y),happy(Y)}
\nsrule{r_2}{happy(X)}{sees(X,Y), enemy(X,Y), unhappy(Y)}
\nsrule{r_3}{unhappy(X)}{sees(X,Y), friend(X,Y), \naf{happy(Y)}}
\nsrule{r_4}{unhappy(X)}{sees(X,Y), enemy(X,Y), happy(Y)}
\nsrule{r_5}{happy(X)}{friend(X,Y), friend(X,Z), Y \neq Z}
\nsrule{r_6}{sees(X,Y)\lor \naf{sees(X,Y)}}{}
\nsrule{r_7}{friend(X,Y)\lor \naf{friend(X,Y)}}{}
\nsrule{r_8}{enemy(X,Y)\lor \naf{enemy(X,Y)}}{}
\nsrule{r_9}{}{happy(X), unhappy(X)}
\nsrule{r_{10}}{}{friend(X,Y), enemy(X,Y)}
\end{programm}

Next figure describes a tree-shaped open answer set with universe $\{x,y,z,t\}$ and interpretation $\{unhappy(x), sees(x,y), enemy(x,y), happy(y),$ $friend (y,z),$ $friend(y,t)\}$ -- one can see from this that $unhappy$ is tree-satisfiable: $x$ is unhappy as she sees an enemy $y$ which in turn is happy, as she has at least two different friends, $z$ and $t$. Note that there are no empty labels on the arcs of the tree and $y$ does not see either of her friends $z$ and $t$; otherwise, as it is not known either about $z$ or about $t$  that they are happy, seeing them would render $y$ unhappy (according to rule $r_3$), and that would lead to an inconsistency (according to rule $r_9$).

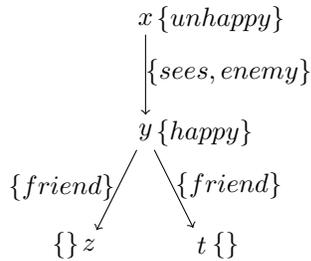
\begin{center}
\begin{tikzpicture}[ auto]
    \node (x) {$x$}
          child{
                node (y) {$y$}
		edge from parent [->]
                child{
			node (z) {$z$}
			edge from parent
			node[above,near end,xshift=-0.6cm] {$\{friend\}$}
		}
    		child{
			node (t) {$t$}
			edge from parent
			node[above,near end,xshift=0.6cm] {$\{friend\}$}
		}
          node[above,near end,xshift=1.1cm] {$\{sees,enemy\}$}};
     \node[right, xshift=-0.2cm] at (x.east) {$\{unhappy\}$};
    \node[right, xshift=-0.2cm] at (y.east) {$\{happy\}$};
    \node[left, xshift=0.2cm] at (z.west) {$\{\}$};
    \node[right][xshift=0.2cm] at (t.west) {$\{\}$};

\end{tikzpicture}
\end{center}
\normalsize
\end{example}

Another fragment of OASP, called \emph{Forest Logic Programs (FoLPs)}, has, as its name suggests, the \textit{forest-model property} \cite{heymans-jal2007}. The \textit{forest-model property} is a generalization of the tree-model property: if a unary predicate $p$ is satisfiable then it is satisfied by a model that can be seen as a special type of labeled forest, where the forest contains for each constant in the program a tree having as root the corresponding constant, and possibly an additional tree with an anonymous root. The forest is special in the sense that it can contain additional arcs from any node in the forest to one of the roots, standing for constants. FoLPs implement the forest-model property by allowing also for constants in the programs. Rules have practically the same tree-shape as CoLPs, with the exception of constants not being treated as successors in the tree\footnote{This means that the `root' term does not necessarily have to be linked with a successor term which is a constant via a binary atom.}. As such, FoLPs are generalizations of CoLPs and are expressive enough to deal with the DL \SHOQ{} (the presence of constants allows the simulation of DL nominals).

\begin{example}
Consider a slightly modified version of the CoLP $P$, $P^{\prime}$:

\[\begin{array}{llll}
\label{local}
\ensuremath{r_{1}}: & & & \\
\ensuremath{\ldots} & & &\\
\ensuremath{r_{10}}: & & & \\
\nrule{r_{11}}{unhappy(j)}{hungry(j)}
\nrule{r_{12}}{hungry(j)}{}
\end{array}
\]

\vspace{7mm}
Two new rules, $r_{11}$ and $r_{12}$, both referencing a constant $j$, have been added to the CoLP. The figure below describes a forest-shaped open answer set with universe $\{j,x,y\}$ and interpretation $\{unhappy(j),$ $hungry(j),$ $ happy(x), $ $sees(x,y),$ $friend(x,y),$ $happy(y),$ $enemy(y,j),$  $sees (y, j)\}$ -- one can see from this that $happy$ is forest-satisfiable: $x$ is happy as it sees a friend $y$ which at its turn is happy, as it sees an enemy, $j$, who is unhappy because it is hungry. The forest is composed of two trees, one with root $j$, the constant appearing in the program, and the other one with root $x$, where $x$ is an anonymous individual, whose content contains the predicate checked to be satisfiable, $happy$.

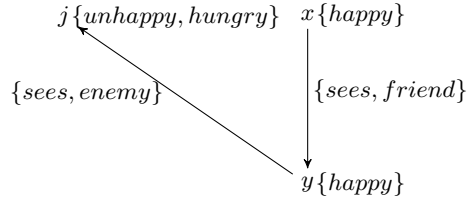
\small
\begin{center}
\begin{tikzpicture}[level distance=2.5cm,scale=0.9, >=stealth',bend angle=45]
    \node(x){$x$}
          child{
                node (y) {$y$}
		        edge from parent [->]
                child[grow=left, above, yshift=2.2cm, xshift=-1.1cm]{
			         node (z) {$j$}
			         edge from parent [bend right]
			         node[below,near end, yshift=-0.12cm, xshift=-0.6cm] {$\{sees,enemy\}$}
                }
          node[above,near end,yshift=0.3 cm,xshift=1.1cm] {$\{sees,friend\}$}
    };
    \node[right, xshift=-0.2cm] at (x.east) {$\{happy\}$};
    \node[right, xshift=-0.2cm] at (y.east) {$\{happy\}$};
    \node[right, xshift=-0.2cm] at (z.east) {$\{unhappy, hungry\}$};
    \end{tikzpicture}
\end{center}
\normalsize
\end{example}

A serious shortcoming of both CoLPs and FoLPs is their lack of effective reasoning procedures. Furthermore, it has not been known so far whether satisfiability checking w.r.t. Forest Logic Programs (FoLPs) is decidable. The decidability of two closely-related fragments of FoLPs, local FoLPs, and acyclic FoLPs, together with a reasoning procedure (for both fragments) based on a reduction to ordinary ASP reasoning has been provided in \cite{heymans-jal2007}. Both fragments are quite inexpressive compared to the whole FoLP fragment. For example, local FoLPs allow only the presence of negated atoms in the successor part of the tree structure of the unary or binary rules\footnote{This restriction does not apply to literals who have a constant as argument.}.

The reduction of reasoning to the ordinary ASP case has been made possible by the fact that local and acyclic FoLPs have the \emph{bounded finite model property}, i.e., if there is an open answer set, then there is an open answer set with a universe that is bounded by a number of elements that can be specified in function of the program at hand.
\begin{example}
The FoLP $P^{\prime}$ can be `adapted' into a local FoLP as follows:

\begin{programm}
\nrule{r_1}{happy(X)}{sees(X,Y),friend(X,Y),}
&&&\ensuremath{\naf{unhappy(Y)}}\\
\nrule{r_2}{happy(X)}{sees(X,Y), enemy(X,Y), }
&&&\ensuremath{\naf{happy(Y)}}\\
\nrule{r_3}{unhappy(X)}{sees(X,Y), friend(X,Y),}
&&&\ensuremath{\naf{happy(Y)}}\\
\nrule{r_4}{unhappy(X)}{sees(X,Y), enemy(X,Y), }
&&&\ensuremath{\naf{unhappy(Y)}}\\
\nrule{r_5}{happy(X)}{friend(X,Y), friend(X,Z), }
&&&\ensuremath{Y \neq Z}\\
\nrule{r_6}{sees(X,Y)\lor \naf{sees(X,Y)}}{}
\nrule{r_7}{friend(X,Y)\lor \naf{friend(X,Y)}}{}
\nrule{r_8}{enemy(X,Y)\lor \naf{enemy(X,Y)}}{}
\nrule{r_9}{}{happy(X), unhappy(X)}
\nrule{r_{10}}{}{friend(X,Y), enemy(X,Y)}
\nrule{r_{11}}{unhappy(j)}{hungry(j)}
\nrule{r_{12}}{hungry(j)}{}
\end{programm}
\normalsize

Note that the two programs, the original FoLP and the local FoLP, are not equivalent: for example, the infinite universe $\{x_1, x_2, x_3, \ldots\}$ and the infinite interpretation $\{happy(x_1),$ $friend(x_1, x_2),$ $sees(x_1,x_2),$ $happy(x_2),$ $friend (x_2, x_3),$ $sees(x_2, x_3),$ $\ldots\}$ form an open answer set of the local FoLP, but they do not form an open answer set of the general FoLP.
\end{example}

Finally, another fragment with reasoning support consists of simple CoLPs. Simple CoLPs are CoLPs that disallow the use of inequality and impose a restriction as concerns predicate recursion, but that are still expressive enough to simulate the DL \ALCH{}. In  \cite{feier+heymans-SoundComplAlgSimpleCoLPs:08}, a sound and complete tableaux-algorithm for \emph{simple CoLPs} has been devised. 
The algorithm constructs so-called completion structures, which are finite representations of (partial) models. The particular restriction on predicate recursion is a sufficient condition to establish the bounded finite model property and to enable the usage of a simple subset blocking condition to ensure the termination of the algorithm.
As is usual in Description Logics \cite{dlbook}, subset blocking consists in checking whether the label of a node of the forest is a subset of the label of one of its ancestors; if this is the case, the initial node is said to be `blocked' by its ancestor, and it is no longer expanded as the content of its label can be justified in a similar way as the content of the label of its ancestor.

In this article, we provide a tableaux-based algorithm for reasoning with the full fragment of FoLPs, and thus implicitly also with full CoLPs: in order to check whether a unary predicate is satisfiable, the algorithm tries to construct a forest model which satisfies the predicate. This is done by evolving a so-called completion structure which essentially is a forest shaped structure which describes a forest model in construction. When certain conditions are met, such a structure is said to be \emph{complete} and \emph{clash-free} and can be unraveled to an actual forest model. The algorithm can be seen as an extension of the algorithm for reasoning with simple CoLPs \cite{feier+heymans-SoundComplAlgSimpleCoLPs:08}; however, due to the lack of any restriction concerning predicate/literal recursion, things get significantly more complex. Unlike in the case of simple CoLPs, termination can no longer be ensured by a classical subset blocking condition; using only such a condition for stopping the expansion of a branch can lead to unsound results: the interpretation obtained by unraveling a clash-free complete completion structure may contain infinite chains of atoms, where the presence of each atom in the interpretation is justified by the presence of next atom. This violates a result regarding OASP which says that every atom in an open answer set has to be finitely justified \cite[Theorem 2]{heymans-amai2006}. A more complex blocking condition has been devised, which when applied guarantees soundness, but which no longer ensures termination, as it may never be fulfilled in the expansion process. However, it turns out that FoLPs, like local and acyclic FoLPs, also have the bounded finite model property: termination is then ensured by exploring forest branches only up to a certain depth.

The algorithm runs in the worst case in double exponential time, one exponential level higher than the algorithm for reasoning with simple CoLPs. The increase in complexity (compared to the algorithm for simple CoLPs, but also compared to tableaux procedures for reasoning with \SHOQ{}) is due to the interaction between the requirement concerning the minimality of open answer sets and the unrestricted recursion in rules which leads to a double exponential bound on the number of individuals which might be needed to satisfy a certain predicate.

We also define simple FoLPs as a particular kind of FoLPs which are in a similar relationship with FoLPs as simple CoLPs with CoLPs: there is a similar restriction on predicate recursion, but unlike the case of simple CoLPs we allow also the presence of constants and inequalities in rule bodies. The algorithm can be simplified in such a case and the worst case complexity drops one exponential level. Simple FoLPs can be seen as a generalization of local FoLPs and acyclic FoLPs.

As already mentioned, FoLPs serve well as an underlying integration vehicle for ontologies and rules. In order to illustrate this, we define \emph{f-hybrid knowledge bases (fKBs)}, consisting of a \SHOQ{} knowledge base and a rule component that is a FoLP, with a non-monotonic semantics similar to the semantics of \dlpluslog{} \cite{rosati-kr2006}, \emph{r-hybrid knowledge bases} \cite{rosati-rr2008}, and \emph{g-hybrid knowledge bases} \cite{heymans-tplp2008}. Our approach differs in two points with current other proposals:

\begin{itemize}
\item In contrast with Description Logic Programs, DL-safe rules, and Description Logic Rules, f-hybrid knowledge bases have, in line with traditional logic programming paradigms, a minimal model semantics for the rule component, thus allowing for nonmonotonic reasoning.
\item To ensure effective reasoning, our approach does not rely on a (weakly) DL-safeness condition such as \cite{motik,rosati-kr2006,rosati-rr2008}, which restricts the interaction of the rule component with the DL component. Instead, we rely on a translation of the hybrid knowledge to FoLPs.
\end{itemize}

The major contributions of the paper can be summarized as follows:

\begin{itemize}

\item We define in Section \ref{sec:algorithm} an algorithm for deciding satisfiability w.r.t. FoLPs, inspired by tableaux-based methods from DLs. We show that this algorithm is terminating, sound, and complete, and runs in double exponential time. The algorithm is non-trivial from two perspectives: both the minimal model semantics of OASP, compared to the model semantics of DLs, as well as the open domain assumption, compared to the closed domain assumption of ASP \cite{gelfond88stable}, pose specific challenges.

\item We show in Section \ref{sec:fhybrid} that FoLPs are expressive enough to simulate the DL \SHOQ{} with \emph{fKBs} as an alternative characterization for hybrid representation and (nonmonotonic) reasoning of knowledge, that supports a tight integration of ontologies and rules.

\end{itemize}

The article is organized as follows. A short overview of Open Answer Set syntax and semantics together with some notations are presented in Section \ref{sec:prelim}. Next, Section \ref{sec:FLP} formally introduces  FoLPs and the \emph{forest model property}. The actual tableaux algorithm for reasoning with FoLPs is described in Section \ref{sec:algorithm}. A new hybrid formalism, \emph{f-hybrid KBs}, which combines \SHOQ{} KBs with FoLPs, is introduced in Section \ref{sec:fhybrid}. Reasoning with the new formalism is enabled by a concept satisfiability preserving translation from  \SHOQ{} KBs to FoLPs, the translation being described in the same section. A less expressive fragment of FoLPs, simple Forest Logic Programs, is described in Section \ref{sec:simpleFoLPs}. Finally, Section \ref{sec:discussion} discusses some related work, while Section \ref{sec:conclusions} draws some conclusions. Detailed proofs can be found in the Appendix. 

\section{Preliminaries}
\label{sec:prelim}
We recall the open answer set semantics from \cite{heymans-jal2007}. A term is either a \textit{constant} or a \textit{variable}\footnote{No function symbols are allowed.}, and is denoted by a string of letters where a constant starts with a lower-case letter and a a variable with an upper case letter. An atom is of the form $p(t_1,\ldots,t_n)$, where $p$ is a predicate name, and $t_1, \ldots, t_n$ are terms. We further allow for \emph{equality atoms} $s=t$, where $s$ and $t$ are terms.  A \textit{literal} is an atom $L$ or a negated atom $\naf{L}$.  An \emph{inequality literal} $\naf{(s=t)}$ will often be denoted with $s\neq t$. An atom (literal) that is not an equality atom (inequality literal) will be called a \emph{regular atom (literal)}. For a regular literal $L$, $pred(L)$, and $args(L)$ denote the predicate, and the (tuple of) arguments of $L$\footnote{If the literal $L$ has just one argument, $args(L)$ will return the argument itself.}, respectively. For a set $\alpha$ of literals or (possibly negated) predicates, $\posi{\alpha} = \set{l \mid l \in \alpha, l \text{ an atom or a predicate}}$ and $\nega{\alpha} = \set{l \mid \naf{l} \in \alpha, l \text{ an atom or a predicate}}$. For example, $\posi{\{a,\naf{b},c\neq d\}} = \{a\}$ and $\nega{\{a,\naf{b},c \neq d\}} =\{b,c=d\}$. For a set $S$ of atoms, $\naf{S} = \set{\naf{L} \mid L \in S}$. For a set of (possibly negated) predicates $\alpha$, we will often write $\alpha(x)$ for $\{a(x) \mid a \in \alpha\}$ and $\alpha(x,y)$ for $\{a(x,y) \mid a \in \alpha\}$.

\par
A \textit{program} is a countable set of rules \prule{\alpha}{\beta},  where $\alpha$ is a finite set of regular literals and $\beta$ is a finite set of literals. The set $\alpha$ is the \textit{head} of the rule and represents a disjunction, while $\beta$ is called the \textit{body} and represents a conjunction.  If $\alpha = \emptyset$, the rule is called a \textit{constraint}. \emph{Free rules} are rules ${q(t_1,\ldots,t_n)\lor\naf{q(t_1,\ldots,t_n)}\gets}{}$ for terms $t_1,\ldots,t_n$; they enable a choice for the inclusion
of atoms.  We call a predicate $q$ \emph{free} in a program if there is a free rule ${q(X_1,\ldots,X_n)\lor\naf{q(X_1,\ldots,X_n)}\gets}{}$ in the program, where $X_1,\ldots,X_n$ are variables. Atoms, literals, rules, and programs that do not contain variables are \textit{ground}. For a rule or a program $P$, let $\cts{P}$ be the constants in $P$, $\vars{P}$ its variables, and \preds{P} its predicates, with $\upreds{P}$ and $\bpreds{P}$, the unary and binary predicates, respectively. For every predicate $q$ and program $P$, let $P_q$ be the set of definite (i.e., disjunction free) rules of $P$ that have $q$ as a head predicate.  A \emph{universe} $U$ for a program $P$ is a non-empty countable superset of the constants in $P$: $\cts{P} \subseteq U$.  We call $\ground{P}{U}$ the ground program obtained from $P$ by substituting every variable in $P$ by every possible element in $U$.  Let $\atoms{P}$ ($\lits{P}$) be the set of regular atoms (literals) that can be formed from a ground program $P$. 

An \textit{interpretation} $I$ of a ground $P$ is a subset of $\atoms{P}$.  We write $I \models p(t_1,\ldots,t_n)$ if $p(t_1,\ldots,t_n) \in I$ and $I \models \naf{p(t_1,\ldots,t_n)}$ if $I\not\models p(t_1,\ldots,t_n)$. Furthermore, for ground terms $s$ and $t$ we write $I\models s=t$ if $s=t$ and $I\models\naf{s=t}$ or $I\models s\neq t$ if $s\neq t$. For a set of ground literals $X$, $I\models X$ if $I \models l$ for every $l \in X$.  A ground rule $r:\prule{\alpha}{\beta}$ is \textit{satisfied} w.r.t. $I$, denoted $I \models r$, if $I \models l$ for some $l \in \alpha$ whenever $I \models \beta$. A ground constraint \prule{}{\beta} is satisfied w.r.t. $I$ if $I \not\models \beta$.

For a ground program $P$ without \NAF, an interpretation $I$ of $P$ is a \textit{model} of $P$ if $I$ satisfies every rule in $P$; it is an \textit{answer set} of $P$ if it is a subset minimal model of $P$. For ground programs $P$ containing \NAF, the \textit{GL-reduct} \cite{gelfond88stable} w.r.t. $I$ is defined as $P^I$, where $P^I$ contains \prule{\posi{\alpha}}{\posi{\beta}} for \prule{\alpha}{\beta} in $P$, $I \models \naf{\nega{\beta}}$, and $I \models \nega{\alpha}$.  $I$ is an \textit{answer set} of a ground $P$ if $I$ is an answer set of $P^I$.

In the following, a program is assumed to be a finite set of rules; infinite programs only appear as byproducts of grounding a finite program with an infinite universe. An \textit{open interpretation} of a program  $P$ is a pair $(U, M)$ where $U$ is a universe for $P$ and $M$ is an interpretation of $P_{U}$.  An \textit{open answer set} of $P$ is an open interpretation $(U,M)$ of $P$ with $M$ an answer set of $\ground{P}{U}$. An $n$-ary predicate $p$ in $P$ is \emph{satisfiable w.r.t. $P$}  if there is an open answer set $(U,M)$ of $P$ and a $(x_1,\ldots,x_n) \in U^n$ such that $p(x_1,\ldots,x_n) \in M$.

We introduce some notations for trees which extend those in \cite{Vardi98reasoningabout}. Let $\cdot$ be a concatenation operator between different symbols such as constants or natural numbers. A \emph{tree} $T$ with root $c$ (also denoted as $T_c$), where $c$ is a specially designated constant, is a set of nodes, where each node is a sequence of the form $c \cdot s$, where $s$ is a (possibly empty) sequence of positive integers formed with the help of the concatenation operator; for $x \cdot d\in T$, $d\in \N^{*}$\footnote{$\N^{*}$ is the set of positive integers}, we must have that $x\in T$.  For example a tree with root $c$ and 2 successors will be denoted as $\{c, c\cdot 1,c\cdot 2\}$ or $\{c,c1,c2\}$ \footnote{By abuse of notation, we consider that there are at most 9 successors for every node, so we can abbreviate $a \cdot b$ with $ab$}.

For a node $x \in T$, we call $\succe_{T}(x)=\{x \cdot n\in T \mid n \in \N^{*}\}$, \textit{successors}\index{successor} of $x$ in $T$. As the successorship relation is captured in the codification of the nodes, a tree is literally the set of its nodes. The \emph{arity} of a tree is the maximum amount of successors any node has in the tree.  The set $A_{T}=\{(x,y) \mid x, y\in T, \Exists{n\in\N^{*}}{y=x \cdot n}\}$ denotes the set of arcs of a tree $T$. We define a partial order $\leq_T$ on a tree $T$ such that for $x,y\in T$, $x \leq_T y$ iff $x$ is a prefix of $y$. As usual, $x <_T y$ if $x \leq_T y$ and $y \not \leq_T x$. A \emph{path from} $x$ \emph{to} $y$ in $T$, where $x<_T y$, denoted with $path_{T}(x,y)$, is a subset of $T$ which contains all nodes which are at the same time greater or equal to $x$ in $T$ and lesser or equal to $y$ in $T$ according to the partial order relation, i.e., $path_{T}(x,y)=\{z\mid x\leq_T z \leq_T y\}$. 
A branch $B$ in a tree $T_c$ is a maximal path (there is no path in $T_c$ which strictly contains it).
We denote the \emph{subtree}\index{subtree} of $T$ at $x$ by $T[x]$, i.e., $ T[x] = \set{y\in T\mid x\leq_T y}$.

\par
A \emph{forest} $F$ is a set of trees $\{T_c \mid c \in C\}$, where $C$ is a finite set of arbitrary constants. The set of nodes $N_{F}$ of a forest $F$ and the set of arcs $A_{F}$ of $F$ are defined as follows: $N_F=\cup_{T\in F} T$ and $A_F=\cup_{T\in F}A_T$. For a node $x \in N_F$, we denote with $\succe_F(x)=\succe_T(x)$, where $x\in T$ and $T \in F$, the set of successors of $x$ in $F$. Also, as for trees, we define a partial order relationship $\leq_F$ on the nodes of a forest $F$ where $x \leq_F y$ iff $x \leq_T y$ for some tree $T$ in $F$.

An extended forest $\EF$ is a tuple $\langle F, \ES \rangle$ where $F =\{T_c \mid c\in C\}$ is a forest and $\mathit{ES}$ is a binary relation which contains tuples of the form $(x,y)$ where $x \in N_F$ and $y \in C$, i.e., $\mathit{ES}$ relates nodes of the forest with roots of trees in the forest. $\mathit{ES}$ extends the  successorship relation: $\succe_{\EF}(x)=\{y \mid y \in \succe_F(x) \text{ or } (x,y) \in \ES\}$.

Figure \ref{figure:extforest} depicts an extended forest.

\begin{figure}[htbp]
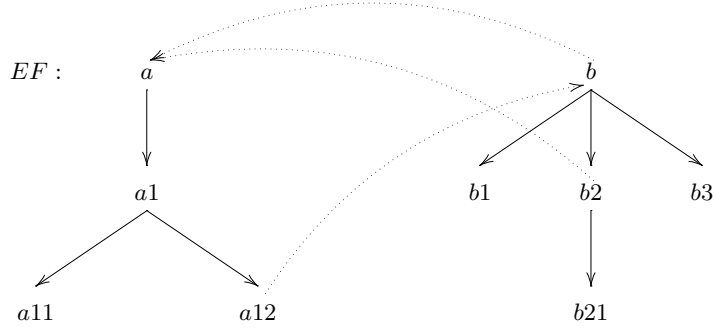

\vspace{5mm}
\begin{center}
\Treek[1]{3}{\K{$EF:$} &\K{$a$} \AR{d} &&&& \K{$b$}\AR{d}\AR{dl}\AR{dr}\UpLinkC[.>]{0,-4}&\\
&\K{$a1$} \AR{dl} \AR{dr} &&&\K{$b1$}&\K{$b2$}\AR{d}\UpLinkC[.>]{-1,-4}&\K{$b3$}\\
\K{$a11$}&&\K{$a12$}\LinkA[.>]{-2,3}&&&\K{$b21$}\\
}
\vspace{3mm}
\caption{An extended forest}
\label{figure:extforest}
\end{center}
\end{figure}

The presence of $\ES$ gives rise to so-called extended trees in $\EF$, where such a tree (actually, a particular type of graph) is one of $T_c \in F$, extended with the arcs $\{(x,y)\mid (x,y) \in \ES, x \in T_c\}$ and with the nodes $\{y \mid (x,y) \in \ES, x \in T_c\}$. The extension of $T_c$ in $\EF$  is denoted with $T_{c}^{\EF}$. For example, the extension of $T_a$ in  $\EF$ from Figure \ref{figure:extforest} contains the extra arc $(a12,b)$ and the extension of $T_b$ in $\EF$ contains the extra arcs $(b,a)$ and $(b2,a)$. An extended subtree with root $x$ of an extended tree $T_{c}^{\EF}$  is denoted with $T_{c}^{\EF}[x]$: it is defined (as a graph) as the extension of $T_c[x]$ with  the arcs $\{(y,z)\mid (y,z) \in \ES, y \in T_c[x]\}$ and with the nodes $\{z \mid (y,z) \in \ES,  y \in T_c[x]\}$. Finally, by $N_{\EF}=N_F$ we denote the set of nodes of an extended forest $\EF$ and by $A_{\EF}=A_F \cup \ES$ the set of arcs of $\EF$.

Finally, a directed graph $G$ is defined as usual by its sets of nodes $V$ and arcs $A$. We introduce two graph-related notations: $paths_G$ denotes the set of paths in $G$, where each path is a tuple of nodes from $V$: $paths_G=\{(x_1, \ldots, x_n) \mid ((x_i, x_{i+1}) \in A)_{1 \leq i <n}\}$, and $conn_G$ denotes the set of pairs of connected nodes from $V$: $conn_G=\{(x,y) \mid\ \Exists{Pt=(x_1, \ldots, x_n) \in paths_G}{x_1=x \wedge x_n=y}\}$. As an extended forest is a particular type of graph, these notations apply also to extended forests.
Additional notation needed for the proofs is introduced in the appendix.

\section{Forest Logic Programs}
\label{sec:FLP}

As mentioned in the introduction, \emph{Forest Logic Programs (FoLPs)} are a fragment of OASP which have the forest model property. In this section we formally introduce the fragment and the notions of \emph{forest satisfiability} and \emph{forest model property}.

\begin{definition}\label{def:FoLP} A \emph{forest logic program (FoLP)} is a program with only unary and binary predicates, and such that a rule is either:
\begin{itemize}
 \item a \emph{free rule}: \begin{equation}\label{eq:free}\prule{a(s)\lor \naf{a(s)}}{}\end{equation}  or,  \begin{equation}\prule{f(s,t)\lor \naf{f(s,t)}}{}\end{equation} where $s$ and $t$ are terms; 
\item a \textit{unary rule}: \begin{equation} \label{eq:unary}\prule{a(s)}{\beta(s), (\gamma_m(s,t_m),\delta_m(t_m))_{1\leq m \leq k},\psi}\end{equation} with $\psi\subseteq \bigcup_{1\leq i\neq j\leq k}\set{t_i\neq t_j}$ and $k \in \N$, or a \emph{binary rule}: \begin{equation}\label{eq:binary}\prule{f(s,t)}{\beta(s), \gamma(s,t), \delta(t)}\end{equation} where $a \in \upreds{P}$ and $f \in \bpreds{P}$, $s$, $t$, and $(t_m)_{1\leq m \leq k}$ are terms,  
$\beta,$ $\delta,$ $(\delta_m)_{1\leq m \leq k} \subseteq \upreds{P} \cup \naf(\upreds{P})$ (sets of (possibly negated) unary predicates), $\gamma$, $(\gamma_m)_{1\leq m \leq k} \subseteq \bpreds{P} \cup \naf(\bpreds{P})$ (sets of possibly negated binary predicates), and
\begin{enumerate}
\item equality and inequality do not appear in any $\gamma$: $\set{=, \neq} \cap \gamma_m = \emptyset$, for $1\leq m \leq k$, and $\set{=, \neq} \cap \gamma = \emptyset$;
\item there is a positive atom that connects the head term $s$ with any successor term which is a variable: $\gamma_m^+\neq\emptyset$, if $t_m$ is a variable, for $1\leq m \leq k$, and $\posi{\gamma} \neq \emptyset$, if $t$ is a variable;
\end{enumerate}
\item a \textit{constraint}: $\prule{}{a(s)} \mbox{ or } \prule{}{f(s,t)}$, where $s$ and $t$ are terms.
\end{itemize}
In every rule, all terms which are variables are distinct\footnote{This restriction precludes the presence in rules of literals of the form $f(X,X)$ or $\naf f(X, X)$ which would break the forest model property.}.
\end{definition}

\begin{example}
Consider again rule $r_5$ from Example 2: \nprule{r_5}{happy(X)}{friend(X,Y), friend(X,Z),} $Y \neq Z$. This rule is a unary rule with head term $X$, and $k=2$, i.e., there are two successor terms, variables $Y$ and $Z$. In this case $\beta=\emptyset$, $\gamma_1=\gamma_2=\{friend\}$, $\delta_1=\delta_2=\emptyset$, and $\psi=\{Y \neq Z\}$. There is an atom which links $X$ with each of the successor terms $Y$ and $Z$: $friend(X, Y)$ and $friend(X, Z)$, respectively.
\end{example}

Constraints can be left out of the fragment without losing expressivity. Indeed, a constraint $\prule{}{body}$ can be replaced by a rule of the form $\prule{constr(x)}{\naf{constr}(x),body}$, for a new predicate $\mathit{constr}$.

We denote with $degree(r)$, where $r$ is a unary rule as in (\ref{eq:unary}), the number $k$. Intuitively, $k$ indicates the maximum number of successor individuals needed to make the rule true. The degree of a free rule is 0.

For a unary predicate $p$, $degree(p)=max\{degree(r) \mid p \in head(r)\}$. Finally, the rank of a FoLP $P$ is defined as: $rank(P)=\sum_{p \in \upreds{P}} degree(P)$.

As already mentioned FoLPs have the \emph{forest model property}: if a unary predicate $p$ is satisfiable then there is a model which satisfies $p$ that can be seen as an extended forest. The forest contains for each constant in the program a tree having the constant as root, and possibly an additional tree with an anonymous root. The predicate checked to be satisfiable, $p$, belongs to the label of one of the root nodes. While the constants appearing in the program are mandatorily part of the universe of any model, having an anonymous root tree is considered as $p$ might be satisfied only in conjunction with an anonymous individual, and not a constant.

\begin{example}
Consider a program with two rules: \prule{q(a)}{p(a), \naf q(a)}, and \prule{p(X) \lor \naf p(X)}{}. While $p$ is satisfiable, $p(a)$ does not appear in any open answer set.
\end{example}

\begin{definition}\label{def:forest-sat1}
Let $P$ be a program.
A predicate $p\in \upreds{P}$ is \emph{forest satisfiable}\index{forest satisfiable} w.r.t. $P$ if there is an open answer set $(U,M)$ of $P$ and there is an extended forest
$\EF \equiv (\set{T_{\roo}} \cup \set{T_a \mid a \in \cts{P}},\ES)$, where $\roo$ is a constant, possibly one of the constants appearing\footnote{Note that in this case $T_{\roo} \in \set{T_a \mid a \in \cts{P}}$. Thus, the extended forest contains for every constant from $P$ a tree which has as root that specific constant and possibly, but not necessarily, an extra tree with unidentified root node.} in $P$, and a labeling function
$\mathcal{L}:\set{T_{\roo}} \cup \set{T_a \mid a \in \cts{P}}\cup A_{\EF} \to 2^{\preds{P}}$ such that
\begin{itemize}
\item $p \in \mathcal{L}(\roo)$,
\item $U = N_{\EF}$, and 
\item $M= \{\mathcal{L}(x)(x)\mid x \in N_{\EF}\} \cup \{\mathcal{L}(x,y)(x,y)\mid (x,y) \in A_{\EF}\}$\footnote{Remember that $\mathcal{L}(x)$ and $\mathcal{L}(x, y)$ are sets of unary and binary predicates, resp., and thus for every $p \in upreds(P)$: $p(x) \in M$ iff $p \in \mathcal{L}(x)$ and for every $f \in bpreds(P)$: $f(x,y) \in M$ iff $f \in \mathcal{L}(x, y)$.}, and
\item for every $(z, z \cdot i) \in A_{\EF}$: $\mathcal{L}(z, z \cdot i),  \neq \emptyset$.
\end{itemize}

We call such a $(U,M)$ a \emph{forest model}\footnote{Note that technically, a forest model is a subset minimal model as it is an open answer set.} and a program $P$ has the \emph{forest model property} if the following property holds:

If $p \in \upreds{P}$ is satisfiable w.r.t. $P$ then $p$ is forest satisfiable w.r.t. $P$.
\end{definition}

\begin{proposition}[\cite{heymans-jal2007}]\label{prop:forestmodelproperty1}
FoLPs have the forest model property.
\end{proposition}

\begin{example}
Let $\EF$ be the extended forest depicted in Example 3: $\EF=(\{T_{\roo}, T_j\}, \{(y,j)\})$, where $\roo=x$. According to the notation we introduced for trees, the successor of $x$ in $T_x$, $y$, has the form $x \cdot i$, with $i \in \N^{*}$. One can see that $happy$, the predicate checked to be satisfiable, is in the label of $\roo$: $happy \in \mathcal{L}(x)$, the universe $U$ of the open answer set is indeed equal to $N_{\EF}=\{x,y,j\}$, and every predicate symbol corresponding to some atom in $M$ is in the label of the argument of the atom, e.g.: $unhappy \in \mathcal{L}(j)$. The reciprocal also holds: every node/arc of the extended forest in conjunction with every predicate symbol in its label forms an atom which is part of the interpretation. It also holds that $x$ and $y=x \cdot i$ are linked by a positive binary predicate: $\mathcal{L}(x, y)^+=\{sees, friend\}\neq \emptyset$.
\end{example}

In \cite{feier+heymans-SoundComplAlgSimpleCoLPs:08}, we introduced the class of simple Conceptual Logic Programs. It is easy to see that every simple CoLP is an FoLP. As satisfiability checking w.r.t. simple Conceptual Logic Programs is \exptime-hard, the following property follows:

\begin{proposition}\label{prop:compforest}
Satisfiability checking w.r.t. FoLPs is \exptime-hard.
\end{proposition}

Note that, at present, we do not have a tight complexity characterization for FoLPs: we have a lower bound (\exptime) established by the inclusion of simple CoLPs in FoLPs, while the algorithm described in this article runs in the worst case in double exponential time, thus establishing an upper bound.

\section{An Algorithm for Forest Logic Programs}\label{sec:algorithm}

In this section, we define a sound, complete, and terminating algorithm for satisfiability checking w.r.t. FoLPs. In \cite{heymans-jal2007} it has been shown that several restrictions of FoLPs which have the finite model property are decidable, but there was no result so far regarding the whole fragment. Thus, the algorithm described in this section also establishes a decidability result for FoLPs.

The basic data structure for our algorithm is a \emph{completion structure}. A completion structure describes a forest model in construction. As such, the main components of the structure are an extended forest $EF$, the forest-shaped universe of the constructed open answer set, and a labeling function $\ct$, which assigns to every node, resp. arc of $EF$, a set of possibly negated unary, resp. binary predicates, called a \textit{content}. The presence of such a predicate symbol/negated predicate symbol in the content of some node or arc indicates the presence/absence in the forest model in construction of the atom formed with that predicate and the current node or arc as argument.
Note that unlike the labeling function $\mathcal{L}$ in definition \ref{def:forest-sat1} which describes which atoms are in the forest model, the labeling function $\ct$ keeps track also of which atoms are not in the forest model. This is needed as the forest model is updated by justifying the presence or absence of a certain atom in itself.

The presence (absence) of an atom in a forest model in construction is justified by imposing that the body of at least one ground rule which has the respective atom in the head is satisfied (no body of a rule which has the respective atom in the head is satisfied). In order to keep track which (possibly negated) predicate symbols in the content of some node or arc have already been justified a so-called status function is introduced. The status function $\st$ assigns the value $\unexpa$ to pairs of nodes/arcs and possibly negated unary/binary predicates which have not yet been `expanded', i.e. justified, and the value $\expa$ to such pairs which have already been considered.

Furthermore, in order to ensure that the constructed forest model is a well-supported one \cite{fages-NewfixpointSem:91}, or in other words, no atom in the model is circularly justified (does not depend on itself) or infinitely justified (does not depend on an infinite chain of other atoms), a graph $G$ which keeps track of dependencies between atoms in the model is maintained.

In the following, for a predicate $p$, $\pm p$ denotes $p$ or $\naf p$, whereby multiple occurrences of $\pm p$ in the same context refer to the same symbol (either $p$ or $\naf p$). The negation of $\pm p$ (in a given context)  is $\mp p$, that is, $\mp p = \naf{p}$ if $\pm p=p$ and $\mp p = p$ if $\pm p=\naf{p}$.

\begin{definition}\label{def:completionstructure} A \textit{completion structure for a FoLP $P$} is a tuple $\langle \EF,$ $\ct,$ $\st,$ $G \rangle$ where:
\begin{itemize}
\item $\EF=\langle F,\ES \rangle$ is an extended forest, its set of nodes being the universe of the forest model in construction,
\item $\ct:N_{\EF} \cup A_{\EF} \to 2^{preds (P) \cup \naf (preds(P))}$ is the `content' function which maps a node of the extended forest to a set of (possibly negated) unary predicates and an arc of the extended forest to a set of (possibly negated) binary predicates such that $\ct(x)\subseteq\upreds{P} \cup \NAF(\upreds{P})$ if $x \in N_{\EF}$, and $\ct(x)\subseteq \bpreds{P} \cup \NAF(\bpreds{P})$ if  $x \in A_{\EF}$,
\item $\st:\{(x, \pm q) \mid \pm q \in \ct(x), x \in N_{\EF} \cup A_{\EF}\}  \to \{\expa, \unexpa\}$ is the `status' function which indicates which predicates in the content of some node/arc are justified, and which are not,
\item $G=\grap{V}{A}$ is a directed graph with vertices $V \subseteq \atoms{\ground{P}{N_{\EF}}}$ and arcs $A\subseteq\atoms{\ground{P}{N_{\EF}}}\times\atoms{\ground{P}{N_{\EF}}}$,
\end{itemize}

\end{definition}

For checking satisfiability of a unary predicate $p$ w.r.t. a FoLP $P$, one starts with an initial completion structure which is defined as follows: the extended forest $EF$ is initialized with the set of single-node trees having as root a constant appearing in $P$ and possibly a new single-node tree with an anonymous root\footnote{This is in order to comply with the generic shape of a forest model described in section \ref{sec:FLP}.}. In case the anonymous root tree exists, its content is initialized with $\{p\}$, the predicate checked to be satisfiable. Otherwise the content of the root of one of the other trees is initialized with $\{p\}$. The contents of the other nodes (roots) are initialized with $\emptyset$. $G$ is initialized to the graph with a single vertex $p(\roo)$. 

\begin{definition}\label{def:initialcompletionstructure}
An \emph{initial completion structure} for checking satisfiability of a unary predicate $p$ w.r.t. a FoLP $P$ is a completion structure $\langle \EF,$ $\ct,$ $\st,$ $G \rangle$ with $\EF=\langle F,\ES \rangle$, $F=\set{T_{\roo}} \cup \set{T_a \mid a \in \cts{P}}$, where $\roo$ is a constant, possibly in $\cts{P}$, and $T_x=\set{x}$, for every $x \in \cts{P} \cup \set{\roo}$, $\ES=\emptyset$, $G=\grap{V}{A}$, $V=\{p(\roo)\}$, $A=\emptyset$, $\ct(\roo)=\{p\}$, and  $\st(\roo,p)=\unexpa$.
\end{definition}

In the following, we show how to expand such an initial completion structure to prove satisfiability of a unary predicate $p$ w.r.t. a FoLP $P$, how to determine when no more expansion is needed, that is, either the structure represents a full open answer set or a clash has occurred, and under what circumstances a \emph{clash} occurs. In particular, \emph{expansion rules} evolve a completion structure, starting with a guess for an initial completion structure for checking satisfiability of $p$ w.r.t. $P$, to a complete clash-free structure that corresponds to a finite representation of an open answer set in case $p$ is satisfiable w.r.t. $P$. \emph{Applicability rules} state the necessary conditions such that these expansion rules can be applied.

\subsection{Expansion Rules}\label{sec:expansionrules}

Expansion rules update the completion structure by making explicit constraints which are necessary to hold for a certain literal to be part of a forest model\footnote{A negative literal `is part' of a forest model when the corresponding atom does not make part of the model.}.

An atom is part of a forest model if there is a ground rule which has the atom as head and all body literals are also part of the forest model; this is taken care of by the \textit{expand unary/binary positive} rules. New domain elements might have to be introduced by these rules in order to obtain such a ground rule.

Conversely, an atom is not part of the forest model if all bodies of ground rules which have as head the atom are not satisfied by the forest model. The rules which enforce this are the \textit{expand unary/binary negative} rules. The absence of an atom in the forest model is proved only when there is no possibility to introduce new individuals in the domain which would lead to a ground rule having the atom in the head and a satisfiable body. As such, there is an interaction between these rules and the rules which justify the presence of atoms in the open answer set.

Newly introduced domain elements give rise to new ground atoms and rules and some of these rules might render the program inconsistent. In order to be sure that the partially constructed model is a complete one every ground atom has to be proved to be either part or not part of the forest model. If the atom is not constrained to be or not to be part of the forest model, a random choice is made. The \textit{choose unary/binary} rules take care of this.

The expansion rules make extensive use of a sequence of operations meant to enforce the presence of a literal $\pm p(z)$ in the forest model (where $z$ is a term in case $p \in \upreds{P}$, and a pair of terms in case $p \in \bpreds{P}$) as part of justifying the presence of another literal $l$. This consists in inserting $\pm p$ in the content of $z$ and mark it as unexpanded, in case the predicate symbol is not already there, and in case $\pm p(z)$ is an atom, ensuring that it is a node in $G$ and if $l$ is also an atom, creating a new arc from $l$ to $\pm p(z)$ to capture the dependencies between the two elements of the forest model. Formally:

\begin{itemize}
    \item let $\ct(z):=\ct(z) \cup \{\pm p\}$ and $\st(z,\pm p):=\unexpa$,
    \item if $\pm p=p$, then let $V:=V \cup \{\pm p(z)\}$,
    \item if $l \in \atoms{\ground{P}{N_{\EF}}}$ and $\pm p=p$, then let $A:=A \cup \{(l, \pm p(z))\}$.
\end{itemize}

As a shorthand, we denote this sequence of operations as $\update(l, \pm p, z)$; more general, $\update(l, \beta, z)$ for a set of (possibly negated) predicates $\beta$, denotes $\forall \pm a \in \beta,$ $\update(l,\pm a,z)$.
In the following, for a completion structure $\langle \EF,$ $\ct,$ $\st,$ $G \rangle$, let $x \in N_{\EF}$ and $(x,y)\in A_{\EF}$ be the node, respectively arc, under consideration.

\subsubsection{(i) Expand unary positive.} Consider a unary positive predicate $p \in \ct(x)$ such that $\st(x,p)=\unexpa$. If $p$ is not a free predicate symbol:

\begin{itemize}
\item pick a rule $r\in P_p$ of the form (\ref{eq:unary}) such that $s$ (the term in the head of the rule) matches $x$. The rule will be used to justify the presence of $p(x)$ in the tentative open answer set.
\item for the $\beta$ in the body of $r$, $\update(p(x),\beta,x)$,
\item consider $k$ successors for $x$: $(y_m)_{1 \leq m \leq k}$, (by picking from the existing successors and/or by introducing new ones), such that:
\begin{itemize}
 \item for every $1 \leq (i,j) \leq k$ such that $t_i \neq t_j \in \psi$: $y_i \neq y_j$;
 \item for every $1 \leq m \leq k$:
 \begin{itemize}
 \item $y_m \in \succe_{\EF}(x)$, or
 \item $y_m$ is defined as a new successor of $x$ in the tree $T_c$, where $x \in T_c$: $y_m:=x \cdot n$, where $n \in \N^{*}$ s.t. $x \cdot n \notin \succe_{\EF}(x)$, and $T_c:=T_c \cup \{y_m\}$, or
 \item $y_m$ is defined as a new successor of $x$ in $\EF$ in the form of a constant: $y_m:=a$, where $a$ is a constant from $\cts{P}$ s.t. $a \notin \succe_{\EF}(x)$. In this case also add $(x,a)$ to $\ES$: $\ES:=\ES \cup \{(x,a)\}$.
 \end{itemize}
  \end{itemize}
 \item for every $1 \leq m \leq k$: $\update(p(x), \gamma_m, (x,y_m))$ and $\update(p(x), \delta_m, y_m)$.
\item set $\st(x,p):= \expa$.
\end{itemize}

If $p$ is free, its status in the content of $x$ is simply updated to expanded: $\st(x,p)=\expa$, as the presence of $p(x)$ in the forest model in construction is trivially justified by the free rule which defines $p$ grounded with $x$.

\subsubsection{(ii) Choose a unary predicate.}

If there is a $p \in \upreds{P}$ such that $p \notin \ct(x)$ and $\naf p \notin \ct(x)$, and for all $q \in \ct(x)$, $\st(x,q)=\expa$, and for all $(x,y) \in A_{\EF}$ and $\pm f \in \ct(x,y)$ (both positive and negative predicates) $\st((x,y),\pm f)=\expa$ then do one of the following:

\begin{itemize}
\item add $p$ to $\ct(x)$ and let $\st(x,p):=$ $\unexpa$, or
\item add $\naf{p}$ to $\ct(x)$ and let $\st(x,\naf{p})$ $ = \unexpa$.
\end{itemize}

In other words, if there are still unary predicates which do not appear in $\ct(x)$ (either in a positive or a negated form) and all positive predicates in the content of $x$ have been justified, as well as all positive or negative predicates in the content of one of the arcs starting in $x$ have been justified, one has to non-deterministically pick such a unary predicate symbol $p$ and inject either $p$ or $\naf p$ in $\ct(x)$.

As mentioned in the introduction to this section, this rule is needed in order to ensure that the partially constructed forest model is part of an actual model: as a result of introducing new domain elements in the process of constructing a forest model, there might be ground rules whose heads are not relevant per se for the satisfiability task at hand, but which are not satisfiable in any total extension of the partial forest model. One tries to effectively construct such an extension by making a random choice for unconstrained ground atoms regarding their membership to the forest model. As an analogy to the DL world, tableau algorithms which check concept satisfiability typically internalize the TBox, i.e. reduce reasoning w.r.t. a terminology to checking satisfiability of a new concept \cite{horrocks99practical}. This new concept is constructed by taking into account all axioms in the TBox and not only those on which the initial concept checked to be satisfiable depends.

As an example consider the program with only two rules: $\prule{a(X) \lor \naf a(X)}{}$ and $\prule{b(X)}{\naf b(X)}$. Suppose one wants to check whether $a$ is satisfiable: while it is trivial to see that $a$ is justified by the first rule, the program has no open answer set due to the inconsistency introduced by the second rule. This will be tracked down by our algorithm by trying to prove $b(\roo)$ and $\naf{b(\roo)}$ (after each of them is inserted in the content of $\roo$ as a result of applying the choose unary rule), and failing in each case.

For reasons described in the next subsection, this rule has priority over the rule which describes the expansion of unary negative predicates.

\subsubsection{(iii) Expand unary negative.}
In general, for justifying that a negative unary literal $\naf p \in \ct(x)$ (or in other words, the absence of $p(x)$ in the constructed forest model), one has to refute the body of every ground (non-free) rule with head atom $p(x)$.  Let $r \in  P_p$ and $r': \prule{p(x)}{\beta(x),}$ $(\gamma_m(x,y_m),$ $\delta_m(y_m))_{1\leq m \leq k},$ $\psi$, with $\psi\subseteq \bigcup_{1\leq i\neq j\leq k}\set{y_i\neq y_j}$, and $k \in \N$, be a ground version of $r$. The body of $r'$ can be either:
\begin{itemize}
\item (i) `locally' refuted: by refutation of a literal from $\beta(x)$. For this, one has to enforce that there is a $\pm q \in \beta$ which does not appear in $\ct(x)$, or in other words: $\mp q \in \ct(x)$; note that this refutes all ground versions of $r$ where the head variable is substituted with $x$.
\item (ii) refuted in the `successor' part of the rule: by refutation of a literal from one of $(\gamma_m(x,y_m))_{1\leq m \leq k}$ or $(\delta_m(y_m)))_{1\leq m \leq k}$, or by impossibility to satisfy $\psi$. In a forest model, all groundings of $r$, in which one of the successor terms has been substituted with $y$, where $y$ is a node in the forest which is not a direct successor of $x$, are refuted: there is no arc which links $x$ to $y$, and as such there are no literals of the form $f(x, y)$ with $f \in bpreds (P)$ in the constructed open answer set. Thus, one has to consider only groundings in which $(y_m)_{1\leq m \leq k}$ are successors of $x$ in $EF$: $(y_m=x \cdot z_m)_{1\leq m \leq k}$, and which satisfy $\psi$. For such ground rules, the body can be refuted by enforcing that there is a  $\pm f \in \delta_m$  which does not appear in $\ct(x,x \cdot z_m)$ (equivalent with: $\mp f \in \ct(x,x \cdot z_m)$) or that there is a  $\pm q \in \gamma_m$ which does not appear in $\ct(x \cdot z_m)$ (equivalent with: $\mp q \in \ct(x \cdot z_m)$), for some $1\leq m \leq k$.
\end{itemize}

As we want to refute the bodies of all ground versions of $r$, we either apply (i) once, or apply (ii) for every assignment of successor terms in $r$ with successors of $x$ in $EF$ which satisfies $\psi$. As $\psi$ imposes a minimum bound on the number of distinct successor terms, if the number of successors of $x$ in $EF$ is smaller than this bound, there is no such assignment which satisfies $\psi$. In this case, all bodies of ground versions of $r$ are refuted.

Formally, for a unary negative predicate $\naf p \in \ct(x)$ for which $\st(x,\naf p)=\unexpa$, and for every rule $r \in P_p$ of the form ($\ref{eq:unary})$ such that $x$ matches $s$ ($s$ is the term from the head of the rule), given that $y_1, \ldots, y_n$ are the successors of $x$ in $EF$, do one of the following:
\begin{itemize}
  \item  pick a $\pm q \in \beta$ and $update(\naf p(x),\mp q,x)$, or
  \item for all $y_{i_1}, \ldots, y_{i_k}$ s. t. $(1 \leq i_j \leq n)_{1 \leq j \leq k}$: if for all $1 \leq j,l \leq k$, $t_j \neq t_l \in \psi \Rightarrow y_{i_j} \neq y_{i_l}$, do one of the following:
\begin{itemize}
\item for some $m$, $1 \leq m \leq k$, pick $\pm f \in \delta_m$ and $\update(\naf p(x), \mp f, (x, y_{i_m}))$, or
\item for some $m$, $1 \leq m \leq k$, pick $\pm q \in \gamma_m$ and $\update(\naf p(x), \mp q, y_{i_m})$.
\end{itemize}
\end{itemize}
Finally, set $\st(x, \naf p):=\expa$.

Note that the introduction of new successors of $x$ gives rise to new ground unary rules with head $p(x)$. Such successors are introduced in the process of expanding positive unary predicates. In order to ensure that $p(x)$ is indeed refuted, this rule should be applied only when all successors of $x$ have been introduced, i.e., when there is no possibility to further expand a positive unary predicate:
\begin{itemize}
\item for all $p \in \upreds{P}$, $p \in \ct(x)$ or $\naf p \in \ct(x)$, and
\item for all $p \in \ct(x)$, $\st(p,x):=\expa$
\end{itemize}
In other words, the rule is applied when neither of the expansion rules \emph{(i) Expand unary positive} or \emph{(ii) Choose unary} can be further applied w.r.t. a certain node $x$: in this case there is and there will be no unexpanded positive predicate in the content of $x$.

\subsubsection{(iv) Expand binary positive.} Consider a binary positive predicate symbol $f \in \ct(x,y)$ such that $\st((x,y),f)=\unexpa$. 
If $f$ is not free, pick a rule $r\in P_f$ of the form (\ref{eq:binary}) such that $x$ matches  $s$ and $y$ matches with $t$ ($s$ and $t$ are the terms from the head of the rule) to justify $f$. For $\beta$, $\gamma$, and $\delta$ corresponding to $r$ do:
$\update(p(x,y),\beta,x)$, $\update(p(x,y),\gamma,(x,y))$, and $\update(p(x,y),\delta,y)$. Finally, let $\st((x,y),f):=\expa$ (this is applied also when $f$ is free).


\subsubsection{(v) Expand binary negative.} For a binary negative predicate symbol $\naf f \in \ct(x,y)$ such that $\st((x, y), \naf f)=unexp$, 
and for every rule $r \in P_f$ of the form (\ref{eq:binary}) such that $x$ matches $s$ and $y$ matches $t$ ($s$ and $t$ are the terms from the head of the rule) do one of the following:
\begin{itemize}
\item pick a $\pm p$ from $\beta$ and
$\update(\naf f(x,y),\mp p,x)$, or
\item pick a $\pm g$ from $\gamma$ and
$\update(\naf f(x,y),\mp g,(x,y))$, or
\item pick a $\pm q$ from $\delta$ and
$\update(\naf f(x,y), \mp q,y))$.
\end{itemize} Finally, let $\st((x,y),\naf{f}):=\expa$.
Note that the expand binary negative rule, unlike its unary counterpart, does not have to consider all successors of $x$, just $y$. As such, there are no complex interactions between this rule and the expand binary positive one.


\subsubsection{(vi) Choose a binary predicate.} If  no (possibly negated) unary predicate $\pm a \in \ct(x)$ can be expanded according to expansion rules (i)-(iii), and for all $(x,y) \in A_{\EF}$ none of $\pm f \in \ct(x,y)$ can be expanded according to rules (iv) and (v), and for some ${f \in \bpreds{P}}$: $f \notin \ct(x,y)$ and $\naf f \notin \ct(x,y)$, then do one of the following:
\begin{itemize}
\item add $f$ to $\ct(x,y)$ and let $\st((x,y),p):=\unexpa$, or
\item add $\naf f$ to $\ct(x,y)$ and let $\st((x,y),\naf{p}) := \unexpa$.
\end{itemize}

\subsection{Applicability Rules}\label{sec:appl}

A second set of rules is not updating the completion structure under
consideration, but restricts the use of the expansion rules. We refer to
these rules as so-called applicability rules.

\subsubsection{(vii) Saturation} We call a node $x \in N_{\EF}$
\emph{saturated} if
\begin{itemize}
\item for all ${p \in \upreds{P}}$ we have $p \in \ct(x)$ or $\naf p \in \ct(x)$ and none of $\pm q \in \ct(x)$ can be expanded according to the rules (i)-(iii) ,
\item for all $(x,y) \in A_{T^{\EF}}$, $T \in \EF$ and $p \in \bpreds{P}$, $p \in \ct(x,y)$ or $\naf p \in \ct(x,y)$ and none of $\pm f \in \ct(x,y)$ can be expanded according to the rules (iv)-(vi).
\end{itemize}

We impose that no expansions can be performed on a node from $N_{\EF}$ which does not belong to $\cts{P}$ until its predecessors are saturated (we exclude constants as they can have more then one predecessor in the completion, including themselves).

\subsubsection{(viii) Blocking}
A node $x \in N_\EF$ is \emph{blocked} if there is an ancestor $y$ of $x$ in $F$, $y<_F x$, $y \not \in \cts{P}$, s.t. $\ct(x) \subseteq \ct(y)$ and the set $connpr_G(y,x)=\{(p,q) \mid (p(y), q(x)) \in conn_G \wedge q \mbox{ is not free}\}$ is empty. We call $(y,x)$ a \emph{blocking pair}. 
No expansions can be performed on a blocked node. Intuitively, if there is an ancestor $y$ of $x$ which is not a constant, whose content includes the content of $x$, one can extend the interpretation such that the contents of $x$ and its outgoing arcs are identical to the contents of $y$ and its outgoing arcs. The newly introduced atoms which have $x$ as an argument will be justified in a similar way as their counterpart atoms which have $y$ as an argument. One can either:
\begin{enumerate}
\item reuse the successors of $y$ as successors of $x$: this consists in the introduction of `backward' arcs in the extended forest from the leaf node $x$ to the said successors. The contents of these backward arcs will replicate the content of their counterpart arcs from $y$ to its successors. The interpretation thus obtained is no longer a forest shaped one. This is the approach we consider for proving the soundness of the algorithm and it is exemplified in Section \ref{subsec:proofdisc}.
\item introduce new successors for $x$ which are similar to the successors of $y$ and which at their turn will be justified similarly to the successors of $y$, and so on. In this case, one obtains an infinite forest interpretation. This approach is exemplified at the end of Section \ref{subsec:illustration}.
\end{enumerate}
However, in order for the interpretation constructed in one of the above ways to be a forest model, it is necessary that no atom in the interpretation is circularly or infinitely justified: a sufficient condition to enforce this is to impose that there are no paths in $G$ from a positive literal $p(y)$ to another positive literal $q(x)$. For more insight into this please check Section \ref{subsec:proofdisc} and the complete soundness proof in the appendix.

\subsubsection{(ix) Redundancy}
A node $x \in N_\EF$ is \emph{redundant} if it is saturated, it is not blocked, and there are $k$ ancestors of $x$ in $F$, $(y_i)_{1 \leq i \leq k}$, where $k=2^p(2^{p^2}-1)+2$, and $p=|\upreds{P}|$, such that $\ct(x)=\ct(y_i)$. In other words, a node is redundant if there are other $k$ nodes on the same branch with the current node which all have content equal to the content of the current node. The presence of a redundant node stops the expansion process. 

In the completeness proof we show that any forest model of a FoLP $P$ which satisfies $p$ can be reduced to another forest model which satisfies $p$ and has at most $k+1$ nodes with equal content in any branch of a tree from the forest model, and furthermore the $(k+1)$st node, in case it exists, is blocked\footnote{The reduction consists in collapsing parts of the forest by replacing a subtree with root $c$ with another subtree with root $d$, where $\ct(c)=\ct(d)$, and $d$ is a (non-constant) successor of $c$ in the forest. However, this reduction can be applied only when certain conditions are met, e.g. there are no blocking nodes on the path between $c$ and $d$. As such, the value of $k$ depends on the number of possible contents for nodes, $2^p$, but it is greater than that, due to the fact that the reduction can be applied only in certain situations.}.  One can thus search for forest models only of the latter type. This rule exploits that result: we discard models which are not in this shrunk search space. For more intuition regarding the reduction of a forest model to a forest model with at most $k+1$ nodes with equal content in any branch of a tree from the forest model, we refer the reader to the completeness proof in the appendix.

\subsection{Clash-Free Complete Completion Structures}

We call a completion structure \emph{contradictory} if for some $x \in N_ {\EF}$ and $a \in \upreds{P}$, $\{a, \naf a\} \subseteq \ct(x)$ or for some $(x,y) \in A_{\EF}$ and $f \in \bpreds{P}$, $\{f, \naf f\} \subseteq \ct(x,y)$.
A \emph{complete completion structure} for a FoLP $P$ and a $p\in\upreds{P}$ is a completion structure that results from applying the expansion rules to an initial completion structure for $p$ and $P$, taking into account the applicability rules, such that no expansion rules can be further applied. Furthermore, a complete completion structure  $\CS = \langle \EF,$ $\ct,$ $\st,$ $G \rangle$ is \emph{clash-free} if:
\begin{itemize}
 \item (1) \CS{} is not contradictory,
 \item (2) $\EF$ does not contain redundant nodes, and
 \item (2) $G$ does not contain positive cycles.
\end{itemize}

Next section will give an example for constructing a clash-free complete completion structure, while section \ref{subsec:proofdisc} will show that a predicate $p$ is satisfiable w.r.t. a FoLP $P$ iff there exists a clash-free complete completion structure of $p$ w.r.t. $P$.

\subsection{Illustration of the algorithm}
\label{subsec:illustration}
Consider a slightly modified version of the FoLP program described in Section \ref{sec:intro}, in which the constraints have been replaced by unary rules as described in Section \ref{sec:FLP}, and the last rule has been removed. We will refer to this program as $P$.

\begin{programm}
\nsrule{r_1}{happy(X)}{sees(X,Y),friend(X,Y),happy(Y)}
\nsrule{r_2}{happy(X)}{sees(X,Y), enemy(X,Y), unhappy(Y)}
\nsrule{r_3}{unhappy(X)}{sees(X,Y), friend(X,Y), \naf{happy(Y)}}
\nsrule{r_4}{unhappy(X)}{sees(X,Y), enemy(X,Y), happy(Y)}
\nsrule{r_5}{happy(X)}{friend(X,Y), friend(X,Z), Y \neq Z}
\nsrule{r_6}{sees(X,Y)\lor \naf{sees(X,Y)}}{}
\nsrule{r_7}{friend(X,Y)\lor \naf{friend(X,Y)}}{}
\nsrule{r_8}{enemy(X,Y)\lor \naf{enemy(X,Y)}}{}
\nsrule{r_9}{c(X)}{\naf c(X), happy(X), unhappy(X)}
\nsrule{r_{10}}{d(X,Y)}{\naf d(X,Y), friend(X,Y), enemy(X,Y)}
\nsrule{r_{11}}{\ unhappy(j)}{hungry(j)}
\end{programm}

We want to check the satisfiability of the predicate $happy$ w.r.t. $P$. For this purpose, we first define an initial completion structure for $happy$ w.r.t. $P$: $\langle EF,$ $\ct,$ $\st,$ $G\rangle$. There is one constant in $P$, $j$, so there will be a tree with root $j$, $T_j$, in $EF$; further, we choose not to include a tree with anonymous root in $EF$, and thus the only choice for placing the initial constraint $happy$ is the content of node $j$. The initial status of $happy$ in this node is unexpanded, so the status function is updated accordingly. The graph $G=(V,A)$ which keeps track of dependencies between atoms in the model in construction is initialized such that $V=\{happy(j)\}$, and $A=\emptyset$. The picture below depicts the initial completion structure for $happy$ w.r.t. $P$. Note that the fact that the status of $happy$ is unexpanded is marked by appending the superscript $u$ to $happy$.

\begin{center}
\begin{tikzpicture}[ auto]
    \node (x) {$j$};
\node[right, xshift=-0.2cm] at (x.east) {$\{happy^{u}\}$};
\end{tikzpicture}
\end{center}

According to the expansion rule \emph{(i) Expand unary positive}, the presence of the unexpanded predicate $happy$ in the content of a node $j$, or in other words of $happy(j)$ in the corresponding tentative open answer set, has to be justified by means of a unary rule with head predicate $happy$ and head term which matches $j$. We apply the expansion rule using the unary rule: $\nprule{r_1}{happy(X)}{sees(X,Y),friend(X,Y),happy(Y)}$: a new successor $j1$ is created for $j$ in $T_j$ and the predicates $sees$ and $friend$ are inserted in the content of the arc $(j, j1)$, and the predicate $happy$ is inserted in the content of $j1$. $G$ is also updated by addition of the nodes $happy(j1)$, $sees(j,j1)$, and $friend(j,j1)$ to $V$, and of the arcs $(happy(j), sees(j,  j1))$,  $(happy(j), friend(j,  j1))$, and $(happy(j), happy(j1))$ to $A$. In other words, $happy(j)$ is in the model in construction if there is an individual $j1$ such that $sees(j,j1)$, $friend(j, j1)$, and $happy(j1)$ are all present in the same open answer set. Next figure depicts the situation after the application of the expansion rule. The predicate $happy$ in the content of $j1$ is marked as unexpanded. The other predicates are either expanded ($happy$ in the content of $j$) or free predicates ($sees$ and $friend$ in the content of $(j, j1)$), and as such they are not superscripted. 

\begin{center}
\begin{tikzpicture}[ auto]
    \node (x) {$j$}
	       child{
			node (y) {$j1$}
			edge from parent[->]
			node[above,near end,xshift=1.2 cm] {$\{sees,friend\}$};
		};
\node[right, xshift=-0.2cm] at (x.east) {$\{happy\}$};
\node[right, xshift=-0.2cm] at (y.east) {$\{happy^{u}\}$};
\end{tikzpicture}
\end{center}

Once again the only unexpanded predicate is $happy$, only this time in the content of $j1$. However, we cannot proceed to its expansion since $j$ is not saturated: there are predicates which do not appear either in a positive or a negative form in the contents of $j$ and its outgoing arcs. Remember that according to applicability rule \textit{(vii) Saturation} no expansions can be performed on a node which is not a constant until its predecessor is saturated. We pick the predicate $hungry$ and apply the expansion rule \emph{(ii) Choose unary} by inserting $\naf hungry$ in the content of $j$. It is not possible to apply \emph{(iii) Expand unary negative} w.r.t. $\naf hungry$ in the content of $j$, as one can still apply the \emph{(ii) Choose unary} rule: as such we pick the predicate $c$ and choose to insert $\naf c$ in the content of $j$\footnote{Note that $c$ (which is used to simulate a constraint) does not appear in the head or body or any other rule than $r_9$ and is never satisfiable: as such, an application of \emph{(ii) Choose unary} rule w.r.t. $c$ is needed for saturating the content of every node, and for simplification of exposition we will always choose to insert $\naf c$ in the content of the node (as the other choice would lead to a contradiction). The same reasoning applies to $d$: for every arc, there has to be an application of the \emph{(vi) Choose binary} rule w.r.t. $d$ and the choice in each case will be to insert $\naf d$ in the content of the arc.}. Once again, $j$ is not saturated and \emph{(ii) Choose unary} can be applied w.r.t. $unhappy$: we choose to insert $unhappy$ in the content of $j$:



\begin{center}
\begin{tikzpicture}[ auto]
    \node (x) {$j$}
	       child{
			node (y) {$j1$}
			edge from parent[->]
			node[above,near end,xshift=1.2cm] {$\{sees,friend\}$}
		};
\node[right, xshift=-0.2cm] at (x.east) {$\{happy, \naf hungry^{u},  \naf c^{u}, unhappy^{u}\}$};
\node[right, xshift=-0.2cm] at (y.east) {$\{happy^{u}\}$};
\end{tikzpicture}
\end{center}

Among the unexpanded predicates in the content of $j$ we pick $unhappy$ as the next candidate for expansion as \textit{(i) Expand unary positive} has priority over \textit{(iii) Expand unary negative}. A rule with head predicate $unhappy$ and head term which matches $j$ is picked to justify the presence of $unhappy(j)$ in the model in construction: $\nprule{r_3}{unhappy(X)}{sees(X,Y), friend(X,Y), \naf{happy(Y)}}$.  Either the successor of $j$, $j1$, is reused or a new one is introduced to satisfy the non-local part of the rule. Suppose we pick up the already existing successor, $j1$. In this case $sees$ and $friend$ are inserted into the content of the arc $(j, j1)$ (they are already there), while $\naf happy$ is inserted into the content of $j1$: this leads to a contradiction as now both $\naf happy$ and $happy$ are in the content of $j1$. 


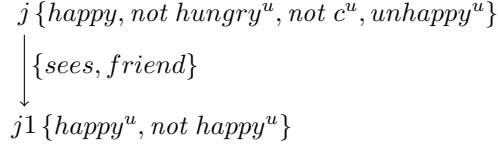
\begin{center}
\begin{tikzpicture}[ auto]
    \node (x) {$j$}
	       child{
			node (y) {$j1$}
			edge from parent[->]
			node[above,near end,xshift=1.2cm] {$\{sees,friend\}$}
		};
\node[right, xshift=-0.2cm] at (x.east) {$\{happy, \naf hungry^{u}, \naf c^{u}, unhappy^{u}\}$};
\node[right, xshift=-0.2cm] at (y.east) {$\{happy^{u}, \naf happy^{u}\}$};
\end{tikzpicture}
\end{center}

The algorithm backtracks and introduces a new successor for $j$, $j2$: $sees$ and $friend$ are inserted into the content of the arc $(j, j2)$, and $\naf happy$ is inserted in the content of $j2$. Now $unhappy$ in the content of $j$ can be marked as expanded, and we proceed further with the expansion process. Suppose we pick $\naf c$ for expansion. There is a single ground rule which defines $c(j)$: $\prule{c(j)}{\naf c(j), happy(j), unhappy(j)}$. According to the expansion rule \emph{(iii) Expand unary negative}, the body of this rule has to be refuted. There are three possible choices for doing this: inserting $c$, $\naf happy$, or $\naf unhappy$ into the content of $j$. Each of the three choices leads to a contradiction. The figure below depicts the case when $\naf unhappy$ was chosen to refute the body of the rule.


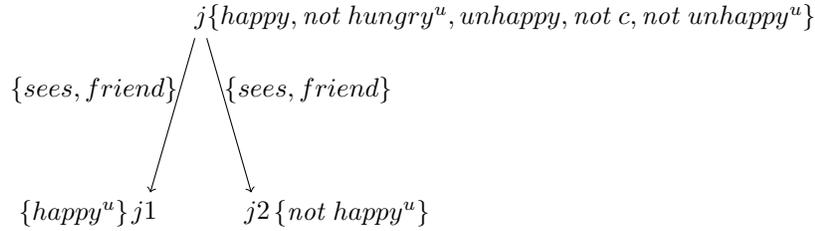
\begin{center}
\begin{tikzpicture}[level distance=2.6cm, auto]
    \node (x) {$j$}
	       child{
			node (y) {$j1$}
			edge from parent[->]
			node[above,near end,yshift=0.55cm, xshift=-0.9cm] {$\{sees,friend\}$}
		}
        child{
			node (z) {$j2$}
			edge from parent[->]
			node[above,near end,,yshift=0.55cm, xshift=0.9cm] {$\{sees,friend\}$}
		};
\node[right, xshift=-0.25cm] at (x.east) {$\{happy, \naf hungry^u, unhappy, \naf c, \naf unhappy^{u}\}$};
\node[right, xshift=-1.5cm] at (y.west) {$\{happy^{u}\}$};
\node[right, xshift=-0.2cm] at (z.east) {$\{\naf happy^{u}\}$};
\end{tikzpicture}
\end{center}
\normalsize

The algorithm backtracks to the previous choice, which was the choice of the rule to justify $unhappy$ in the content of $j$. There are still two more rules in $P$ whose head matches $unhappy(j)$: $r_4$ and $r_{11}$. However,  from the previous developments one can see that even if $unhappy$ is satisfied in some other way, one will eventually reach a contradiction due to the presence of $happy$, $unhappy$, and $\naf c$ in the content of $j$. As such, we skip the remaining two choices as concerns rules to justify $unhappy(j)$. Backtracking further, one has to retract $unhappy$ from the content of $j$, and insert $\naf unhappy$ instead, and mark it as unexpanded. 
Next step is to select $\naf unhappy$ for expansion. According to the expansion rule \emph{(iii) Expand unary negative}, every ground rule which defines $unhappy(j)$ has to be considered and its body to be refuted. There is one instantiation for each rule whose head matches $unhappy(j)$: 
\begin{itemize}
\item $r_3$: $\prule{unhappy(j)}{sees(j, j1), friend(j, j1), \naf happy(j1)}$. The body of this rule has to be refuted: $sees(j, j1)$ and $friend(j, j1)$ are already part of the tentative open answer set so they cannot be refuted. The only remaining choice is to refute $\naf happy(j1)$, thus to insert $happy$ into the content of $j1$.
\item $r_4$: $\prule{unhappy(j)}{sees(j, j1), enemy(j, j1), happy(j1)}$. Here the only choice which does not lead to contradiction is asserting $\naf enemy$ to the content of $j1$. The predicate $enemy$ is a free predicate, defined only by a free rule, so it is trivially expanded.
\item $r_{11}$: $\prule{unhappy(j)}{hungry(j)}$. The body of this rule is refuted by the presence of $\naf hungry$ into the content of $j$.
\end{itemize}



Finally, in order to saturate $j$, we apply the \emph{(vi) Choose binary} rule and insert $\naf d$ in the content of $(j, j1)$.  Then, $\naf d$ is expanded using \emph{(vi) Expand binary negative}: we observe that the body of the ground rule  $\prule{d(j, j1)}{\naf d(j, j1), friend(j,j1), enemy(j,j1)}$ derived from $r_{10}$ is already refuted by the presence of $\naf enemy$ in the content of $(j, j1)$.

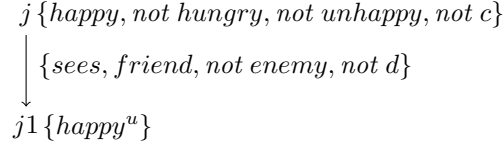
\begin{center}
\begin{tikzpicture}[ auto]
    \node (x) {$j$}
	       child{
			node (y) {$j1$}
			edge from parent[->]
			node[above,near end,xshift=2.65cm] {$\{sees,friend, \naf enemy, \naf d\}$}
		};
\node[right, xshift=-0.2cm] at (x.east) {$\{happy, \naf hungry, \naf unhappy, \naf c\}$};
\node[right, xshift=-0.2cm] at (y.east) {$\{happy^{u}\}$};
\end{tikzpicture}
\end{center}

At this moment, $j$ is saturated and by means of applicability rule \emph{(vii) Saturation} we can proceed to its successor $j1$. One can see that the content of $j1$ is included in the content of $j$, so according to rule \emph{(viii) Blocking}, $(j, j1)$ is a candidate blocking pair. However $G$ contains the arc $(happy(j), happy(j1))$, so $connpr_{G}(j,j1) \neq \emptyset$, and the second condition of the blocking rule is not met. Intuitively, if one would justify $j1$ in a similar manner as $j$, an infinite chain of the type $happy(j), happy(j1), \ldots$ would be present in the model in construction, each atom in the set being justified by the next one in the set, thus there would be atoms in the model which are not finitely justified. Thus, $j1$ cannot be blocked and we proceed to expanding its content. This time we pick rule $\nprule{r_5}{happy(X)}{friend(X,Y), friend(X,Z), Y \neq Z}$ to justify the presence of $happy(j1)$ in the tentative open answer set. To satisfy the body of some grounded version of the rule, two distinct successors of $j1$, $j11$ and $j12$, are created, and $friend$ is asserted to the content of both $(j1, j11)$ and $(j1, j12)$. The arcs $(happy(j1), friend(j1, j11))$ and $(happy(j1), friend(j1, j12))$ are added to $A$ in $G$ to capture the new dependencies between atoms in the tentative open answer set.
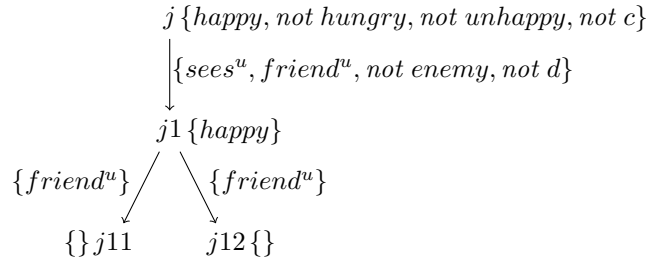
\begin{center}
\begin{tikzpicture}[ auto]
    \node (x) {$j$}
	       child{
			node (y) {$j1$}
			edge from parent[->]
			child{
				node (z) {$j11$}
				edge from parent[->]
				node[above,near end,xshift=1.8cm] {$\{friend^{u}\}$};
			}
			child{
				node (t) {$j12$}
				edge from parent[->]
				node[above,near end,xshift=-1.8cm] {$\{friend^{u}\}$};
			}
			node[above,near end,xshift=2.7cm] {$\{sees^{u},friend^{u}, \naf enemy, \naf d\}$};
		};
\node[right, xshift=-0.2cm] at (x.east) {$\{happy, \naf hungry, \naf unhappy, \naf c\}$};
\node[right, xshift=-0.2cm] at (y.east) {$\{happy\}$};
\node[right, xshift=-0.2cm] at (t.east) {$\{\}$};
\node[right, xshift=-0.4cm] at (z.west) {$\{\}$};
\end{tikzpicture}
\end{center}
Now we proceed to saturate $j1$ by choosing to add $\naf c$, $\naf hungry$, and $\naf unhappy$ to the content of $j1$ by repeatedly applying the expansion rule \emph{(vi) Choose unary negative}. The first two additions are expanded in a similar manner as their counterparts in the content of $j$. As concerns $\naf unhappy$, we have to consider again all three rules which define the predicate $unhappy$. The justification w.r.t. $r_{11}$ is similar as above, as the rule is a local rule. There are two successors of $j1$, $j11$ and $j12$, so there are two ground versions of $r_3$: $\prule{unhappy(j1)}{sees(j1,j11), friend(j1,j11), \naf{happy(j11)}}$, and $\prule{unhappy(j1)}{sees(j1,j12), friend(j1,j12), \naf{happy(j12)}}$, and two ground versions of rule $r_4$: $unhappy(j1)$ $\gets$ $sees(j1,j11), enemy(j1,j11), happy(j11)$, and $unhappy(j1)$ $\gets$ $sees(j1,j12), enemy(j1,j12), happy(j12)$. The bodies of all these four ground rules have to be refuted. This is achieved by asserting $happy$ to the content of $j11$, $\naf sees$ to the content of $(j1, j12)$, and $\naf enemy$ to both the contents of $(j1, j11)$ and $(j1, j12)$. Finally, we saturate $j1$ by completing the contents of the arcs $(j1, j11)$ and $(j1, j12)$ in a similar manner as for the arc $(j, j1)$.
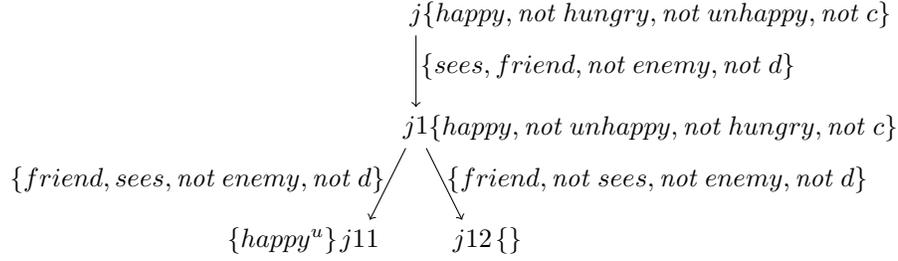
\begin{center}
\begin{tikzpicture}[ auto]
    \node (x) {$j$}
	       child{
			node (y) {$j1$}
			edge from parent[->]
			child{
				node (z) {$j11$}
				edge from parent[->]
				node[above,near end,xshift=3.7cm] {$\{friend, \naf sees, \naf enemy, \naf d\}$};
			}
			child{
				node (t) {$j12$}
				edge from parent[->]
				node[above,near end,xshift=-3.4cm] {$\{friend, sees, \naf enemy, \naf d\}$};
			}
			node[above,near end,xshift=2.55cm] {$\{sees, friend, \naf enemy, \naf d\}$};
		};
\node[right, xshift=-0.25cm] at (x.east) {$\{happy, \naf hungry, \naf unhappy, \naf c\}$};
\node[right, xshift=-0.25cm] at (y.east) {$\{happy, \naf unhappy, \naf hungry, \naf c\}$};
\node[right, xshift=-0.2cm] at (t.east) {$\{\}$};
\node[right, xshift=-1.5cm] at (z.west) {$\{happy^{u}\}$};
\end{tikzpicture}
\end{center}

At this moment, $j1$ is also saturated and we observe that the contents of both its successors are included in its own content. Unlike the case where $\ct(j1)\subset \ct(j)$, but $connpr_G(j, j1) \neq \emptyset$, we have that both $connpr_G(j1, j11)=\emptyset$, and $connpr_G(j1, j12)$ $= \emptyset$, thus both $(j1, j11)$ and $(j1, j12)$ are blocking pairs. Thus, the completion structure depicted in the figure above is a complete clash-free completion structure. We can derive a forest-shaped open answer set by unraveling the structure, as explained already in the context of rule \emph{(viii) Blocking}. The contents of $j11$ and $j12$ are made to be identical to the content of $j1$ and they are justified similarly as the content of $j1$. This will give rise to two new successors for both $j11$ and $j12$, which again will be justified in the same manner, etc. The obtained forest model is depicted in the figure below.

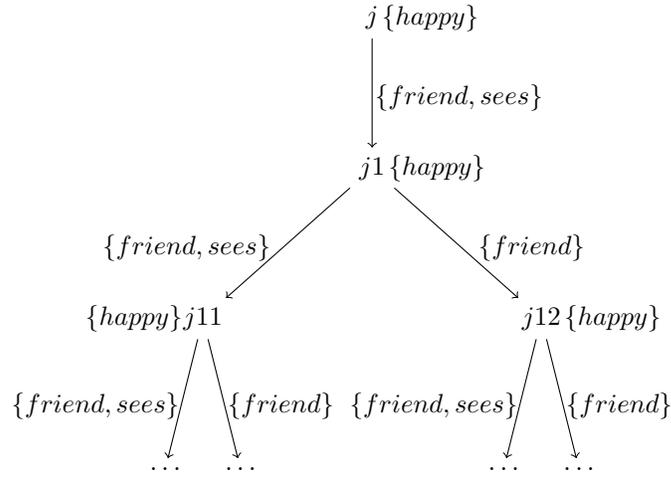
\begin{center}
\begin{tikzpicture}[level distance=20mm,level 1/.style={sibling distance=40mm},level 2/.style={sibling distance=45mm},level 3/.style={sibling distance=20mm}, level 3/.style={sibling distance=10mm}, auto]
    \node (x) {$j$}
	       child{
			node (y) {$j1$}
			edge from parent[->]
			child{
				node (z) {$j11$}
				edge from parent[->]
                child{
				    node (z1) {$\ldots$}
				    edge from parent[->]
node[above,near end,xshift=1.4cm] {$\{friend\}$};
			     }
			     child{
				    node (z2) {$\ldots$}
				    edge from parent[->]
                node[above,near end,xshift=-1.8cm] {$\{friend, sees\}$};
			     }
				node[above,near end,xshift=3.65cm] {$\{friend\}$};
			}
			child{
				node (t) {$j12$}
				edge from parent[->]
                child{
				    node (t1) {$\ldots$}
				    edge from parent[->]
				    node[above,near end,xshift=1.4cm] {$\{friend\}$};
			     }
			     child{
				    node (t2) {$\ldots$}
				    edge from parent[->]
				    node[above,near end,xshift=-1.8cm] {$\{friend, sees\}$};
			     }
				node[above,near end,xshift=-4.0cm] {$\{friend, sees\}$};
			}
			node[above,near end,xshift=1.15cm] {$\{friend, sees\}$};
		};
\node[right, xshift=-0.2cm] at (x.east) {$\{happy\}$};
\node[right, xshift=-0.2cm] at (y.east) {$\{happy\}$};
\node[right, xshift=-0.2cm] at (t.east) {$\{happy\}$};
\node[right, xshift=-1.3cm] at (z.west) {$\{happy\}$};
\end{tikzpicture}
\end{center}

Thus, $happy$ is satisfiable w.r.t. $P$. The open answer set which satisfies $happy$ is $(U, M)$, with $U=\{j, j1, j11, j12,$ $ j111,$ $ j112,$ $ \ldots\}$, and $M=\{happy(j)\} $ $\cup \{happy(js),$ $friend$ $(js,$ $ js1),$ $ friend (js, js2), sees(js, js1)$ $|s=1, $ $11, 12, 111, 112,$ $ \ldots \}$.

\subsection{Termination, Soundness, and Completeness}
\label{subsec:proofdisc}

We show that an initial completion structure for a unary predicate $p$ and a FoLP $P$ can always be expanded to a complete completion structure
(\emph{termination}), that, if there is a clash-free complete completion structure, $p$ is satisfiable w.r.t. $P$ (\emph{soundness}), and finally, that, if $p$ is satisfiable w.r.t. $P$, there is a clash-free complete completion structure (\emph{completeness}).

\begin{proposition}[termination]\label{prop:termination}
Let $P$ be a FoLP and $p \in \upreds{P}$.  Then, one can construct a finite complete
completion structure by a finite number of applications of the expansion
rules to the initial completion structure for $p$ w.r.t. $P$, taking into account
the applicability rules.
\end{proposition}

\begin{proof}[Proof sketch]
Assume one cannot construct a complete completion structure by a finite number of applications of the expansion rules, taking into account the applicability rules. Clearly, if one has a finite completion structure that is not complete, a finite application of expansion rules would complete it unless successors are introduced.  However, one cannot introduce infinitely many successors: every infinite path in the extended forest will eventually contain $|k+1|$ saturated nodes with  equal content, where $k$ is as in the redundancy rule, and thus either a blocked or a redundant node, which is not further expanded. Furthermore, the arity of the trees in the completion structure is bound by the number of successor variables in unary rules, more precisely by $rank(P)$, where $P$ is the FoLP under consideration.
\end{proof}

\begin{proposition}[soundness]\label{prop:soundness}
Let $P$ be a FoLP and $p\in \upreds{P}$. If there exists a complete clash-free completion structure for $p$ w.r.t. $P$, then $p$ is satisfiable w.r.t. $P$.
\end{proposition}

\newcommand{\blocked}{\ensuremath{\mathit{blocked}}}
\newcommand{\Texte}{\ensuremath{T_\mathit{ext}}}
\newcommand{\texte}{\ensuremath{t_\mathit{ext}}}
\newcommand{\Gexte}{\ensuremath{G_\mathit{ext}}}
\newcommand{\Vexte}{\ensuremath{V_\mathit{ext}}}
\newcommand{\Aexte}{\ensuremath{A_\mathit{ext}}}
\newcommand{\Eexte}{\ensuremath{E_\mathit{ext}}}
\newcommand{\nEF}{\ensuremath{\EF^{'}}}
\newcommand{\nAEF}{\ensuremath{A^{'}}}
\newcommand{\orig}[1]{\ensuremath{\overline{#1}}}

\begin{proof}[Proof sketch]
From a clash-free complete completion structure, one can construct an open interpretation and show that this interpretation is an open answer set of $P$ that satisfies $p$.
One way to construct such an open interpretation, by unraveling the completion structure to an infinite structure (an open answer set with an infinite universe and an infinite interpretation), has been exemplified in the previous section. However, for simplicity of the proof we chose a different approach: from a forest-shaped completion structure we generate a graph-shaped open answer set by extending the content of the blocked nodes to be identical to the content of the corresponding blocking nodes and introducing additional arcs from blocked nodes to successors of blocking nodes which mirror the arcs from the blocking nodes themselves to their successors (thus, also inheriting their content). Also, at this stage all negated predicates from the contents of nodes/arcs can be ignored. Considering our example from section \ref{subsec:illustration}, the complete clash-free completion structure described there gives rise to the graph-shaped open answer set depicted by Figure \ref{fig:graphopenanswer}.

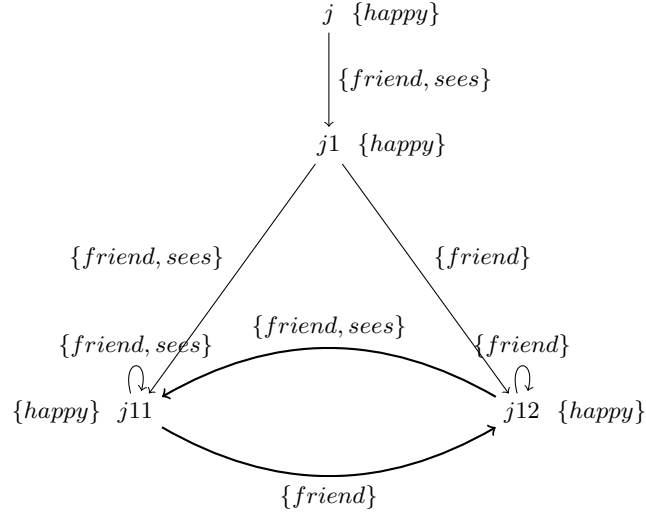
\begin{figure}[htbp]
\begin{center}
\begin{tikzpicture}[auto]
\GraphInit
\node[label=right:{$\{happy\}$}](j){$j$};
\node[yshift=-1.5cm, label=right:{$\{happy\}$}] at (j.south)(j1){$j1$} ;
\node[xshift=-2.3cm, yshift=-3.3cm, label=left:{$\{happy\}$}] at (j1.south west)(j11){$j11$} ;
\node[xshift=2.3cm, yshift=-3.3cm, label=right:{$\{happy\}$}] at (j1.south east)(j12){$j12$} ;
\draw [->] (j)  to node[auto] {$\{friend, sees\}$} (j1);
\draw     [->] (j1) to node[auto, swap] {$\{friend, sees\}$} (j11);
\draw     [->] (j1) to node[auto] {$\{friend\}$} (j12);
\draw     [->] (j11) [loop above] to node[auto] {$\{friend, sees\}$} (j11);
\draw     [->] (j12) [loop above] to node[auto] {$\{friend\}$} (j12);
\tikzset{EdgeStyle/.style = {->}, bend right, swap}
\Edge[label={$\{friend\}$}](j11)(j12)
\Edge[label={$\{friend, sees\}$}](j12)(j11)
\end{tikzpicture}
\end{center}
\caption{Graph-shaped open answer set derived from a clash-free complete completion structure}
\label{fig:graphopenanswer}
\end{figure}

The universe of the open interpretation is the set of nodes of the new graph (identical to the set of nodes of the extended forest), while the interpretation is the set of atoms having as arguments nodes/arcs of the graph and as predicate symbols predicates in the content of these nodes/arcs. In the example above, the open answer set is: $\{happy(j)$, $friend(j,j1)$, $sees(j, j1)$, $happy(j1)$, $friend(j1, j11)$, $sees(j1, j11),$ $happy(j11)$, $friend(j11, j11)$, $sees(j11, j11),$, $friend(j11, j12)$, $sees(j11, j12)$, $\ldots\}$. Intuitively, the atoms having as arguments non-blocked nodes are justified by the way the completion structure was constructed, while atoms having a blocked node as one of the arguments are justified in a similar way to their counterparts\footnote{The counterpart atom of an atom $p(x)/f(x,y)$, where $x$ is a blocked node is the atom $p(z)/f(z,y)$, where $(z,x)$ is a blocking pair.}.

The blocking condition which states that there should be no path from a $p(x)$ to a $q(y)$ in $G$ if $(x, y)$ is a blocking pair, is crucial in showing that this open interpretation is minimal. The intuition was given in the previous section where we discussed how although the content of node $j1$ was included in the content of node $j$ at a certain point in the expansion process they do not form a blocking pair as there is a path from $happy(j)$ to $happy(j1)$. For more details, we refer the reader to the complete proof in appendix.
\end{proof}

\begin{proposition}[completeness]\label{prop:completeness}
Let $P$ be a FoLP and $p\in \upreds{P}$. If $p$ is satisfiable w.r.t. $P$, then there exists a clash-free complete completion structure for $p$ w.r.t. $P$.
\end{proposition}

\begin{proof}[Proof sketch]
If $p$ is satisfiable w.r.t. $P$ then $p$ is forest-satisfiable w.r.t. $P$. We construct a clash-free complete completion structure for $p$ w.r.t. $P$, by guiding the non-deterministic application of the expansion rules with the help of a forest model of $P$ which satisfies $p$ and by taking into account the constraints imposed by the saturation, blocking, and redundancy rules. The proof is inspired by completeness proofs in DL for tableau, for example in \cite{horrocks99practical}, but requires additional mechanisms to eliminate redundant parts from Open Answer Sets.

There are two main stages in the proof: in the first stage, a so-called \emph{complete clash-free relaxed completion structure} is constructed with the help of a forest model of $P$ which satisfies $p$. Such a structure is defined/constructed similarly as a classical completion structure apart from the fact that the redundancy rule is not employed. Accordingly, for a relaxed completion structure to be clash-free the condition regarding the absence of redundant nodes is not relevant.

The second stage consists in transforming such a complete clash-free relaxed completion structure into a clash-free complete completion structure. The transformation consists in several successive steps, each step `shrinking' the structure, by cutting some parts of it, in such a way that the new structure is still a complete clash-free relaxed completion structure. It is shown that the result of this transformation is a structure for which every branch of the tree has at most $k$ nodes with equal content, with $k$ as defined in the redundancy rule, and thus, it is a complete clash-free completion structure.
For more details, we refer the reader to the appendix.
\end{proof}

\begin{proposition}\label{prop:complexity}
The algorithm runs in the worst case in double exponential time in the size of the program.
\end{proposition}

\begin{proof}[Proof sketch]
That the algorithm takes in the worst case at least double exponential time can be seen from the fact that an extended forest in a completion structure has in the worst case a double exponential number of nodes in the size of the program: there are maximum $k+1$ nodes with equal content on any branch of a tree in the completion, where $k=2^n(2^{n^2}-1)+2$, and $n=|\upreds{P}|$, there  are $2^n$ different possible configurations for the content of a unary node, the number of trees in the extended forest is bounded by $|cts(P)|+1$, and the arity of any such tree is bounded by $r=rank(P)$; thus the bound on the number of nodes is $b=(c+1)r^{2^{2n+n^2}-2^{2n}+2^{n+1}}$, which is double exponential in the size of $P$.

We consider the transformation of the algorithm to a deterministic procedure. One can see the deterministic procedure as constructing an AND/OR extended forest with depth double in the size of the largest depth encountered when running the nondeterministic algorithm. At odd levels, there are OR nodes with unexpanded content (they contain just the constraints imposed by their predecessor or the predicate checked to be satisfiable in case of one root node and an empty set for the other root nodes), while at even levels, there are AND saturated nodes which are `realizations' of their predecessor, i.e., they (together with their outgoing arcs and direct successors) describe a possible way to saturate the predecessor node. For every OR node, each of its `realizations' spawns a new copy of the graph $G$. A leaf of the AND/OR extended forest is labeled with \emph{false} if it is a redundant node and with \emph{true} otherwise. A predicate $p$ is satisfiable in such a structure if the root node of every tree in the structure evaluates to \emph{true}.

First of all, we notice that it takes polynomial time to justify the presence of a unary predicate in the content of a node and the presence of a (possibly negated) binary predicate in the content of an arc. Justifying the presence of a negated unary predicate in the content of a node takes exponential time (all groundings of certain unary rules have to be considered, and, in general, there is an exponential number of such groundings). As such, justifying the content of a node takes exponential time, while justifying the content of an arc takes polynomial time.

We count how many ways there are to saturate the content of a node: in the worst case there is an exponential number of choices for justifying the presence of a (possibly negated) unary predicate in the content of a node, a polynomial number of choices to justify the presence of a (possibly negated) binary predicate in the content of a node, and an exponential number of choices regarding the possible content of a node/arc. As such, in the worst case there is an exponential number of choices to saturate a node, thus an exponential number of successors to an OR node, and the maximum branching factor of the AND/OR extended forest is exponential in the size of $P$. The maximum depth is also exponential in the size of $P$ as it is double of the maximum depth of a complete completion structure which is $2^{2n}(2^{n^2}-1)+2^{n+1}$, where $n$ is as above. Thus, the AND/OR extended forest has in the worst case a double exponential number of nodes and arcs and justifying the content of each of these nodes and arcs can be done in exponential time.

There will also be a double exponential number of dependency graphs generated (as an exponential number of them is spawned at each OR node), and each of them has double exponential size (the number of atoms in an open answer set is bounded by $(b-1)m +bn$, where $m=|bpreds(P)|$, and $b$ and $n$ are as above. Checking for the existence of certain paths in such a graph (necessarily for the blocking condition) can be done again in double exponential time.  As such the construction of the AND/OR extended forest and of the dependency graphs can be done in double exponential time. The evaluation of the AND/OR extended forest can be done in double exponential time in the size of $P$, and thus the deterministic procedure, and implicitly our algorithm, runs in the worst case in double exponential time.

\end{proof}

Note that such a high complexity is expected when dealing with tableau-like algorithms. For example in Description Logics, although satisfiability checking in \SHIQ{} is \exptime-complete, practical algorithms  run in non-deterministic double exponential time \cite{tobies}.

\begin{proposition}\label{prop:finitemodelproperty}
FoLPs have the bounded finite model property: if there is an open answer set, there is an open answer set with a universe that is bounded by a number of elements which can be specified in function of the program at hand.
\end{proposition}

\begin{proof}[Proof sketch]
The property follows as a corollary of the soundness and completeness
results. The completeness proof shows that from an open answer set one
can construct a clash-free complete completion structure with maximum
$b$ nodes, where $b$ is defined as in the proof for the complexity
result. At the same time, the soundness result shows that any
clash-free complete structure gives rise to an open answer set whose
universe is exactly the set of nodes of the completion. Thus, any open
answer set can be reduced to an open answer set with a bounded-size
universe.
\end{proof}

Note that the bounded finite model property opens the way also for
standard Answer Set Programming reasoning. Let $P$ be a FoLP.
We define the program $P_{k}$ to be a new program
obtained from $P$ by addition of a constraint $$\gets \naf p(x_1),
\ldots, \naf p(x_k),\naf p(c_1),\ldots, p(c_m)\ ,$$
where $k$ is a natural number, $1 \leq k \leq b-|\cts{P}|$,
$x_1$, $\ldots$, $x_k$ are some newly introduced individuals, and
$\cts{P} = \{c_1,\ldots, c_m\}$. To check whether $p$ is satisfiable w.r.t. $P$ one can simply check answer set existence for the programs $P$, $P_1$, $\ldots,$ $P_{b-|\cts{P}|}$.
Once an answer set is found for one of these programs it can be
concluded that $p$ is satisfiable and the procedure is curtailed. If
no answer set is found, then $p$ is not satisfiable. As $b$ is double exponential in the size of $P$, ${b-|\cts{P}|}$ is also double exponential in the size of $P$.  It follows that constructing the programs $P_1$, $\ldots, P_{b-|\cts{P}|}$ starting from $P$ is also double exponential in the size of $P$ (one has to add to $P$ in each case a new rule with $1, 2, \ldots$, $b-|\cts{P}|$ atoms). Checking the existence of answer sets of $P$, $P_1$, $\ldots P_{b-|\cts{P}|}$, involves a double exponential number of calls to an oracle which checks the existence of answer sets for a non-ground program with bounded predicate arities. According to \cite{eiter+faber+fink+woltran-ComplexResASPBoundedArities:07}
checking answer set existence for a non-ground program with bounded
predicate arities is in $\np^\np (=\Sigma_2^p)$. Thus, such an algorithm runs in the worst case in double exponential time with an oracle in $\Sigma_2^p$. As this is worse than the run-time of our
algorithm (double exponential time, Proposition
\ref{prop:complexity}), we indeed have an indication that our tableaux
algorithm is more efficient than naively using the bounded finite
model property and finite Answer Set Programming.

\section{F-hybrid Knowledge Bases}
\label{sec:fhybrid}

In this section, we introduce \emph{f-hybrid} knowledge bases, a formalism that combines knowledge bases expressed in the Description Logic \SHOQ{} with forest logic programs.

\emph{Description logics (DLs)} are a family of logical formalisms based on frame-based systems \cite{Minsky85} and useful for knowledge  representation. Its basic language features include the notions of {\em concepts}\index{concept} and {\em roles}\index{role} which are used to define the relevant concepts and relations in some (application) domain.  Different DLs can then be identified, among others, by the set of constructors that are allowed to form complex concepts or roles; see, for example, the 2 left-most columns of Table \ref{table1}, that define the constructs in \SHOQ{} \cite{horrocks-shoq}.

The semantics of DLs is given by interpretations $\Int = (\DeltaI, \cdot^{\Int})$ where $\DeltaI$ is a non-empty domain and $\cdot^{\Int}$ is an interpretation function. We summarize the constructs of \SHOQ{} with their interpretation in Table
\ref{table1}.

\begin{table}
\caption{\label{table1}Syntax and Semantics of \SHOQ{} Constructs}
{\small
\begin{center}
\begin{tabular}{|l |c| c|}
\hline
construct name & syntax & semantics\\
\hline
\rule{0em}{1.1em}atomic concept $\C$ 	
			& $A$ 		& $A^{\Int} \subseteq \DeltaI$\\
role  	& $R$ 		& $R^{\Int} \subseteq \DeltaI \times \DeltaI$\\
nominals $\Ind$ 	& $\{o\}$ 	& $\{o^{\Int}\} \subseteq \DeltaI$,\\
\hline
\rule{0em}{1.1em}concept conj.
		        & $C \sqcap D$ 	& $(C \sqcap D)^{\Int} = C^{\Int} \cap D^{\Int}$\\
concept	disj. & $C \sqcup D$ 	& $(C \sqcup D)^{\Int} = C^{\Int} \cup D^{\Int}$\\
negation 		& $\lnot C$ 	& $(\lnot C)^{\Int} = \DeltaI \setminus C^{\Int}$\\
exists restriction 	& $\exists R.C$	& $(\exists R.C)^{\Int} = \{x ~\vert~ \exists y: (x,y) \in R^{\Int} \textrm{~and~} y \in C^{\Int}\}$\\
value restriction 	& $\forall R.C$	& $(\forall R.C)^{\Int} = \{x ~\vert~ \forall y: (x,y) \in R^{\Int} \Rightarrow y \in C^{\Int}\}$\\
atleast restriction 	& $\geq n S.C$	& $(\geq n S.C)^{\Int} = \{x ~\vert~ \#\{y  ~\vert~ (x,y) \in S^{\Int} \mbox{ and } y \in C^{\Int}\} \geq n\}$\\
atmost restriction 	& $\leq n S.C$	& $(\leq n S.C)^{\Int} = \{x ~\vert~ \#\{y  ~\vert~ (x,y) \in S^{\Int} \mbox{ and } y \in C^{\Int}\} \leq n\}$\\
\hline
\end{tabular}
\end{center}
}
\end{table}
A \SHOQ{} \emph{knowledge base} is a set of {\em terminological axioms} $C \sqs D$ with $C$ and $D$ \SHOQ-concept expressions, {\em role axioms} $R \sqs S$ with $R$ and $S$ roles, and {\em transitivity axioms} $\trans(R)$ for a role name $R$. If the knowledge base contains an axiom $\trans(R)$, we call $R$ \emph{transitive}.  For the role axioms in a knowledge base, we define $\sqsast$ as the transitive closure of $\sqs$. A {\it simple role} $R$ in a knowledge base is a role that is not transitive nor does it have any transitive subroles (w.r.t. to reflexive transitive closure $\sqsast$\index{\sqsast} of $\sqs$). Terminological and role axioms express a subset relation: an interpretation $\Int$ \emph{satisfies} an axiom $C_1 \sqs C_2$ ($R_1 \sqs R_2$) if $C_1^{\Int} \subseteq C_2^{\Int}$ ($R_1^{\Int} \subseteq R_2^{\Int}$). An interpretation satisfies a transitivity axiom $\trans{(R)}$ if $R^{\Int}$ is a transitive relation. An interpretation is a \emph{model} of a knowledge base $\Sigma$ if it satisfies every axiom in $\Sigma$.  A concept $C$ is \emph{satisfiable} w.r.t. $\Sigma$ if there is a model $\Int$ of $\Sigma$ such that $C^{\Int} \neq \emptyset$. In order to avoid undecidability of satisfiability checking, \emph{number restrictions} (at most and at least) are always such that the role $R$ in, e.g., $\geq n R.C$, is  (see, e.g., \cite{horrocks99practical}).

We will assume the \textit{unique name assumption}\index{unique name assumption} by imposing that $o^{\Int} =o$ for individuals $o\in \mathbf{I}$.  Note that individuals are thus assumed to be part of any domain $\DeltaI$. Note that OWL does not have the unique name assumption \cite{owlguide},
and thus different individuals can point to the same resource. However, the open answer set semantics gives a Herbrand interpretation to constants, i.e., constants are interpreted as themselves, and for consistency we assume that also DL nominals are interpreted this way.

\begin{example}\label{ex:shoqknowb}
Consider the following \SHOQ{} knowledge base $\Sigma$:
\begin{knowb}
\taxiom{Father}{\exists child.Human \sqcap \neg Female}
\taxiom{\{john\}}{(\leq 2 child.  Human)}
\end{knowb}
Intuitively, the first terminological axiom says that fathers have a human child and are not female. The second axiom says that \lit{john} has less than 2 human children.
\end{example}

\begin{definition}\label{def:f-hybrid}
  An \emph{f-hybrid knowledge base} is a pair $\langle \Sigma,P\rangle$ where $\Sigma$ is a \SHOQ{} knowledge base and $P$ is a FoLP.
\end{definition}

Atoms and literals in $P$ might have as the underlying predicate an atomic concept or role name from $\Sigma$, in which case they are called \emph{DL atoms} and \emph{DL literals} respectively. Additionally, there might be other predicate symbols available, but without loss of generality we assume they cannot coincide with complex concept or role descriptions. Note that we do not impose Datalog safeness or \emph{(weakly) DL safeness} \cite{Motik+Rosati-Reconciling10,rosati-jws,rosati-rr2008,rosati-kr2006} for the rule component.  Intuitively, the restricted shape of FoLPs suffices to guarantee decidability; FoLPs are in general neither Datalog safe nor weakly DL-safe; we will discuss the relation with weakly DL-safeness in detail in Section \ref{sec:discussion}.

\begin{example}\label{ex:fhybrid}
An f-hybrid knowledge base $\langle \Sigma, P\rangle$, with $\Sigma$ as in Example \ref{ex:shoqknowb} and $P$, the FoLP,
\begin{program}
\tsrule{unhappy(X)}{\naf{Father(X)}}
\end{program} indicates that persons that are not fathers are unhappy, where $\lit{Father(X)}$ is a DL literal.
\end{example}

Similarly as in \cite{heymans-tplp2008}, we define, given a DL interpretation $\Int=(\DeltaI,\cdot^{\Int})$ and a ground program $P$, the \emph{projection} $\Pi(P,\Int)$ of $P$ with respect to $\Int$, as follows: for every rule $r$ in $P$,

\begin{itemize}
\item if there exists a DL literal in the head of the form
    \begin{itemize}
    \item $A(t_1,\ldots,t_n)$ with $(t_1,\ldots,t_n)\in A^{\Int}$, or
    \item $\naf{A(t_1,\ldots,t_n)}$ with $(t_1,\ldots,t_n)\not\in A^{\Int}$,
    \end{itemize}
then delete $r$,

\item if there exists a DL literal in the body of the form
    \begin{itemize}
    \item $A(t_1,\ldots,t_n)$ with $(t_1,\ldots,t_n)\not\in A^{\Int}$, or
    \item $\naf{A(t_1,\ldots,t_n)}$ with $(t_1,\ldots,t_n)\in A^{\Int}$,
    \end{itemize}
then delete $r$,

\item otherwise, delete all DL literals from $r$.
\end{itemize}
{
Intuitively, the projection ``evaluates'' the program with respect to $\Int$ by removing (evaluating) rules and DL literals consistently with $\Int$;
conceptually this is similar to the GL-reduct, which removes rules and negative literals consistently with an interpretation of the program.
}
\begin{definition}\label{def:semanticsfhybrid}
  Let $\langle \Sigma,P \rangle$ be an f-hybrid knowledge base. An \emph{interpretation} of $\langle\Sigma,P\rangle$ is a tuple $(U,\Int,M)$ such that
\begin{itemize}
\item $U$ is a universe for $P$,
\item $\Int=(U,\cdot^\Int)$ is an interpretation of $\Sigma$, and
\item $M$ is an interpretation of $\Pi(P_U,\Int)$.
\end{itemize}
Then, $(U,\Int,M)$ is a  \emph{model} of an f-hybrid knowledge base $\langle\Sigma,P\rangle$ if $\Int$ is a model of $\Sigma $ and $M$ is an answer set of $\Pi(P_U,\Int)$.
\end{definition}

The semantics of an f-hybrid knowledge base $\langle \Sigma, P\rangle$ is such that if $\Sigma=\emptyset$, a model of $\langle \Sigma,P\rangle$ corresponds to an open answer set of $P$, and if $P=\emptyset$, a model of $\langle \Sigma,P\rangle$ corresponds to a DL model of $\Sigma$.  In this way, the semantics of f-hybrid knowledge bases is nicely layered on top of both the DL semantics and the open answer set semantics.

\begin{example}\label{ex:semantics}
For the f-hybrid knowledge base $\langle \Sigma, P \rangle $ in Example \ref{ex:fhybrid}, take a universe $U = \{\lit{john,x}\}$ and $\cdot^\Int$
defined such that $\lit{Father}^\Int = \{x\}$, $\lit{child}^\Int = \{(x,\lit{john})\}$, $\lit{Female}^\Int = \emptyset$, $\lit{Human}^\Int = U$, and $\lit{john}^\Int = \lit{john}$. It is easy to see that $\Int=(U,\cdot^\Int)$ is indeed a model of $\Sigma$.

We project the program $P$ taking into account $\Int$, such that $P_U$ is the program
\begin{program}
\tsrule{unhappy(x)}{\naf{Father(x)}}
\tsrule{unhappy(john)}{\naf{Father(john)}}
\end{program}
and since $x\in \lit{Father}^\Int$ and $\lit{john}\not\in \lit{Father}^\Int$, we have that $\Pi(P_U,\Int)$ is
\begin{program}
\tsrule{unhappy(john)}{}
\end{program}
such that $M=\{unhappy(john)\}$ is an answer set of $\Pi(P_U,\Int)$, and $(U,\Int,M)$ is a model of $\langle \Sigma, P \rangle $.
\end{example}

For $p$ a concept expression from $\Sigma$ or a predicate from $P$, we say that $p$ is \emph{satisfiable} w.r.t.~$(\Sigma,P)$ if there is a model
$(U,\Int,M)$ such that $p^{\Int}\neq \emptyset$ or $p(x_1,\ldots,x_n)\in M$ for some $x_1,\ldots,x_n$ in $U$, respectively.  Note that Definition \ref{def:semanticsfhybrid} is in general applicable to other DLs than \SHOQ{} as well as to other programs than FoLPs. Indeed, in \cite{heymans-tplp2008}, a similar definition was used for \dlrom{} and \emph{guarded programs}.

We can reduce satisfiability checking w.r.t. f-hybrid knowledge bases to satisfiability checking of FoLPs only. Roughly, for each concept expression one introduces a new predicate together with rules that define the semantics of the corresponding DL construct.  Constraints then encode the axioms, and the first-order interpretation of DL concept expressions is simulated using free rules.

Taking the knowledge base $\Sigma$ of Example \ref{ex:fhybrid}, $Father\sqsubseteq$ $\exists child.Human$ $\sqcap$ $\neg Female$ can be translated to the constraint \prule{}{Father(X),\naf{(\exists child.Human \sqcap \neg Female)(X)}}, where \lit{(\exists child.Human \sqcap \neg Female)} is a predicate defined by the rules
\[
\prule{(\exists child.Human \sqcap \neg Female)(X)}{(\exists child.
Human)(X),(\neg Female)(X)}
\]
i.e., a DL conjunction translates to a set of literals in the body. Further, we define an exists restriction and negation as follows:
\begin{program}
\tsrule{\exists child.  Human(X)}{child(X,Y),Human(Y)}
\tsrule{\neg Female(X)}{\naf{Female(X)}}
\end{program}
Finally, the first-order semantics of concepts and roles is obtained as follows:
\begin{program}
\tsrule{Father(X)\lor\naf{Father(X)}}{}
\tsrule{Female(X)\lor\naf{Female(X)}}{}
\tsrule{Human(X)\lor\naf{Human(X)}}{}
\tsrule{child(X,Y)\lor\naf{child(X,Y)}}{}
\end{program}

Similarly, the axiom
\axiom{\{john\}}{(\leq 2 child.  Human)}
is translated as the constraint
\[
\prule{}{\{john\}(X),\naf{(\leq 2 child.  Human)(X)}}
\] and
rules
\begin{program}
\tsrule{\{john\}(john)}{}
\tsrule{(\leq 2 child.  Human)(X)}{\naf{(\geq 3 child.  Human)(X)}}
\tsrule{(\geq 3 child.  Human)(X)}{child(X,Y_1),child(X,Y_2),child(X,Y_3),}
&& \mathit{Human(Y_1),Human(Y_2),Human(Y_3),}\\
&& \mathit{Y_1\neq Y_2,Y_1\neq Y_3,Y_2\neq Y_3}
\end{program}

Before proceeding with the formal translation, we define the \emph{closure} of a \SHOQ{} knowledge base $\Sigma$, \clos{\Sigma},  as the smallest set satisfying the following conditions:

\begin{itemize}
\item for each $C\sqs D$ an axiom in $\Sigma$ (role or terminological), $\set{C,D}\subseteq \clos{\Sigma}$,
\item for each $\trans(R)$ in $\Sigma$, $\set{R}\subseteq \clos{\Sigma}$,
\item for every $D$ in \clos{\Sigma}, we have
\begin{itemize}
\item if $D = \lnot D_1$, then $\set{D_1} \subseteq \clos{\Sigma}$,
\item if $D = D_1 \sqcup D_2$, then $\set{D_1,D_2} \subseteq \clos{\Sigma}$,
\item if $D = D_1 \sqcap D_2$, then $\set{D_1,D_2} \subseteq \clos{\Sigma}$,
\item if $D = \exists R.D_1$, then $\set{R,D_1}\cup \set{\exists S.D_1\mid S\sqsast R, S\neq R, \trans(S)\in \Sigma} \subseteq \clos{\Sigma}$,
\item if $D = \forall R.D_1$, then $\set{\exists R.\lnot D_1} \subseteq  \clos{\Sigma}$,
\item if $D = \numberrestless{n}{Q}{D_1}$, then $\set{\numberrestgreater{n+1}{Q}{D_1}} \subseteq \clos{\Sigma}$,
\item if $D = \numberrestgreater{n}{Q}{D_1}$, then $\set{Q,D_1} \subseteq \clos{\Sigma}$.
\end{itemize}
\end{itemize}
Concerning the addition of the extra $\exists S.D_1$ for $\exists R.D_1$ in the closure, note that $x\in(\exists R.D_1)^{\Int}$ holds if there is some $(x,y)\in R^{\Int}$ with $y\in D_1^{\Int}$, and, in particular, $S\sqsast R$ with $S$ transitive such that $(x,u_0)\in S^{\Int}, \ldots, (u_n,y)\in S^{\Int}$ with $y\in D_1^{\Int}$.  The latter amounts to $x\in (\exists S.D_1)^{\Int}$.  Thus, in the open answer set setting, we have that $\exists R.D_1(x)$ is in the open answer set if $R(x,y)$ and $D_1(y)$ hold or $\exists S.D_1(x)$ holds for some transitive subrole $S$ of $R$.  The predicate $\exists S.D_1$ will be defined by adding recursive rules, hence the inclusion of such predicates in the closure.
\par
Furthermore, for a $\numberrestless{n}{Q}{D_1}$ in the closure, we add $\set{\numberrestgreater{n+1}{Q}{D_1}}$, since we will base our definition of the former predicate on the DL equivalence $\numberrestless{n}{Q}{D_1}\equiv \neg \numberrestgreater{n+1}{Q}{D_1}$.
\par
Formally, we define $\Phi(\Sigma)$ to be the following FoLP, obtained from the \SHOQ{} knowledge base $\Sigma$:
\begin{itemize}
\item For each terminological axiom $C\sqs D\in \Sigma$, add the constraint
\begin{equation}\label{eq:axiom}
\prule{}{C(X),\naf{D(X)}}
\end{equation}
\item For each role axiom $R\sqs S\in \Sigma$, add the
constraint
\begin{equation}\label{eq:role}
\prule{}{R(X,Y),\naf{S(X,Y)}}
\end{equation}

\item Next, we distinguish between types of concept expressions that appear in \clos{\Sigma}.  For each $D\in \clos{\Sigma}$:
\begin{itemize}
\item if $D$ is a concept name, add
\begin{equation}\label{eq:freeconcept}
\prule{D(X)\lor \naf{D}(X)}{}
\end{equation}
\item if $D$ is a role name, add
\begin{equation}\label{eq:freerole}
\prule{D(X,Y)\lor \naf{D}(X,Y)}{}
\end{equation}
\item if $D=\{o\}$, add
\begin{equation}\label{eq:nominal}
\prule{D(o)}{}
\end{equation}
\item if $D=\neg E$, add
\begin{equation}\label{eq:neg}
\prule{D(X)}{\naf{E(X)}}
\end{equation}
\item if $D=E\sqcap F$, add
\begin{equation}\label{eq:conj}
\prule{D(X)}{E(X),F(X)}
\end{equation}
\item if $D=E\sqcup F$, add
\begin{equation}\label{eq:disj}
\begin{programxy}
\tsrule{D(X)}{E(X)}
\tsrule{D(X)}{F(X)}
\end{programxy}
\end{equation}
\item if $D=\exists Q.E$, add
\begin{equation}\label{eq:exists}
\prule{D(X)}{Q(X,Y),E(Y)}
\end{equation}
and for all $S\sqsast Q$, $S\neq Q$, with $\trans(S)\in \Sigma$, add rules
\begin{equation}\label{eq:transexists}
\prule{D(X)}{(\exists S.E)(X)}
\end{equation}
If $\trans(Q)\in\Sigma$, we further add the rule
\begin{equation}\label{eq:transexists2}
\prule{D(X)}{Q(X,Y),D(Y)}
\end{equation}
\item if $D=\forall R.E$, add
\begin{equation}\label{eq:forall}
\prule{D(X)}{\naf{(\exists R.\neg E)(X)}}
\end{equation}
\item if $D= \numberrestless{n}{Q}{E}$, add
\begin{equation}\label{eq:numberless}
\prule{D(X)}{\naf{\numberrestgreater{n+1}{Q}{E}(X)}}
\end{equation}
\item if $D=\numberrestgreater{n}{Q}{E}$, add
\begin{equation}\label{eq:numbergreater}
\psrule{D(X)}{Q(X,Y_1),\ldots, Q(X,Y_n), E(Y_1),\ldots, E(Y_n),({Y_i\neq
Y_j})_{1\leq i\neq j\leq n}}
\end{equation}
\end{itemize}
\end{itemize}

Rule $(\ref{eq:exists})$ is what one would intuitively expect for the exists restriction.  However, in case $Q$ is transitive this rule is not enough.  Indeed, if $Q(x,y)$, $Q(y,z)$, $E(z)$ are in an open answer set, one expects $(\exists Q.E)(x)$ to be in it as well if $Q$ is transitive.  However, we have no rules enforcing $Q(x,z)$ to be in the open answer set without violating the FoLP restrictions. We can solve this by adding to $(\ref{eq:exists})$ the rule $(\ref{eq:transexists2})$, such that such a chain $Q(x,y)$, $Q(y,z)$, with $E(z)$ in the open answer set correctly deduces $D(x)$.
\par
It may still be that there are transitive subroles of $Q$ that need the same recursive treatment as above.  To this end, we introduce rule $(\ref{eq:transexists})$.
\par
We do not need such a trick with the number restrictions since the roles $Q$ in a number restriction are required to be simple, i.e., without transitive subroles.

\begin{proposition}\label{prop:folpispol}
Let $\langle\Sigma,P\rangle$ be a \SHOQ{} knowledge base.  Then, $\Phi(\Sigma)\cup P$ is a FoLP, and has a size that is polynomial in the size of $\Sigma$.
\end{proposition}
\begin{proof}
Observing the rules in $\Phi(\Sigma)$, it is clear that this program is a FoLP.
\par
The size of the elements in $\clos{\Sigma}$ is linear and the size of $\clos{\Sigma}$ itself is polynomial in $\Sigma$. The size of the FoLP
$\Phi(\Sigma)$ is polynomial in the size of $\clos{\Sigma}$.  The only non-trivial case in showing the latter arises by the addition of rule $(\ref{eq:numbergreater})$ which introduces $\frac{n(n-1)}{2}$ inequalities for a number restriction $\numberrestgreater{n}{Q}{E}$. We assume, as is not uncommon in DLs (see, e.g., \cite{tobies}), that the number $n$ is represented in unary notation
\[
\underbrace{11\ldots1}_n
\]
such that the number of introduced inequalities is quadratic in the size of the number restriction.
\end{proof}
\begin{proposition}\label{prop:translation}
Let $\langle \Sigma,P\rangle$ be an f-hybrid knowledge base.  Then, a predicate $p$ is satisfiable w.r.t.~$(\Sigma,P)$ iff $p$ is satisfiable w.r.t.~$\Phi(\Sigma)\cup P$.
\end{proposition}
\begin{proof} The proof goes along the lines of the proof in \cite[Theorem 1]{heymans-tplp2008}.

\noindent ($\Rightarrow$). Assume $p$ is satisfiable w.r.t.~$(\Sigma,P)$, i.e., there exists a model $(U,{\Int},M)$ of $(\Sigma,P)$ in which $p$ has a non-empty extension.  Now, we construct the open interpretation $(U,N)$ of $\Phi(\Sigma)\cup P$ as follows:
\[
N=M \cup \set{\pred{C}(x) \mid x \in C^{\Int}, C\in \clos{\Sigma}}
\cup \set{\pred{R}(x_1,x_2)\mid
(x_1,x_2) \in {R}^{\Int}, R\in \clos{\Sigma}}
\]
with $C$ and $R$ concept expressions and role names respectively.
\par
It is easy to verify that $(U,N)$ is an open answer set of $\Phi(\Sigma)\cup P$ and that $(U,N)$ satisfies $p$.


\noindent ($\Leftarrow$). Assume $(U,N)$ is an open answer set of $\Phi(\Sigma) \cup P$ such that $p$ is satisfied.  We define the interpretation
$(U,{\Int},M)$ of $(\Sigma, P)$ as follows:
\begin{itemize}

\item  $\Int=(U,\cdot^\Int )$ is defined such that $A^{\Int} =\set{x \mid \pred{A}(x) \in N}$ for concept names $A$, ${P}^{\Int} = \set{(x_1,x_2)
\mid \pred{{P}}(x_1,x_2) \in N}$ for role names ${P}$ and ${o}^{\Int} = o$, for $o$ a constant symbol in $\Sigma$. $\Int$ is then an interpretation of $\Sigma$.

\item $M = \setmin{N}{\set{p(x_1,\ldots,x_n)\mid p\in \clos{\Sigma}}}$, such that $M$ is an interpretation of $\Pi(P_U,\Int)$.
\end{itemize}
As a consequence, $(U,\Int,M)$ is an interpretation of $\langle \Sigma,P\rangle$ and it is easy to verify that $(U,\Int,M)$ is a model of $(\Sigma,P)$ which satisfies $p$.
\end{proof}

Note that Proposition \ref{prop:translation} also holds for satisfiability checking of concept expressions $C$: introduce a rule $p(X)\gets C(X)$ in $P$ and check satisfiability of $p$.

Using the translation from f-hybrid knowledge bases to  forest logic programs in Proposition \ref{prop:translation} and the polynomiality of this
translation (Proposition \ref{prop:folpispol}), together with the complexity of the terminating, sound, and complete algorithm for satisfiability checking w.r.t. FoLPs, we have the following result:

\begin{proposition}\label{prop:fhybridmember}
Satisfiability checking w.r.t. f-hybrid knowledge bases is in \xnexptime{2} in the size of the f-hybrid knowledge base.
\end{proposition}

As satisfiability checking of \ALC concepts w.r.t. an \ALC TBox  (note that \ALC is a fragment of \SHOQ) is \exptime-complete \cite[Chapter 3]{dlbook}, we have that satisfiability checking w.r.t. f-hybrid knowledge bases is \exptime-hard.

\begin{proposition}\label{prop:fhybridhard}
Satisfiability checking w.r.t. f-hybrid knowledge bases is \exptime-hard.
\end{proposition}

\section{Simple Forest Logic Programs}
\label{sec:simpleFoLPs}

\emph{Simple Conceptual Logic Programs (CoLPs)}, were defined in \cite{feier+heymans-SoundComplAlgSimpleCoLPs:08} as a fragment of \emph{Conceptual Logic Programs (CoLPs)} \cite{heymans-amai2006}. As mentioned in the introduction, simple Conceptual Logic Programs simplify Conceptual Logic Programs by introducing a restriction on predicate recursion in programs. Here we adopt a similar restriction on Forest Logic Programs, and we obtain a fragment which we call simple Forest Logic Programs (simple FoLPs). As we will see, our algorithm can be easily adapted such that it checks satisfiability w.r.t. simple FoLPs in exponential time, one exponential level lower than the time needed for FoLPs.


Some preliminaries are needed for introducing this fragment. For such a FoLP $P$, let $D(P)$ be the \emph{marked positive predicate dependency graph}: $D(P)$ is a directed graph that has as vertices the non-free predicates from $P$ and as arcs tuples $(p,q)$ if there is either a rule of the form (\ref{eq:unary}) or a rule of the form (\ref{eq:binary}) with a head literal $l_1$ and a positive body literal $l_2$ such that $pred(l_1)=p$, and $pred(l_2)=q$. An edge $(p,q)$ is called \emph{marked}, if $q$ is a predicate in some $\delta_m$ for rules (\ref{eq:unary}), respectively $\delta$ for rules (\ref{eq:binary}). In order for $P$ to be a simple FoLP, $D(P)$ must not contain any cycle that has a marked edge.


The restriction on $D(P)$ ensures that there is no path from some atom $p(x)$ to some atom $p(y)$ in the atom dependency graph of $P_U$ which does not contain some atom $q(z)$, such that $q$ is free, where $p \in\upreds{P}$,  $q\in\preds{P}$, $U$ is some arbitrary universe, and $x,y \in U$, $x \neq y$.
Consider the program $P$:

\begin{programn}
\nrule{r_1}{p(X)}{q(X), f(X,Y),\naf{p(Y)}}
\nrule{r_2}{q(X)}{p(X)}
\nrule{r_3}{f(X,Y)}{g(X,Y), q(Y)}
\end{programn}

The marked positive dependency graph is depicted in Figure \ref{fig:depgra}. While $(p,q,p)$ is an unmarked cycle, $(q,p,f,q)$ is a marked cycle, and thus $P$ is not a simple FoLP. However, if the last rule of $P$ is dropped, it becomes a simple FoLP.

\begin{figure}[htbp]
\begin{center}
\begin{tikzpicture}[>=latex']
\usetikzlibrary{calc}
  \tikzset{node distance = 2cm}
  \GraphInit[vstyle=Normal]
  \Vertex{p}
  \EA(p){q}
  \SO(p){f}
  \EA(f){g}
  \tikzset{LabelStyle/.style = {fill=white,sloped}}
  \tikzset{EdgeStyle/.style = {->,bend right}}

  \Edge(p)(q)
  \Edge(q)(p)
  \Edge(p)(f)
  \Edge[label=$*$](f)(q)

  \Edge(f)(g)
\end{tikzpicture}
\end{center}
\caption{Marked Dependency Graph $D(P)$}
\label{fig:depgra}
\end{figure}
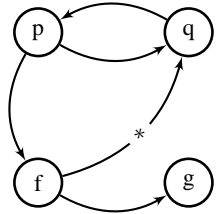

\subsection{Reasoning with Simple FoLPs}

Similarly as for FoLPs we define an initial completion structure for checking the satisfiability of a unary predicate $p$ w.r.t. a FoLP $P$. The completion is expanded via expansion rules, whose application is governed by applicability rules. All expansion rules for FoLPs (rules (i)-(vi)) are employed also in this case. As concerns the applicability rules, rule \emph{(vii) Saturation} stays the same, rule \emph{(viii) Blocking} is modified such that instead of the complex condition for FoLPs an anywhere subset blocking technique is applied, and rule \emph{(ix) Redundancy} is dropped. We give below the formal definition for the new blocking rule:

\subsubsection{(viii') Blocking}
A node $x \in N_\EF$ is \emph{blocked} if there is a saturated node $y\in N_\EF$, with $y \not \in \cts{P}$, such that $\ct(x) \subseteq \ct(y)$. Like for FoLPs, we call $(y,x)$ a \emph{blocking pair}. No expansions can be performed on a blocked node.

Intuitively, if there is a saturated node $y$ in $EF$ which is not a constant, whose content includes the content of $x$, as there are no paths from any $p(x)$ to some $q(y)$ (due to the restriction that there is no cycle in the marked positive dependency graph of $P$), one can reuse the justification for $y$ when dealing with $x$. Note that $y$ and $x$ do not have to be on the same path in a tree in $\EF$. Such a blocking technique is called ``anywhere blocking''.

The notions of \emph{contradictory}, \emph{clash-free}, \emph{complete} completion structure are defined analogously as for FoLPs.

\begin{proposition}[termination]\label{prop:stermination}
Let $P$ be a simple FoLP and $p \in \upreds{P}$.  Then, one can construct a finite complete
completion structure by a finite number of applications of the expansion
rules (i)-(vi) to the initial completion structure for $p$ w.r.t. $P$, taking into account
the applicability rules (vii) and (viii').
\end{proposition}

\begin{proof}[Proof sketch]
Clearly, if one has a finite completion structure that is not complete, a finite application of expansion rules would complete it unless successors are introduced. One cannot introduce successors indefinitely as given the finite number of possible contents of a node, the blocking condition will eventually be met.
\end{proof}

\begin{proposition}[soundness]\label{prop:ssoundness}
Let $P$ be a simple FoLP and $p\in \upreds{P}$. If there exists a complete clash-free completion structure for $p$ w.r.t. $P$ (expanded according to rule (i)-(vii) and (viii')), then $p$ is satisfiable w.r.t. $P$.
\end{proposition}

\begin{proof}[Proof sketch]
Similarly to the case for FoLPs, from a clash-free complete completion structure, one can construct an open interpretation and show that this interpretation is an open answer set of $P$ that satisfies $p$.  Here, due to the restrictions on the the predicate dependency graph of the program, the subset blocking condition is enough to ensure minimality of such an open interpretation. There are no infinite dependency chains which are not cycles in the atom dependency graph of the grounded program.
\end{proof}

\begin{proposition}[completeness]\label{prop:scompleteness}
Let $P$ be a simple FoLP and $p\in \upreds{P}$. If $p$ is satisfiable w.r.t. $P$, then there exists a clash-free complete completion structure for $p$ w.r.t. $P$.
\end{proposition}

\begin{proof}[Proof sketch]
If $p$ is satisfiable w.r.t. $P$ then $p$ is forest-satisfiable w.r.t. $P$. We construct a clash-free complete completion structure for $p$ w.r.t. $P$, by guiding the non-deterministic application of the expansion rules with the help of a forest model of $P$ which satisfies $p$ and by taking into account the constraints imposed by the saturation and the new blocking rule.
\end{proof}

\begin{proposition}\label{prop:scomplexity}
The algorithm runs in the worst case in exponential time in the size of the program.
\end{proposition}

\begin{proof}[Proof sketch]
The size of a completion structure is bounded by the following factors: if we leave all the leaves of the trees in the completion apart, there are at most $2^p+c$ nodes, where $p=|\upreds{P}|$, and $c=|cts(P)|$, as there are at most $2^p$ different possible configurations for the content of a unary node, and all the nodes which are not leaves or constants have to have different content (otherwise they would form blocking pairs and at least one of them would be a leaf). The maximum number of leaves is $r(2^p+c-1)$, where $r=rank(P)$ is the maximum arity of any of the trees in the extended forest. So, the completion has in the worst case an exponential number of nodes in the size of the program: $b=(2^p+c)(r+1)-r$. As was the case for FoLPs, the nondeterministic algorithm can be determinized using an AND/OR extended forest. The new deterministic version will still run in the worst case in exponential time, and thus we can conclude that the algorithm runs in exponential time.
\end{proof}

Note that the complexity of simple FoLPs is one level lower than the complexity of full FoLPs, the decrease in complexity being achieved by employing the anywhere blocking technique. This, at its turn, has been made possible through the restriction imposed on the shape of simple FoLPs. By allowing anywhere blocking for full FoLPs we would lose the soundness of the algorithm (in particular the interpretation constructed as described in the soundness proof would not always be minimal).

\begin{proposition}\label{prop:sfinitemodelproperty}
Simple FoLPs have the bounded finite model property: if there is an open answer set, there is an open answer set with a universe that is bounded by a number of elements which can be specified in function of the program at hand.
\end{proposition}

\begin{proof}[Proof sketch]
The property follows as a corollary of the soundness and completeness results. The completeness proof shows that from an open answer set one can construct a clash-free complete completion structure with maximum $b$ nodes, where $b$ is as defined above. At the same time, the soundness result shows that any clash-free complete structure gives rise to an open answer set whose universe is exactly the set of nodes of the completion. Thus, any open answer set can be reduced to an open answer set with a bounded-size universe.
\end{proof}

\subsection{Simple F-hybrid Knowledge Bases}

Similar with defining f-hybrid knowledge bases one can define simple f-hybrid knowledge bases which are combinations of \ALCHOQ{} knowledge bases with simple FoLPs. An \ALCHOQ{} knowledge base can be seen as a \SHOQ{} knowledge base where no transitive roles are allowed.

\begin{definition}\label{def:sf-hybrid}
  A \emph{simple f-hybrid knowledge base} is a pair $\langle \Sigma,P\rangle$ where $\Sigma$ is an \ALCHOQ{} knowledge base and
  $P$ is a simple FoLP.
\end{definition}

Note that the f-hybrid KB in example \ref{ex:fhybrid} is a simple f-hybrid KB.

The semantics of simple f-hybrid knowledge bases is defined similarly as the semantics of f-hybrid knowledge bases. We employ the same strategy for reasoning with simple f-hybrid knowledge bases as the one used for reasoning with f-hybrid knowledge bases: translating satisfiability checking in the DL part of the knowledge base, the \ALCHOQ{} knowledge base, into satisfiability checking in the LP part of the hybrid formalism, FoLPs. In order to do this we define the  \emph{closure} \clos{\Sigma} of an \ALCHOQ{} knowledge base $\Sigma$ and the transformation $\Phi(\Sigma)$ from an \ALCHOQ{} knowledge base to a FoLP in a similar way as their homonym transformation in Section \ref{sec:fhybrid}: we simply drop the axioms which deal with transitivity in the general case. In particular, by dropping axiom \ref{eq:transexists2}, the obtained FoLP becomes a simple FoLP:

\begin{proposition}\label{prop:sfolpispol}
Let $\langle\Sigma,P\rangle$ be an \ALCHOQ{} knowledge base.  Then, $\Phi(\Sigma)\cup P$ is a simple FoLP, and has a size that is polynomial in the size of $\Sigma$.
\end{proposition}

\begin{proof}[Proof sketch]
That $\Phi(\Sigma)\cup P$ is a FoLP which has a size that is polynomial in the size of $\Sigma$ follows from proposition \ref{prop:folpispol} and the fact that any \ALCHOQ{} is a \SHOQ{} knowledge base. That the resulted FoLP is a simple FoLP can be seen by analysis of the shape of axioms used for defining $\Phi$ introduced in Section \ref{sec:fhybrid}: the only axiom which introduces predicate recursion is axiom \ref{eq:transexists2} which has been eliminated in this version of the translation.
\end{proof}

\begin{proposition}\label{prop:stranslation}
Let $\langle \Sigma,P\rangle$ be a simple f-hybrid knowledge base.  Then, $p$ is satisfiable w.r.t.~$(\Sigma,P)$ iff $p$ is satisfiable w.r.t.~$\Phi(\Sigma)\cup P$.
\end{proposition}

The proof for the above proposition is similar with the proof for \ref{prop:translation}. That there exists such a polynomial translation from simple f-hybrid knowledge bases to  forest logic programs, together with the complexity of the terminating, sound, and complete algorithm for satisfiability checking w.r.t. simple FoLPs, we have the following result:

\begin{proposition}\label{prop:simplefhybridmember}
Satisfiability checking w.r.t. simple f-hybrid knowledge bases is in \exptime{}.
\end{proposition}

As satisfiability checking of \ALC concepts w.r.t. an \ALC TBox  (note
that \ALC is a fragment of \ALCHOQ) is \exptime-complete \cite[Chapter
3]{dlbook}, we have that satisfiability checking w.r.t. simple f-hybrid knowledge bases is \exptime-hard, and combined with the result above, that satisfiability checking w.r.t. simple f-hybrid knowledge bases is \exptime-complete.

\begin{proposition}\label{prop:simplefhybridhard}
Satisfiability checking w.r.t. simple f-hybrid knowledge bases is \exptime-complete.
\end{proposition}




\section{Discussion and Related Work}
\label{sec:discussion}

We compare f-hybrid knowledge bases to r-hybrid knowledge bases from \cite{rosati-rr2008}, which extend \dlpluslog{} from \cite{rosati-kr2006} with inequalities and negated DL atoms.

In \cite{rosati-rr2008}, an r-hybrid knowledge base consists of a DL knowledge base (the specific DL is a parameter) and a disjunctive Datalog program where each rule is \emph{weakly DL-safe}:
\begin{itemize}
\item every variable in the rule appears in a positive atom in the body
of the rule (\emph{Datalog safeness}), and
\item every variable either occurs in a positive non-DL atom in the body of
the rule, or it only occurs in positive DL atoms in the body of the rule.
\end{itemize}

The semantics of r-hybrid and f-hybrid knowledge bases overlap to a large extent. The main difference is that f-hybrid knowledge bases do not make the \emph{standard names assumption}, in which basically the domain of every interpretation is the same infinitely countable set of constants.

Some key differences to note are the following:

\begin{itemize}

\item We do not require Datalog safeness neither do we require weakly DL-safeness.  Indeed, f-hybrid knowledge bases may have a rule component (i.e., the program part) that is not weakly DL-safe. Take the f-hybrid knowledge base $\langle \Sigma, P \rangle$ from Example \ref{ex:fhybrid} with $P$:
\begin{program}
\tsrule{unhappy(X)}{\naf{Father(X)}}
\end{program}
The atom $\lit{Father(X)}$ is a DL-atom such that the rule is neither Datalog safe nor weakly DL-safe. Modifying the program to
\begin{program}
\tsrule{unhappy(X)}{Human(X),\naf{Father(X)}}
\end{program}
leads to a Datalog safe program ($X$ appears in a positive atom $\lit{Human(X)}$ in the body of the rule), however, it is still not weakly DL-safe as $X$ is not appearing only in positive DL-atoms.

On the other hand, both the above rules are FoLPs and thus constitute a valid component of an f-hybrid knowledge base.

\item In the case of r-hybrid knowledge bases, due to the safeness conditions, it suffices for satisfiability checking to ground the rule component with the constants appearing explicitly in the knowledge base.\footnote{\cite{rosati-rr2008,rosati-kr2006} considers checking satisfiability of knowledge bases rather than satisfiability of predicates. However, the former can easily be reduced to the latter.} One does not have such a property for f-hybrid knowledge bases.  Consider the f-hybrid knowledge base $\langle \Sigma, P\rangle$ with  $\Sigma=\emptyset$ and the program $P$

\begin{program}
\tsrule{a(X)}{\naf{b(X)}}
\tsrule{b(0)}{}
\end{program}

This program is a FoLP, but it is not Datalog safe nor is it weakly DL-safe.  Grounding only with the constants in the program yields the projection

\begin{program}
\tsrule{a(0)}{\naf{b(0)}}
\tsrule{b(0)}{}
\end{program}

such that $a$ is not satisfiable.  However, grounding with, e.g., $\{0,x\}$, one gets
\begin{program}
\tsrule{a(0)}{\naf{b(0)}}
\tsrule{a(x)}{\naf{b(x)}}
\tsrule{b(0)}{}
\end{program}
such that $a$ is indeed satisfiable, in correspondence with one would expect.

\item Decidability for satisfiability checking  of r-hybrid knowledge bases is guaranteed if decidability of the conjunctive query containment/union of conjunctive queries containment problems is guaranteed for the DL at hand. In contrast, we relied on a translation of DLs to FoLPs for establishing decidability, and not all DLs can be translated this way; we illustrated the translation for \SHOQ{}.

\end{itemize}

Conceptual modeling using FoLPs is not restricted to simulating DL KBs: one can also translate \emph{object-role modeling (ORM)} models as sets of FoLP rules. In \cite{phdthesis-heymans}p.96 a translation of a particular ORM model to a CoLP (thus, also a FoLP) is provided. While a formal translation from ORM models to CoLPs/FoLPs is not provided there, the example translation shows how one can use CoLP satisfiability checking to verify that the various ORM object types can be populated, that some derived properties do (not) hold, etc.

MKNF$^+$ knowledge bases \cite{Motik+Rosati-Reconciling10}, consist of a DL component and a component of so-called MKNF$^+$ rules. Such MKNF$^+$ rules allow for modal operators $\mathbf{K}$ and $\mathbf{not}$ in front of atoms, but also for non-modal atoms, unlike their predecessor, hybrid MKNF knowledge bases \cite{motikrosati-2006,motik-iswc2006}; non-modal atoms can be eliminated by a transformation leading to MKNF knowledge bases. Also, unlike the rules in hybrid MKNF knowledge bases, atoms in MKNF$^+$ rules are `generalized', in the sense that they can be arbitrary first-order formulae. This allows the approach to capture languages like \emph{EQL-Lite}$(\mathcal{Q})$ \cite{calvanese:eql-lite:07}, dl-programs by \cite{eiter-ai2008} and disjunctive dl-programs by \cite{Lukasiewicz04anovel}. Other approaches to integrating ontologies and rules which are generalized by MKNF$^+$ knowledge bases are: \cite{levy96carin}, $\mathcal{AL}$-log \cite{doni-lenz-nard-scha-98}, DL-safe rules \cite{motik}, the Semantic Web Rule Language (SWRL) \cite{horrocks-www2004}, and r-hybrid knowledge bases \cite{rosati-rr2008}.

MKNF knowledge bases are in the general case undecidable. In order to regain decidability a \emph{DL-Safety} condition is imposed, together with a notion of admissibility which concerns decidability for the DL inference. As with r-hybrid knowledge bases, our f-hybrid knowledge bases do not have such a restriction of the interaction between the structural DL component and the rule component, but rely instead on the existence of an integrating framework (FoLPs under an open answer set semantics) for which we provided reasoning support in this article.


\emph{Description Logic Programs} \cite{grosof} represent the common subset of OWL-DL ontologies and Horn logic programs (programs without negation as failure or disjunction). As such, reasoning can be reduced to normal LP reasoning. In \cite{motik}, a clever translation of \SHIQD (\SHIQ with data types) combined with \emph{DL-safe rules} to disjunctive Datalog is provided. The translation relies on a translation to clauses and subsequently applying techniques from basic superposition theory. Reasoning in $\mathcal{DL}\mathit{+log}$ \cite{rosati-kr2006} and r-hybrid knowledge bases (see above) does not use a translation to other approaches, but defines a specific algorithm based on a partial grounding of the program and a test for containment of conjunctive queries over the DL knowledge bases. \emph{dl-programs} \cite{eiter-ai2008} have a more loosely coupled take on integrating DL knowledge bases and logic programs by allowing the program to query the DL knowledge base while as well having the possibility to
send (controlled) input to the DL knowledge base. Reasoning is done via a stable model computation of the logic program, interwoven with queries that are oracles to the DL part.

\emph{Description Logic Rules (DL rules)} \cite{Krotzsch+Rudolph+Hitzler-DLRules:2008} are defined as decidable fragments of SWRL. Rules have a tree-like structure
similar to the structure of FoLPs. They are positive rules with only unary and binary atoms, corresponding to concept expressions and role names in a specific DL, where some relations between the terms appearing in the atoms in a rule have to be fulfilled: (i) every term can be reached by maximum one path from another term (a term reaches another if it is the first argument of the first atom in a chain of binary atoms where the last argument of the last atom is the term reached), (ii) the first term in the head is an `initial' term, i.e., it is not reached from any other term, (iii) each non-initial node is reached from exactly one initial node. Thus, a syntactical comparison between FoLP rules and DL rules yields the following:

\begin{itemize}
\item FoLPs allow for a negation as failure operator, while DL rules do not support any type of negation
\item FoLPs allow for binary atoms conjunctions, i.e. the presence of binary atoms having identical arguments in the body of a rule, while DL rules disallow this (the presence of such atoms would imply the presence of two paths between the two terms which compose the arguments of these atoms)
\item DL rules allow for term tree depths higher than 1, i.e., for constructions like $f(X,Y),$ $ g(Y, Z), \ldots$ in the body of a rule. FoLPs allow only term trees of depth 1, but such constructions can be seen as syntactic sugar in our language as one can always simulate a rule with term tree depth of $n$ via $n$ FoLP rules with term tree depth of 1.
\item DL rules allow for unsafe rules like $f(X,Y) \gets C(X)$, or $f(X, Y) \gets g(Z, T)$, while FoLP rules do not allow for such constructions.
\end{itemize}

Although Description Logic Rules have tree-shaped bodies and are from this perspective similar to FoLPs, their semantics is not a minimal model semantics. Like Description Logics, their semantics is first-order based. Depending on the underlying DL, one can distinguish between $\mathcal{SROIQ}$ rules, $\mathcal{E}\mathcal{L}^{++}$ rules, Description Logic Program rules, and ELP rules \cite{Krotzsch+Rudolph+Hitzler-ELPRules:2008}.

The most expressive fragment, $\mathcal{SROIQ}$ rules, does not actually extend $\mathcal{SROIQ}$, as the rules can be mapped to $\mathcal{SROIQ}$. In order to ensure that such a translation is possible some more restrictions are imposed on the rule component. One of these restrictions concerns the fact that simple roles are defined also with respect to the definition of their counterpart binary atoms in the rule KB: any binary atom which is defined via a rule with more than one atom in the body corresponds to a non-simple role, and thus cannot appear in a qualified number restriction, a role disjunction axiom or a role reflexivity axiom. Obviously, there is no such restriction on FoLPs as the translation is performed in the other direction, from the DL KB to the rule KB, and thus there is no need to have such a simplicity assumption in the rule KB.

In the case of $\mathcal{E}\mathcal{L}^{++}$ rules, the DL rules are the core expressive mechanism to which the $\mathcal{E}\mathcal{L}^{++}$ KBs are reduced. No simplicity or regularity constraints are imposed on the rule KB.

Description Logic Program rules have as an underlying formalism the DLP fragment described above. So-called DL2 KBs are defined as combinations of DLP rules KBs with DLP KBs, which additionally might contain role disjunction axioms and/or role asymmetry axioms. No simplicity or regularity condition is imposed. Such a KB can be transformed into a set of function-free first-order Horn rules.

The last type of DL rules, ELP rules, can be seen as an extension of both $\mathcal{E}\mathcal{L}^{++}$ rules and Description Logic Program rules, hence their name. In \cite{Krotzsch+Rudolph+Hitzler-ELPRules:2008} a new type of DL rules, so-called extended DL rules, is introduced. This extended type of rules allows for `role conjunctions' in rule bodies, i.e., constructions like $f(X, Y), g(X, Y)$ as long as both $f$ and $g$ are simple roles, or the presence of binary atoms $f(X, X)$ in the rule bodies as long as $f$ is simple. Also, a relaxed restriction on simple roles\footnote{The restriction is relaxed as compared to the restriction on $\mathcal{SROIQ}$ rules; there is no such restriction for general DL rules.} is introduced: only certain role chains are omitted from DL rules with simple roles in the head, rules like $f(X, Y) \gets a(X) \wedge b(Y)$ and $f(X, Y) \gets g(X, Y) \wedge D(Y)$ not precluding $f$  to be a simple role. Note that rules of the first type are not allowed by FoLPs.

The focus in DL rules is on extending DLs with rule bases which are as
expressive as possible while at the same time preserving the
computational properties of the initial DL. This leads sometimes to
rather intricate syntactical characterizations of different fragments.
Syntactically, some of these fragments allow for more complex rule
shapes than FoLP rules, but FoLPs distinguish themselves through the
fact that they have a \emph{negation as failure} operator and adopt a minimal
model semantics, thus adding a different type of expressivity to such
combinations of rules and ontologies, which is not specific to the DL
world. This seems to come at the price of reasoning complexity (note
that we do not have a tight characterization of FoLPs).

There are several extensions of DL which adopt a minimal-style semantics like autoepistemic \cite{donini+nardi+rosati-DLMinimalKnowledgeandNAF:02}, default \cite{baader+hollunder-DefaultsTerminRepresSys:95} and circumscriptive DL \cite{bonatti+lutz+wolter-NonMonDLCirc:06,Grimm+Hitzler-CircumscriptiveOWL:2008,grimm+hitzler-PreferentialTableauxCalculusALCO:09}. The first two are restricted to reasoning with explicitly named individuals, while \cite{Grimm+Hitzler-CircumscriptiveOWL:2008,grimm+hitzler-PreferentialTableauxCalculusALCO:09} allow for defeats to be based on the existence of unknown individuals. A tableau-based method for reasoning with the DL $\mathcal{ALCO}$ in the circumscriptive case has been introduced in \cite{Grimm+Hitzler-ReasoningCircumscriptiveALCO:2007}. A special preference clash condition is introduced there to distinguish between minimal and
non-minimal models which is based on constructing a new classical DL knowledge base and checking its satisfiability.

Datalog$^\pm$ \cite{Calì09FrameworkTractable,Calì_Gottlob_Lukasiewicz_2009UnifyingApproach} is an extension of Datalog which can simulate some DLs from the DL-Lite family \cite{dllite}. The extension consists in allowing a special type of rules with existentially quantified variables in the head, called tuple generating dependencies (TGDs). Note that our free rules are different from TGDs, as they allow for universally quantified variables which do not appear in the body of the rule to appear in the head.

The formalism is undecidable in the general case. Like in the case of OASP, several syntactical restrictions have been imposed on the shape of TGDs in order to regain decidability. Two such restrictions are: (1) every rule should have a guard, an atom which contains all variables in the rule body, giving rise to \emph{guarded Datalog$^\pm$}, and (2) every rule should have a singleton body atom, giving rise to \emph{linear Datalog$^\pm$}. The guardedness condition has been relaxed to \emph{weakly-guardedness}, where the weak guard has to contain only the variables in the body that appear in so-called affected positions, positions where newly invented values can appear during reasoning \cite{CaliGK08TamingChase}. Reasoning relies on a proof technique from database theory, the chase algorithm, which repairs databases according to the set of dependencies.

Some further generalizations to the guarded fragment of Datalog$^{\pm}$ are so-called \emph{sticky sets} of TGDs \cite{Cali+Gottlob+Pieris:AdvancedProcessingOntologicalQueries10}, \emph{weakly-sticky} sets of TGDS, and \emph{sticky-join} sets of TGDs \cite{2010_Cali_Query-Answering-Non-Guarded} which generalize both sticky sets and linear TGDs. All these fragments are defined by imposing restrictions on multiple occurrences of variables in rule bodies. The syntactical restrictions on rules bodies are orthogonal to the ones we imposed for achieving decidability on FoLPs: neither Datalog$^{\pm}$ rules are enforced to have a tree-shape like FoLPs, nor variables in FoLP rules have to fulfill the conditions required for the different sets of TGDs to belong to one of the previously mentioned decidable fragments of Datalog$^{\pm}$. TGDs do not contain negation. However, so-called stratified normal TGDs have been introduced, which are TGDs whose body atoms can appear in a negated form together with a semantics in terms of canonical models. FoLPs support full negation as failure (under the stable models semantics).

In the area of proof systems for Answer Set Programming, \cite{lin+you-AbductionLPNewDef:02} describes a goal rewrite system for brave reasoning under the stable model semantics which is sound and complete only for partial stable models. If the program has no odd loops (cycles in the predicate dependency graph of the program), its partial stable models and its stable models coincide. Note that such programs cannot have constraints as they are represented using rules in which a predicate depends negatively on itself. The problem with such rules is that they can render the program inconsistent, and thus, the rewriting, even if it is successful, is no longer valid. In our approach, we overcome this problem by going beyond the dependencies generated by the predicate checked to be satisfiable: we construct a complete answer set by taking care that the content of every node in the completion structure is saturated.
As concerns termination, \cite{lin+you-AbductionLPNewDef:02} distinguishes between positive, negative, odd, and even loops and deal with them accordingly. In terms of our approach, this amounts to checking for cycles in the dependency graph $G$ and identifying inconsistencies. However, for achieving termination, \cite{lin+you-AbductionLPNewDef:02} proposes to consider only ``domain restricted programs'', which can be instantiated only on domain predicates over variables which do not appear in the head. In our case, we do not have such a restriction: there are FoLPs (actually CoLPs) in which no constant appears and which still have infinite groundings. As such, we need the more complicated blocking mechanism for ensuring that there are no atoms with infinite justifications in the open answer set.

A resolution-based calculus for credulous reasoning in ASP which is sound for ground order-consistent programs and complete for ground finite recursive programs is introduced in \cite{bonatti+pontelli+son:credulousASP-08}. The calculus is extended to the nonground case, where it is proved to be sound for programs whose ground versions are order consistent, and complete for finitely recursive, odd-cycle free programs. In particular, the calculus is not sound for programs which have odd cycles, which are needed for simulating constraints. An extension for ground programs with constraints is provided, but no general solution is provided for the non-ground case. As already mentioned we have no problems in dealing with such constraints. Also the calculus is not complete for programs which are not finitely recursive, i.e., for programs for which there is at least a ground atom which depends on an infinite number of other ground atoms (w.r.t. the atom dependency graph of the grounded program). Our approach deals with programs which may not be finitely recursive: consider a FoLP which contains the rule $\prule{a(X)}{f(X, Y), a (Y)}$; grounding the program with an infinite universe leads to an infinite path in its atom dependency graph of the form $a(x_1), a(x_2), \ldots$.

A formalism related to FoLPs is $\mathbb{FDNC}$ \cite{simkus}. $\mathbb{FDNC}$ is an extension of ASP with function symbols where rules are syntactically restricted in order to maintain decidability. While the syntactical restriction is similar to the one imposed on FoLP rules, predicates having arity maximum two, and the terms in a binary literal can be seen as arcs in a forest (imposing the Forest Model Property), the direction of deduction is different: while for FoLPs, all binary literals in a rule body have an identical first term which is also the term which appears in the head, for $\mathbb{FDNC}$ (with the exception of one rule type) the second term is the one which also appears in the head. $\mathbb{FDNC}$ rules are required to be safe unlike FoLP ones. The complexity for standard reasoning tasks for $\mathbb{FDNC}$ is \exptime-complete and worst-case optimal algorithms are provided.

\cite{gebser-iclp06} introduces a system based on tableau methods for Answer Set Programming (ASP). Unlike in our case, where a clash-free complete completion structure represents an open answer set which satisfies a certain predicate, a branch in a tableau as described in \cite{gebser-iclp06} corresponds to a successful/unsuccessful computation of an answer set and an entire tableau represents a traversal of the search space. Note that in the case of FoLPs a computation of all models is not feasible as their number may be infinite. Also, the tableau calculi in \cite{gebser-iclp06} addresses only the propositional ASP case, as any ASP program can be grounded using only the constants present in the program, while in our case grounding is performed dynamically, introducing new individuals when needed.

\cite{Lierler-iclp08} describes an extension of an abstract framework for executing DPLL which computes supported models and stable models of a ground logical program. The framework employs a graph structure for encoding the different computation paths. Models are constructed in a bottom-up fashion: transition rules prescribe how new atoms are derived as being part/not being part of the model based on existing support/counter-support for such atoms. As such, there are similarities between these transition rules and our expansion rules which justify the presence/absence of unary/binary atoms in an open answer set. However, our expansion rules also have to introduce new elements in the domain and to perform grounding, and thus, they become much more complex. The abstract DPLL framework has also a nondeterministic choice rule which assigns the value true to a certain literal which is otherwise not constrained. This rule is similar in a sense with our Choose unary/binary expansion rules: while our approach is a top-down approach, and we are not interested in constructing models per se, it turned out to be necessary to construct a whole model for ensuring soundness of the approach.

\section{Conclusions and Outlook}
\label{sec:conclusions}

We introduced FoLPs, a logic programming paradigm suitable for integrating ontologies and rules, and provided a sound, complete, and terminating algorithm for satisfiability checking that runs in double exponential time. We showed how to use FoLPs as the underlying integration vehicle for reasoning with f-hybrid knowledge bases, a non-monotonic framework that integrates \SHOQ{} with FoLPs, without having to resort to (weakly) DL-safeness. We also introduced a restricted variant of FoLPs, simple FoLPs, which allow integration of \ALCHOQ{} knowledge bases with themselves and provided a sound, complete, and terminating algorithm for satisfiability checking that runs in exponential time.

From a theoretical perspective, the combination of stable model semantics and open domains posed specific challenges for our tableau-based algorithm: among these, were ensuring that every atom in the constructed model is finitely justified, and that the constructed model is part of an actual open answer set. In dealing with this, our approach differentiates from other existing approaches in the literature.

We are currently looking into extensions of FoLPs (and of the tableau algorithm) which would allow one to simulate DLs richer than \SHOQ{}, in the direction of $\mathcal{SROIQ}(\mathbf{D})$, the DL underlying OWL-DL\footnote{\url{http://www.w3.org/2007/OWL}} in OWL 2.

\bibliographystyle{acmtrans}


\newpage

\appendix
\begin{center}
    {\bf APPENDIX}
  \end{center}

\section{Additional Preliminaries}

A \emph{labeled tree}\index{labeled tree} is a pair $(T,t)$ where $T$ is a tree and $t: T \to \Sigma$ is a labeling function; sometimes we will identify the tree $(T,t)$ with $t$. For a labeled tree $t:T\to\Sigma$, the subtree of $t$ at $x\in T$ is $t[x]:T[x]\to \Sigma$ such that $t[x](y) = t(y)$ for $y\in T[x]$.


A labeled forest is a tuple $(F, f)$ where $F$ is a forest and $f:N_F \to \Sigma$ is a labeling function; sometimes we will identify the forest $(F,f)$ with $f$. A labeled forest $(F,f)$, with $F=\{T_c \mid c \in C\}$, induces a set of labeled trees $\{(T_c,t_c) \mid c \in C\}$, with $t_c: T_c \to \Sigma$ defined as follows: $t_c(x)=f(x)$, for any $x \in T_c$. Figure \ref{figure:forest} depicts a labeled forest which contains two labeled trees $t_a$ and $t_b$ (their roots are $a$ and $b$, respectively).

\vspace{10mm}
\begin{center}
\begin{figure}[htbp]
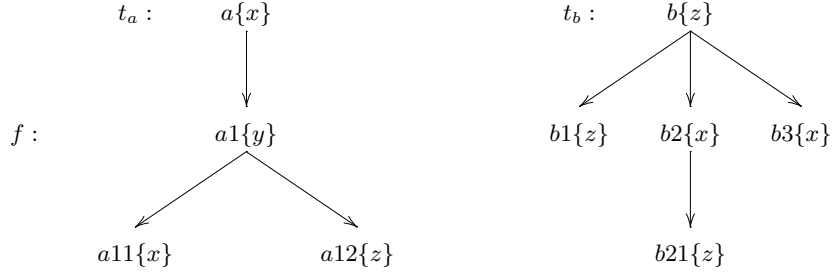

\Treek[1]{3}{&\K{$t_a:$} &\K{$a\{x\}$} \AR{d} &&& \K{$t_b:$}& \K{$b\{z\}$}\AR{d}\AR{dl}\AR{dr}&\\
\K{$f:$}&&\K{$a1\{y\}$} \AR{dl} \AR{dr} &&&\K{$b1\{z\}$}&\K{$b2\{x\}$}\AR{d}&\K{$b3\{x\}$}\\
&\K{$a11\{x\}$}&&\K{$a12\{z\}$}&&&\K{$b21\{z\}$}\\
}
\vspace{3mm}
\caption{A Simple Labeled Forest}
\label{figure:forest}
\end{figure}
\end{center}

A labeled extended forest is a tuple $\langle \EF, \mathit{ef} \rangle$ where $\EF$ is an extended forest and $\ef:N_{EF} \to \Sigma$ is a labeling function; sometimes we will identify the extended forest $\langle \EF, \ef\rangle$ with $\ef$. A labeled extended forest can be seen as a set of labeled extended trees, where a labeled extended tree is a tuple $(T^{\ef}, t^{\ef})$, where $T^{\ef}$ is an extended tree and $t^{\ef}: T^{\ef} \to \Sigma$ is a labeling function defined such that $t^{\ef}(x)=\ef(x)$, for $x \in T^{\ef}$. For a labeled extended tree $t^{\ef}:T^{\ef}\to\Sigma$, the subtree of $t^{\ef}$ at $x\in T$ is $t^{\ef}[x]:T^{\ef}[x]\to \Sigma$ such that $t^{\ef}[x](y) = t^{\ef}(y)$ for $y\in T^{\ef}[x]$.

Figure \ref{figure:labextforest} depicts an extended labeled forest (a labeled version of the extended forest from Figure \ref{figure:extforest}).

\begin{center}
\begin{figure}[htbp]
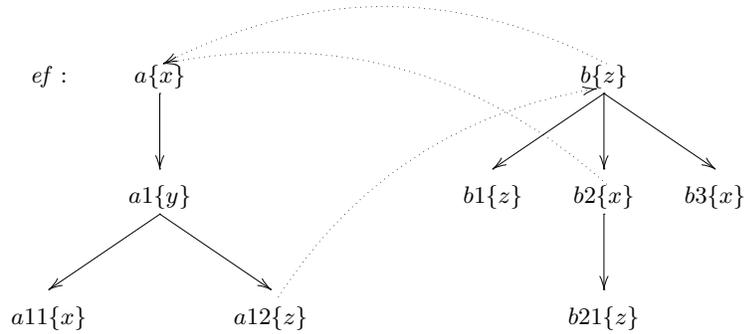

\vspace{5mm}
\Treek[1]{3}{\K{$\ef:$} &\K{$a\{x\}$} \AR{d} &&&& \K{$b\{z\}$}\AR{d}\AR{dl}\AR{dr}\UpLinkC[.>]{0,-4}&\\
&\K{$a1\{y\}$} \AR{dl} \AR{dr} &&&\K{$b1\{z\}$}&\K{$b2\{x\}$}\AR{d}\UpLinkC[.>]{-1,-4}&\K{$b3\{x\}$}\\
\K{$a11\{x\}$}&&\K{$a12\{z\}$}\LinkA[.>]{-2,3}&&&\K{$b21\{z\}$}\\
}
\vspace{3mm}
\caption{A labeled extended forest}
\label{figure:labextforest}
\end{figure}
\end{center}

We introduce the operation of replacing in a labeled extended forest $\ef$ an extended subtree $t^{\ef}[x]$ with another extended subtree $t^{\ef}[y]$, where both $x$ and $y$ are from $N_{\EF}$, and denote this operation with $replace_{\ef}(x, y)$. Figure \ref{fig:replaceoperation} describes the result of applying the replace operation on the extended forest from Figure \ref{figure:extforest} with two different sets of operators. In the first case, $t^{\ef}_b[b2]$ is replaced with $t^{\ef}_a[a1]$, while in the second case $t^{\ef}_a[a1]$ is replaced with $t^{\ef}_a[a12]$. Note that the names of nodes of the subtree which is replaced are not changed with the names of the nodes from the replacing subtree, but new names are generated for the new nodes in concordance with the naming scheme for nodes of that tree. Also, observe how in the first replacement one of the 'extra' arcs of $t_b$, $(b2,a)$, is dropped (it was part of the replaced extended subtree) and a new 'extra' arc is introduced, $(b22,b)$, which mirrors the arc $(a12,b)$ from the replacing extending subtree. Similarly, in the second transformation, $(a12,b)$ is dropped and $(a1,b)$ is introduced.

\begin{center}
\begin{figure}[htbp]
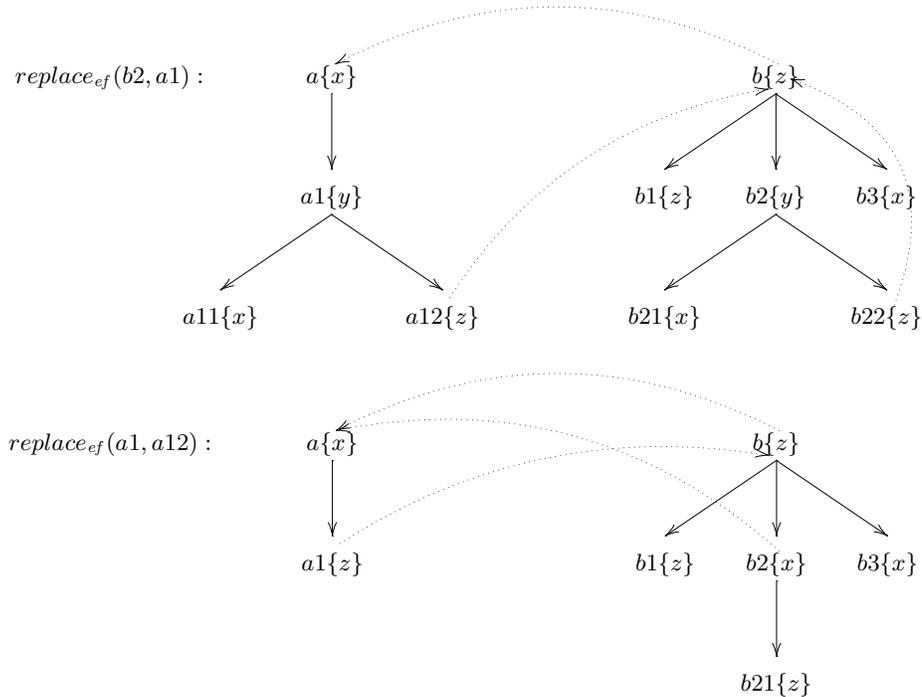

\vspace{5mm}
\Treek[1]{3}{\K{$replace_{\ef}(b2,a1):$} &&\K{$a\{x\}$} \AR{d} &&&& \K{$b\{z\}$}\AR{d}\AR{dl}\AR{dr}\UpLinkC[.>]{0,-4}&\\
&&\K{$a1\{y\}$} \AR{dl} \AR{dr} &&&\K{$b1\{z\}$}&\K{$b2\{y\}$}\AR{dr}\AR{dl}&\K{$b3\{x\}$}\\
&\K{$a11\{x\}$}&&\K{$a12\{z\}$}\LinkA[.>]{-2,3}&&\K{$b21\{x\}$}&&\K{$b22\{z\}$}\UpLinkF[.>]{-2,-1}}\\
\vspace{10mm}
\Treek[1]{3}{\K{$replace_{\ef}(a1,a12):$} &&\K{$a\{x\}$} \AR{d} &&&& \K{$b\{z\}$}\AR{d}\AR{dl}\AR{dr}\UpLinkC[.>]{0,-4}&\\
&&\K{$a1\{z\}$}\LinkA[.>]{-1,4} &&&\K{$b1\{z\}$}&\K{$b2\{x\}$}\AR{d}\UpLinkC[.>]{-1,-4}&\K{$b3\{x\}$}\\
&&&&&&\K{$b21\{z\}$}\\
}
\vspace{3mm}
\caption{Two applications of the replace operator on $\ef$}
\label{fig:replaceoperation}
\end{figure}
\end{center}

\section{Proofs}

\subsection{Soundness Proof}

\begin{proof}
From a clash-free complete completion structure for $p$ w.r.t. $P$, we  construct an open interpretation, and show that this interpretation is an open answer set of
$P$ that satisfies $p$. Let $\langle \EF,$ $\ct,$ $\st,$ $G \rangle$ be such a clash-free complete completion structure with $EF=\langle F,\ES \rangle$ the extended forest and $G=(V,A)$ the corresponding dependency graph and let $\bl$ be the set of blocking nodes corresponding to the completion.
\begin{enumerate}

\item \emph{Construction of open interpretation}.

We construct a new graph $\Gexte=(\Vexte,\Aexte)$ by extending $G$ in the following way: first, we set $\Vexte=V$ and $\Aexte=A$, and then for every pair $(x,y)\in \bl$ do the following:

\begin{itemize}
 \item (a) for every $p$ such that $p(x) \in V$, add $p(y)$ to $\Vexte$: $\Vexte=\Vexte \cup \{p(y)\}$;
\item  (b) for every $f$ and $z$ such that $f(x,z) \in V$, add $f(y,z)$ to $\Vexte$: $\Vexte=\Vexte \cup \{f(y,z)\}$;
\item (c) for every $p$, $q$ such that $(p(x), q(x)) \in \Aexte$,  add $(p(y), q(y))$ to $\Aexte$: $\Aexte=\Aexte \cup \{(p(y), q(y))\}$;
\item (d) for every $p$, $q$, $z$ such that $(p(x), q(z)) \in \Aexte$,  and $z \neq x$ add $(p(y), q(z))$ to $\Aexte$: $\Aexte=\Aexte \cup \{(p(y), q(z))\}$;
\item (e) for every $p$, $f$, $z$ such that $(p(x), f(x,z)) \in \Aexte$,  add $(p(y), f(y,z))$ to $\Aexte$: $\Aexte=\Aexte \cup \{(p(y), f(y,z))\}$;
\item (f) for every $f$, $q$, $z$ such that $(f(x,z),q(x)) \in \Aexte$,  add $(f(y,z),q(y))$ to $\Aexte$: $\Aexte=\Aexte \cup \{(f(y,z), q(y))\}$;
\item (g) for every $f$, $q$, $z$ such that $(f(x,z),q(z)) \in \Aexte$,  add $(f(y,z),q(z))$ to $\Aexte$: $\Aexte=\Aexte \cup \{(f(y,z), q(z))\}$;
\item (h) for every $f$, $g$, $z$ such that $(f(x,z),g(x,z)) \in \Aexte$,  add $(f(y,z),g(y,z))$ to $\Aexte$: : $\Aexte=\Aexte \cup \{(f(y,z), g(y,z))\}$;
\end{itemize}

Basically, this amounts to copying the content of the blocking node into the content of the blocked node, and also all the connections from/within the blocking node as connections from/within the blocked node (or, in other words, the content of the blocked node is identical with the content of the blocking node and it is justified in a similar way).

Let there be an open interpretation $(U,M)$, with $U=N_{\EF}$, i.e., the universe is the set of nodes in the extended forest, and $M=\Vexte$, i.e., the interpretation corresponds to the set of nodes in the extended graph.

\item \emph{$M$ is a model of $P_U^M$}. All free rules are trivially satisfied.

Take a ground unary rule: $r^{\prime}:\prule{a(x)}{\posi{\beta}(x), (\gamma_m^{+}(x,y_m), \delta_m^{+}(y_m))_{1\leq m\leq k}}$ from $P_{U}^{M}$ originating from $r:a(s)$ $\gets \beta(s),$ $(\gamma_m(s,t_m),$ $\delta_m(t_m))_{1\leq m \leq k},$ $\psi$, with $\nega{\beta}(x) \nsubseteq M$, for all $1\leq m\leq k$: $\nega{\gamma_m}(x,y_m) \nsubseteq M$ and $\nega{\delta_m}(y_m) \nsubseteq M$, and for all $t_i \neq t_j \in \psi$: $y_i \neq y_j$. Assume that $M \models \posi{\beta}(x) \cup \bigcup_{1\leq m\leq k} \gamma_m^{+}(x,y_m)\cup \bigcup_{1\leq m\leq k} \delta_m^{+}(y_m)$ (together with the assumptions about the negative part of the rule, this amounts to $M \models \beta(x) \cup \bigcup_{1\leq m<\leq k}\gamma_m(x,y_m) \cup \bigcup_{1\leq m \leq k}\delta_m(y_m) \cup \psi$) and $a(x) \notin M$ (the rule is not satisfied).

Depending on $x$ there are two cases:
\begin{itemize}
\item $x$ is not a blocked node. Then $\naf a \in \ct(x)$, $x$ is saturated, and no expansion rules can be further applied to $\naf a$. This means that for every ground rule derived from a rule $r \in P_a$ with head $a(x)$, the \textit{expand unary negative} rule has been applied. Such a rule is $r^\prime$. The application of the \textit{expand unary negative} rule to $\naf a \in \ct(x)$ and $r^\prime$ leads to one of the following situations:
\begin{itemize}
\item there is a unary predicate symbol $\pm q \in \beta$, such that $\mp q \in \ct(x)$ (the result of $update(\naf a(x),\mp q,x)$), or in other words, $\mp q(x) \in M$. This contradicts with $M \models \beta(x)$.
\item there are two successors of $x$, $y_i$ and $y_j$ such that $y_i = y_j$ and $t_i \neq t_j \in \psi$. This contradicts the assumption that for all $t_i \neq t_j \in \psi$: $y_i \neq y_j$.
\item for some $1 \leq m \leq k$, there is a binary/unary predicate symbol $\pm f \in \gamma_m$/$\pm q \in \delta_m$ such that $\mp f \in \ct(x,y_m)$/$\mp q \in \ct(y_m)$ (the result of $update($ $\naf a(x),$ $\mp f,(x,y_m))$ / $update(\naf a(x),\mp q,y_m)$), or in other words, $\mp f(x,y_m)$ $ \in $ $M/$ $\mp q$ $(y_m) \in M$. This contradicts with $M \models \gamma_m(x,y_m)$/$M \models \delta_m(y_m)$.
\end{itemize}
\item $x$ is a blocked node. Let $y$ be such that $(y,x) \in \bl$: by replacing $x$ with $y$ in $r^{\prime}$, one obtains a ground rule $r^{''}$ which again should not be satisfied because due to the construction of $M$, $M \models \beta(x) \cup \bigcup_{1\leq m<\leq k}\gamma_m(x,y_m) \cup \bigcup_{1\leq m \leq k}\delta_m(y_m) \cup \psi$ implies $M \models \beta(y) \cup \bigcup_{1\leq m<\leq k}\gamma_m(y,y_m)$ $\cup \bigcup_{1\leq m \leq k}\delta_m(y_m)$ $ \cup$ $ \psi$ and $a(x)\notin M$ implies $a(y) \notin M$. Thus, this case is reduced to the previous one.
\end{itemize}
Both cases lead to a contradiction, thus the original assumption that rule $r^{\prime}$ is not satisfied by $M$ was false. Thus, every unary rule is satisfied by $M$.

The proof for the satisfiability of binary rules is similar.

\item \emph{$M$ is a minimal model of $P_U^M$}. Before proceeding with the actual proof we introduce a notation and a lemma which will prove useful in the following. Let $\nEF$ be the directed graph $(N_{\EF}, \nAEF)$ which has as nodes all the nodes from $\EF$ and as arcs all the arcs of $\EF$ plus some 'extra' arcs which point from blocked nodes to successors of corresponding blocking nodes $\nAEF=A_{\EF} \cup \{(y,z) \mid \exists x \mbox{ s. t. } (x,y) \in \bl \wedge z \in succ_{\EF}(x)\}$. The new graph captures in a more accurate way the structure of $M$: blocked nodes are connected to successors of the corresponding blocking nodes, as their contents is justified similarly to the content of the blocking nodes. Figure \ref{figure:newextendedforest} exemplifies the construction of $\nEF$ from an extended forest $\EF$ by addition of extra arcs: $(x,y)$ is a blocking pair, $z_1, \ldots, z_n$, and $b$ are the successors of $x$, so extra arcs from $y$ to each of these successors are added (the dotted arrows). Among the successors of $x$ the one which is on the same path with $y$ is singled out and denoted with $z$.


\begin{center}
\begin{figure}[htbp]
\Treek[0]{3}{ & & & \K{$a$} \ARdashc{d} \ARdashc{drr} \ARdashc{dll}&&&&&&\K{$b$} \ARdashc{d} \ARdashc{drr}\ARdashc{dll}\\
&&&\K{$x$} \ARc{d} \ARc{drr} \ARc{dll} \UpLinkG{-1,6}&&&&\UpLinkH{-1,-4}&&&&\\
&\K{$z_1$} \ARdashc{d} \ARdashc{dl}& \K{$\ldots$} & \K{$z$}
\ARdash{dd} \ARdashc{dl} \ARdashc{dr}
& \K{$\ldots$} &
\K{$z_n$}
\ARdashc{d} \ARdashc{dr} &&\\
&&&&&&\\
&&&\K{$y$} \UpLinkCD[.>]{-2,-2}\UpLinkCD[.>]{-2,0}\UpLinkCD[.>]{-2,2}\UpLinkG[.>]{-4,6} \\
}
\caption{Constructing $\nEF$: $(x,y)$ is a blocking pair}
\label{figure:newextendedforest}
\end{figure}
\end{center}

\begin{lemma}\label{lemma:connection}
For every $x, y \in N_{\EF}$, if there is a path $Pt_1=(p(x), \ldots, l_1) \in paths_G / paths_{G_{ext}}$, with $l_1=q(y)$ for some $q \in \upreds{P}$ or $l_1=g(y,z)$ for some $g \in \bpreds{P}$, and $x \neq y$, then there is a path $Pt_2=(x, \ldots, y) \in paths_{\EF}/ paths_{\nEF}$ such that for every $z \in Pt_2$ there is a unary atom $l_2 \in Pt_1$ with $args(l_2)=z$.
\end{lemma}
\begin{proof}

Let $S=(x_1=x, x_2, \ldots, x_n)$ be a tuple of nodes from $\EF/ \nEF$ constructed in the following way: consider each element $l$ of $Pt_1$ at a time: if $args(l)=y$ and $y$ is not already part of the tuple, add $y$ to the tuple. We show that $S \in paths_{\EF} / paths_{\nEF}$ and furthermore that $x_n=y$.

For every two consecutive elements of $S$, $x_i$ and $x_{i+1}$, with $1 \leq i <n$, there must be two unary atoms $l'$ and $l''$ in $Pt$, with $args(l')=x_i$ and $args(l'')=x_{i+1}$, respectively, such that there is no other unary atom $l$ in the sub-path of $Pt_1$: $(l', \ldots, l'')$. It is easy to see that such a sub-path has the form: $(l'=r(x_i), f_1(x_i,x_{i+1}), \ldots, f_m(x_i,x_{i+1}), l''=s(x_{i+1}))$, with $r, s \in \upreds(P)$, and $f_1, \ldots f_m \in \bpreds(P)$, and thus $(x_i, x_{i+1}) \in A /  \nAEF$ for every $1 \leq i <n$: $(x_1, \ldots, x_n)$ is a path in $\EF/ \nEF$.

To see that $x_n=y$, consider the opposite: $x_n \neq y$. Then there must be a unary atom $l=r(x_n)$ in $Pt_1$ with $args(l)=x_n$ such that there is no other unary atom in the sub-path of $Pt_1$: $(r(x_n), \ldots, g(y,z))$. This would imply that the sub-path has the form $r(x_n), f_1(x_n,t), \ldots, f_m(x_n,t), g(y,z)$, where $t$ is some successor of $x_n$ in $\EF/ \nEF$: $(x_n, t) \in A / \nAEF$. But there is no arc of the form $(f_m(x_n,t), g(y,z))$ in $A / \nAEF$ with $x_n \neq y$, so we obtain a contradiction.


\end{proof}

Now we can proceed to the actual proof of statement. Assume there is a model $M^{\prime} \subset M$ of $Q=P_U^M$. Then $\Exists{l_1\in M}{l_1 \notin M^{\prime}}$. Take a rule $r_1 \in Q$ of the form $\prule{l_1}{\beta_1}$ with $M \models \beta_1$; note that such a rule always exists by construction of $M$ and expansion rule (i) . If $M^{\prime} \models \beta_1$, then $M^{\prime} \models l_1$ (as $M^{\prime}$ is a model), a contradiction. Thus, $M^\prime \not\models \beta_1$ such that $\Exists{l_2 \in \beta_1}{l_2 \notin M^\prime}$. Continuing with the same line of reasoning, one obtains an infinite sequence $\{l_1, l_2,\ldots\}$ with $(l_i \in M)_{1 \leq i}$ and $(l_i \notin M^\prime)_{1 \leq i}$. $M$ is finite (the complete clash-free completion structure has been constructed in a finite number of steps, and when constructing $M$($\Vexte$) we added only a finite number of atoms to the ones already existing in $V$), thus there must be $1 \leq (i,j)$, $i \neq j$, such that $l_i=l_j$. We observe that $(l_i, l_{i+1})_{1 \leq i} \in \Eexte$ by construction of $\Eexte$ and expansion rule (i), so our assumption leads to the existence of a cycle in $\Gexte$.


\begin{claim}\label{claim:directcycle}
Let $C=(l_1, l_2, \ldots, l_n=l_1)$  be a cycle in $\Gexte$. If one of the following holds:
\begin{itemize}
 \item (i) there is no unary atom in $C$ and for every $l_i=f_i(x_i,y_i)$, $1 \leq i \leq n$, $x_i$ is not blocked
 \item (ii) there is at least one unary atom in $C$ and for every unary atom in $C$: $l_{j}$ with $args(l_j)=y_j$, $y_j$ is not a blocked node in $CS$, $1 \leq j \leq n$.
\end{itemize}
then  $C$ is a cycle in $G$.
\end{claim}

\begin{proof}
From the construction of $\Gexte$ one can see that any arc which is added to $G$ is of the form $(p(x),l)$ or $(f(x,y),l)$, where $p$ is some unary predicate, $f$ is some binary predicate, and $x$ is a blocked node.
It is clear that when condition (i) or condition (ii) holds there is no arc of the first form in $C$. As concerns arcs of the latter type, it is again obvious that there are no such arcs if condition (i) is fulfilled. In case condition (ii) holds, assume there is an arc $(f(x,y), l)$ where $x$ is a blocking node. We know that there must be at least one unary atom in the cycle. Let this be $p(z)$. In this case there is a path in $G$ (and also in $\Gexte$) from $p(z)$ to $f(x,y)$ and $z$ is different from $x$ by virtue of (ii). According to lemma \ref{lemma:connection} this path contains a unary atom with argument $x$ (as any path in $\EF$ from $z$ to $x$ contains $x$). However this contradicts with condition (ii) which says that there is no such atom in $C$.
\end{proof}

\begin{claim} \label{claim:inducedcycle}
Let $C=(l_1, l_2, \ldots, l_n=l_1)$  be a cycle in $\Gexte$. If one of the following holds:
\begin{itemize}
 \item (i) there is no unary atom in $C$ and for some $l_i=f_i(x_i,y_i)$, $1 \leq i \leq n$, $x_i$ is blocked
 \item (ii) there is at least one unary atom in $C$ and all unary atoms have the same argument $y$ which is a blocked node
\end{itemize}
then $G$ contains a cycle.
\begin{proof}
We will treat the two cases separately:

(i) First, notice that in this case (when there is no unary atom in the cycle), $args(l_1)=args(l_2)=\ldots=args(l_n)=(x,y)$ as there is no arc in $\Aexte$ from a binary atom $f(x,y)$ to another binary atom $g(z,t)$, with $x \neq z$ or $y \neq t$ (by construction of $\Gexte$). So the cycle can be written as $C=(f_1(x,y), f_2(x,y), \ldots, f_n(x,y)=f_1(x,y))$, where $(f_i \in \bpreds{P})_{1 \leq i \leq n}$. Let $z$ be the blocking node corresponding to $x$: $(z,x)\in \bl$. As $((f_i(x,y), f_{i+1}(x,y)) \in \Aexte)_{1 \leq i <n}$, it follows that $((f_i(z,y), f_{i+1}(z,y)) \in A)_{1 \leq i <n}$, so $C'=(f_1(z,y), f_2(z,y), \ldots, f_n(z,y)=f_1(z,y))$ is a cycle in $G$.

(ii) Let $p_1(y), p_2(y), $ $\ldots, p_n(y)$ be the unary atoms in $C$ with $y$ being a blocked node. Without loss of generality we consider $p_n=p_1$. Then the cycle can be written as: $C=(p_1(y),$ $f_{11}(y,z_1), \ldots,$ $f_{1m_{1}}(y,z_1),$ $p_2(y),$  $f_{21}$ $(y,z_2),$  $\ldots, f_{2m_{2}}(y,z_2)), \ldots p_n(y)=p_1(y)$ where $(f_{ij}\in \bpreds{P})_{1 \leq i <n, 1\leq j \leq m_i}$, $((y, z_i) \in \nAEF)_{1 \leq i <n}$ (as the only binary atoms reachable from $p(y)$ are of the form $f(y,z)$, where $(y,z) \in \nAEF$). Similar with the previous case one can show that $C'=(p_1(x), f_{11}(x,z_1), \ldots, f_{1m_{1}}(x,z_1), p_2(x), f_{21}(x,z_2),$  $\ldots, f_{2m_{2}}(x,z_2)), \ldots p_n(x)=p_1(x)$, where $x$ is the corresponding blocking node for $y$: $(x,y) \in \bl$ is a cycle in $G$.
\end{proof}
\end{claim}

\begin{claim} \label{claim:blockingviolation}
Let $C=(l_1, l_2, \ldots, l_n=l_1)$  be a cycle in $\Gexte$. If there are at least two unary atoms in $C$ with different arguments and at least one unary atom has as argument a blocked node $y$ then there is a path in $G$ from an atom $l_1$ to an atom $l_2$ where $args(l_1)=x$, $args(l_2)=y$, and $x$ is the corresponding blocking node for $y$: $(x,y) \in \bl$.
\end{claim}

\begin{proof}
Let $t$ be the argument of a unary atom in the cycle different from $y$. As there is a path in $\Gexte$ from some $p(t)$ to some $q(y)$ and also viceversa from some $q(y)$ to some $p(t)$ according to lemma \ref{lemma:connection} there must also be a path in $\nEF$ from $t$ to $y$ and a path from $y$ to $t$. In other words there exists a cycle in $\nEF$ which involves both $y$ and $t$. Furthermore for every element of the cycle in $\nEF$, there is a unary atom in $C$ which has this element as an argument.
From the way $\nEF$ was constructed (see also Figure \ref{figure:newextendedforest}), one can see that any cycle in $\nEF$ which involves a blocked node $y$  which makes part from a tree $T$ in the corresponding simple forest contains the path in $T$ from $z$ to $y$, where $z$ is the node which is a successor of $x$  in $T$, and is on the same path in $T$ as $x$ and $y$, $x$ being the corresponding blocking node for $y$: formally, $(x,y) \in \bl$, $z \in succ_T(x)$, $z \in path_T(x,y)$. There are two kinds of cycles in $\nEF$:
\begin{itemize}
\item cycles which contain $x$, $z$, and  $y$ (these cycles will contain also elements from other trees than $T$): in this case there is a unary atom $l_1$ with argument $x$ in $C$ and there is as well a unary atom $l_2$ with argument $y$ in $C$ (from the condition of the claim) - so the claim is satisfied
\item cycles which contain $z$, and $y$, but do not contain $x$ (actually, this is a unique such cycle which has all elements from $path_T(z,y)$): in this case there are two unary atoms $l_2$, and $l_3$ in $C$, with arguments $y$, and $z$ respectively, such that there is no other unary atom on the path induced by $C$ in $\Gexte$ from $l_2$ to $l_3$. In this case this path has the form: $p(y), f_1(y,z), \ldots, f_n(y,z), q(z)$. Due to the construction of $\Gexte$, the existence of the path $(p(y), f_1(y,z), \ldots, f_n(y,z), q(z))$ in $\Gexte$ implies the existence of the path $(p(x), f_1(x,z), \ldots, f_n(x,z), q(z))$ in $G$. At the same time note that there is a path in $G$ from $q(z)$ to $p(y)$. So, $(p(x), q(z)) \in conn_G$ and $(q(z), p(y)) \in conn_G$, thus $(p(x), p(y)) \in conn_G$  and the claim is satisfied.
\end{itemize}
\end{proof}

One can see that the hypotheses of the three claims cover all possible types of cycles $C$ in $\Gexte$ and that the consequences of having such a cycle are contradicting in each case with the fact that $\langle \EF,$ $\ct,$ $\st,$ $G,$ $\bl \rangle$ is a complete clash-free completion structure (in the case of the first two claims, one obtains that there must be a cycle in $G$, while the conclusion of the third claim contradicts with the blocking condition for a pair of blocking nodes from $\bl$). Thus, there cannot be such a cycle $C$ in $\Gexte$ and $M$ is minimal.
\end{enumerate}
\end{proof}

\subsection{Completeness Proof}
\begin{proof}
If $p$ is satisfiable w.r.t. $P$ then $p$ is forest-satisfiable w.r.t. $P$ (Proposition \ref{prop:forestmodelproperty1}). We construct a clash-free complete completion structure for $p$ w.r.t. $P$, by guiding the non-deterministic application of the expansion rules with the help of a forest model of $P$ which satisfies $p$ and by taking into account the constraints imposed by the saturation, blocking, redundancy, and clash rules. The proof is inspired by completeness proofs in Description Logics for tableau, for example in \cite{horrocks99practical}, but requires additional mechanisms to eliminate redundant parts from Open Answer Sets.

In order to proceed we need to introduce the notion of \emph{relaxed completion structure} which is a tuple $\langle \EF,$ $\ct,$ $\st,$ $G,$ $\bl \rangle$, where $\EF$ is an extended forest, and  $G$, $\ct$, $st$, $bl$ represent the same kind of entities as their homonym counterparts in the definition of a completion structure. An \emph{initial relaxed completion structure for checking satisfiability of a unary predicate $p$ w.r.t. a FoLP $P$} is defined similarly as an initial completion structure for checking satisfiability of $p$ w.r.t. $P$. A relaxed completion structure is evolved using the expansion rules (i)-(vi) and the applicability rules (vii)-(viii). Note that the \emph{redundancy} rule is left out. A complete clash-free relaxed completion structure is a relaxed completion structure evolved from an initial relaxed completion structure for $p$ and $P$, such that no expansion rules can be further applied, which is not contradictory and for which $G$ does not contain positive cycles.

The first step of the proof consists in constructing a complete clash-free relaxed completion structure starting from a forest model of a FoLP $P$ which satisfies $p$. Note that in the general case, constructing a complete clash-free relaxed completion structure might be a non-terminating process (the termination for the construction of complete clash-free completion structures was based on the application of the redundancy rule), but as we will see in the following, the process does terminate when a forest model is used as a guidance.

So, let $(U,M)$ be an open answer set of a FoLP $P$ which satisfies $p$ which at the same time is a forest model of $P$. Then there exists an extended forest $\EF = \langle \set{T_{\roo}} \cup \set{T_a \mid a \in \cts{P}},\ES \rangle$, where $\roo$ is a constant, possibly one of the constants appearing in $P$, and a labeling function $\mathcal{L}:\set{T_{\roo}} \cup \set{T_a \mid a \in \cts{P}}\cup A_{\EF} \to 2^{\preds{P}}$ which fulfill the conditions from definition \ref{def:forest-sat1}.


We define an initial relaxed completion structure $\CS_0=\langle \EF^{\prime},$ $\ct,$ $\st,$ $G,$ $\bl\rangle$  for $p$ and $P$ such that $\EF^{\prime}= \langle F^{\prime},\ES^{\prime} \rangle$, $F^{\prime}=\set{T^{\prime}_{\roo}} \cup \set{T^{\prime}_a \mid a \in \cts{P}}$, where $\roo$ is the same $\roo$ used to define $\EF$, and $T_x=\set{x}$, for every $x \in \cts{P} \cup \set{\roo$}, and  $\ES^{\prime}=\emptyset$, $G=\grap{V}{A}$, $V=\{p(\roo)\}$, $A=\emptyset$, and $\ct(\roo)=\{p\}$, $\st(\roo,p)=\unexpa$, $\bl=\emptyset$. We will  evolve this completion structure using rules (i)-(viii). To this purpose we inductively define a function $\pi:N_{\EF^{\prime}} \to U $ that relates nodes in the relaxed completion structure to nodes in the forest model satisfying the following properties:

\[
\ddag
\begin{cases}
\{q \mid q\in \ct(z)\} \subseteq \mathcal{L}(\pi(z)), \mbox{ for all }z \in N_{\EF^{\prime}}\ \\
\{q \mid \naf{q}\in \ct(z)\} \cap \mathcal{L}(\pi(z)) = \emptyset, \mbox{ for all }z\in N_{\EF^{\prime}}\ \\
\end{cases}
\]

Intuitively, the positive content of a node/edge in the completion structure is contained in the label of the corresponding forest model node, and the negative content of a node/edge in the completion structure cannot occur in the label of the corresponding forest model node.

\begin{claim}
Let $\CS$ be a relaxed completion structure derived from $\CS_0$ and $\pi$ a function that satisfies (\ddag). If an expansion rule is applicable to $\CS$ then the rule can be applied such that the resulting relaxed completion structure $\CS'$ and an extension $\pi'$ of $\pi$ still satisfies (\ddag).
\end{claim}

We start by setting $\pi(x) = x$, for every $x \in \cts{P} \cup \{\roo\}$ (the roots of the trees in the relaxed completion structure correspond to the roots of the trees in the forest model). It is clear that (\ddag) is satisfied for $\CS_0$.  By induction let $\CS$ be a relaxed completion structure derived from $\CS_0$ and $\pi$ a function that satisfies (\ddag).
We consider the expansion rules and the applicability rules saturation and blocking:
\begin{enumerate}

\item \emph{Expand unary positive}. As $q\in \ct(x)$, we have, by the
induction hypothesis, that $q\in \Lc{\pi(x)}$.  Since $M$ is a minimal model there is an $r\in P_q$ of the form (\ref{eq:unary}) and a ground version $r':q(\pi(x))\gets \posi{\beta}(\pi(x)), (\gamma_m^{+}(\pi(x),z_m))_{1 \leq m \leq k}, (\delta_m^{+}(z_m))_{1 \leq m \leq k} \in (P_q)_U^M$ such that $M \models \posi{\beta}(\pi(x)) \cup (\gamma_m^{+}(\pi(x),z_m))_{1 \leq m \leq k} \cup (\delta_m^{+}(z_m))_{1 \leq m \leq k}$. Set $\rl(q,x) = r$ and $\update(q(x),\beta, x)$. Next, for each $1 \leq m \leq k$:
	\begin{itemize}
        \item If $z_m = \pi(z)$ for some $z$ already in $\EF^{\prime}$, take $y_m = z$; also, if $z \in cts(P)$ and $(x,z) \notin ES^{\prime}$ then $ES^{\prime}=ES^{\prime} \cup \{(x,z)\}$,
        \item if $z_m=\pi(x) \cdot s$ and $z_m$ is not yet the image of $\pi$ of some node in $\EF^{\prime}$, then add $x\cdot s$ as a new successor of $x$ in $F^{\prime}$: $T^{\prime}_c=T^{\prime}_c \cup \{x \cdot s\}$, where $x \in T^{\prime}_c$, set $\pi(x \cdot s) = \pi(x) \cdot s$ and $\pi(x, x \cdot s) = (\pi(x),\pi(x) \cdot s)$.
    \item $\update(q(x), \gamma_m, (x,y_m))$,
    \item $\update(q(x), \delta_m, y_m)$.
    \end{itemize}

 In other words we removed the nondeterminism from the \emph{expand unary positive rule}, by  choosing the rule $r$ and the successors corresponding to the open answer set $(U,M)$. One can verify that (\ddag) still holds for $\pi$.
 \item One can deal with the rules (ii-vi) in a similar way, making the
 non-deterministic choices in accordance with $(U,M)$.

\item \emph{Saturation}. No expansion rule can be applied on a node from $\EF^{\prime}$ which is not a constant until its predecessor is saturated. This rule is independent of the particular open answer set which guides the construction, so it is applied as usually.

\item \emph{Blocking}. Consider a node $x \in N_{\EF^{\prime}}$ which is selected for expansion. If there is a saturated node $y \in N_{\EF^{\prime}}$ which is not a constant, $y <_{T_c} x$, where $T_c \in F^{\prime}$,  $\ct(x) \subseteq \ct(y)$, and $connpr_{G}(y,x)=\emptyset$ then $x$ is blocked and $(y,x)$ is added to the set of blocking pairs: $\bl=\bl \cup \{(y,x)\}$. Furthermore, we impose that if there are more nodes $y$ which satisfy the condition we will consider as the blocking node for $x$ the one which is closest to the root of the tree $T_c$ (the tree from which $x$ makes part), so the node $y$ for which there is no node $z$ such that $z<_{T_c} y$, $\ct(x) \subseteq \ct(z)$, and $connpr_{G}(z,y)=\emptyset$. This choice over possible blocking nodes is relevant for the next stage of the proof, where a complete clash-free relaxed completion structure is transformed into a complete clash-free completion structure. The condition (\ddag) still holds for $\pi$ as we have not modified the content of nodes, but just removed some unexpanded nodes.

\end{enumerate}

So, (\ddag) holds for $CS'$ which was evolved from $CS$, no matter which expansion rule or applicability rule was used. It is easy to see, that if (\ddag) holds for a particular relaxed completion structure $CS$ then this fact together with the fact that $(U, M)$ is an open answer set of $P$ guarantees that $CS$ is clash-free. So, in order to obtain a complete clash-free relaxed completion structure one has just to apply rules (i-viii) in the manner described above. To see that the process terminates, assume it does not. Then, for every $x, y \in N_{\EF^{\prime}}$ such that $x<_F^{\prime}y$ and $\ct(x)=\ct(y)$, the blocking rule cannot be applied, so there is a path from a $p(x)$ to some $q(y)$. This suggests the existence of an infinite path in $G$ (as on any infinite branch in a tree from $F^{\prime}$ there would be an infinite number of nodes with equal content - there is a finite amount of values for the content of a node), which contradicts with the fact that any atom in an open answer set is justified in a finite number of steps\cite[Theorem 2]{heymans-amai2006}.

At this point we have constructed a complete clash-free relaxed completion structure $CS$ for $p$ w.r.t $P$ starting from a forest open answer set for $P$ which satisfies $p$.

The preference relation over different blocking nodes choices in the construction above has several consequences described by the following results:

\begin{lemma}\label{lemma:blockingnodefirst}
 Let $CS=\langle \EF,$ $\ct,$ $\st,$ $G$ $\bl \rangle$ be a complete clash-free relaxed completion structure constructed in the manner described above ($\EF= \langle \F,\ES \rangle$). Then, for every $x$ such that there exists a $y$ so that $(x,y) \in \bl$ ($x$ is a blocking node in $CS$), there is no node $z<_{T_c} x$, $T_c \in F$ such that $\ct(z)=\ct(x)$.
\end{lemma}

\begin{proof}
 Assume by contradiction that $x$ is a blocking node in $CS$, so, there is a $y$ such that $(x,y) \in \bl$, and that there exists also $z<_{T_c} x$, $T_c \in F$ such that $\ct(z)=\ct(x)$. Observe that $conn_G(z,y)=\{(p(z), q(y))\mid p \in \ct(z) \wedge q \in \ct(y) \wedge (\exists r \in \ct(x) \mbox{ s. t. }(p(z), r(x)) \in conn_G(z,x)\wedge(r(x), q(y)) \in connpr_G(x,y))\}$ (according to lemma \ref{lemma:connection} the existence of a path from a $p(z)$ to a $q(y)$ in $G$ implies the existence of a path from $z$ to $y$ in $\EF$; all paths from $z$ to $y$ in $\EF$ include the path from $z$ to $y$ in $T_c$ and conversely $x$, and then according to the same lemma there must be a atom in the initial path in $G$ with argument $x$: $r(x)$ in this case). But $connpr_G(x, y)=\emptyset$ as $(x,y) \in \bl$, so $connpr_G(z, y)=\emptyset$. Additionally, $\ct(z)=\ct(x)\supseteq\ct(y)$, so the existence of $z$ is in contradiction with the preference condition over potentially blocking nodes. Thus, the lemma holds.
\end{proof}

\begin{corollary}\label{corollary:limitblockingnodes}
Let $CS=\langle \EF,$ $\ct,$ $\st,$ $G,$ $\bl \rangle$ be a complete clash-free relaxed completion structure constructed in the manner described above ($\EF=\langle E,\ES \rangle$) and $IB$ a branch of a tree $T_c$ from $F$. Then there are at most $2^p$ distinct blocking nodes in $IB$ where $p=|\upreds{P}|$.
\end{corollary}

\begin{proof}
The result follows from the fact that there cannot be two blocking nodes with equal content on the same path in a tree according to the previous
lemma and the finite number of values for the content of a node which is given by the cardinality of the power set of $\upreds{P}$.
\end{proof}

The next step is to transform a relaxed clash-free complete completion structure $CS=\langle \EF,$ $\ct,$ $\st,$ $G,$ $\bl \rangle$, where $\EF=\langle \F, \ES \rangle$, into a complete clash-free completion structure, that is, a complete clash-free relaxed completion structure which has no redundant nodes. This is done by applying a series of successive transformations on the relaxed completion structure - each transformation ``shrinks'' the completion structure in the sense that the newer returned relaxed completion structure has a lesser number of nodes than the original one and is still complete and clash-free. The result of applying the transformation is a relaxed clash-free complete completion structure which has a bound on the number of nodes on any branch which matches the bound $k$ from the redundancy condition, which is thus a clash-free complete completion structure.

A way to shrink a (relaxed) completion structure is that whenever two nodes $u$ and $v$ in a tree $T_c$ from $F$ are on the same path, $u<_{T_c} v$, and they have equal content, $\ct(u)=\ct(v)$, the subtree $T_c[u]$ is replaced with the subtree $T_c[v]$. We call such a transformation $collapse_{CS}(u,v)$ and its results is a new relaxed completion structure $CS'=\langle \EF',$ $\ct',$ $\st',$ $G',$ $\bl' \rangle$, where the elements of this new completion structure are defined in the following.

Let $ef:N_{\EF} \to C$ be a labeled extended forest which associates to every node of $\EF$ a label from a set of distinguished constants $C$ such that $\ef(x) \neq \ef(y)$ for every $x$ and $y$ in $N_{\EF}$ such that $x \neq y$. Let $\ef'=replace_\ef(u,v)$ be a new labeled extended forest and $\EF'$ be the corresponding unlabeled extended forest. For every $x \in \EF'$ let $\overline{x}$ be the counterpart of $x$ in $\EF$ in the sense that: $\ef'(x)=\ef(\overline{x})$. Note that for every $x \in \EF'$ there is a unique such counterpart in $\EF$. For simplicity we also introduce the notation $\overline{S}$ to refer to the counterpart tuple (the tuple of counterpart nodes) corresponding to the tuple of nodes from $S$ from $T'$ . Formally, $\overline{(x_1,\ldots,x_n)}=(\overline{x_1}, \ldots, \overline{x_n})$. With the help of this notion of counterpart node we will define also the other components of the resulted completion structure ($\EF'$ has already been defined):
\begin{itemize}
\item $G'=(V', A')$. The set of nodes $V'$ of the new graph $G'$ contains all atoms $l$ for which there is a atom in $V$ formed with the same predicate symbol as $l$ and having as arguments the counterpart of the arguments of $l$. Additionally, $V'$ contains binary atoms which connect the predecessor of $u$ (it is the same both in $\EF$ and $\EF'$) with the new node $u$ which were also present in $V$ - this is necessarily as $\overline{u}=v$, so otherwise these connections would be lost:
      \begin{align*}
      V'=&\{l_1 \mid \exists l_2 \in V \mbox{ s. t. } pred(l_1)=pred(l_2) \wedge \overline{args(l_1)}=args(l_2) \} \cup \\
      &\{f(z,u) \mid z \in T' \wedge f(z,u) \in V\}.
      \end{align*}
      The set of arcs $A'$ of the new graph $G'$ contains all pair of atoms $(l_1,l_2)$ for which there is a corresponding pair in $E$, $(l_3,l_4)$, such that $l_3$ and $l_4$  have the same predicate symbols as $l_1$ and $l_2$, respectively, and their argument tuples are the counterpart of the argument tuples of $l_1$, and $l_2$, respectively. Additionally, $A'$ contains arcs from $A$ which connect atoms whose arguments include the predecessor of $u$ (it is the same both in $T$ and $T'$) with atoms whose arguments include the new node $u$ - this is necessarily as $\overline{u}=v$, so otherwise these connections would be lost:
      \begin{align*}
      A'=&\{(l_1,l_2) \mid \exists (l_3, l_4) \in A \mbox{ s. t. } pred(l_1)=pred(l_3) \wedge pred(l_2)=pred(l_4) \\
      &\qquad \qquad\wedge \overline{args(l_1)}=args(l_3) \wedge \overline{args(l_2)}=args(l_4)\} \cup \\
      &\{(l_1, l_2) \mid (l_1, l_2) \in E \wedge u \in arg(l_2) \wedge z \in arg(l_1) \wedge z<u\}.
      \end{align*}
\item $\ct'(x)=\ct(\overline{x})$, for every $x \in \ef'$;
\item $\st'(x)=\st(\overline{x})$, for every $x \in \ef'$;
\item $\bl'=\{(x,y) \mid (\overline{x}, \overline{y}) \in \bl \wedge connpr_{G'}(x,y)=\emptyset\}$. We maintain those blocking pairs whose counterparts in $\EF$ formed a blocking pair, and which further more still fulfill the blocking condition.
\end{itemize}

Note that the result of applying the transformation on a complete clash-free relaxed completion structure might be an incomplete clash-free relaxed completion structure. If completeness of the original structure was achieved by applying among others the blocking rule, the transformation might leave some branches ``unfinished'' in case the blocking node is eliminated or simply because two nodes who formed a blocking pair are still found in the new structure, but they do not longer fulfill the blocking condition. We will describe two cases in which the transformation can be applied without losing the completeness of the resulted structure by means of two lemmas. Before that, however, we need to state a general result which will prove useful in the demonstration of the two lemmas. The result states that if as a result of applying the $collapse$ transformation on a complete clash-free relaxed completion structure one obtains a completion structure in which the path between a blocking pair of nodes remains untouched (every node in the original path is the counterpart of some node in the new structure), then the nodes which have as counterparts the nodes of the blocking pair form a blocking pair in the new completion structure.

\begin{lemma}\label{lemma:integrityblockingpath}
Let $CS=\langle \EF,$ $\ct,$ $\st,$ $G,$ $\bl \rangle$, $\EF=\langle \F, \ES\rangle$ be a complete clash-free relaxed completion structure and $CS'=\langle \EF',$ $\ct',$ $\st',$ $G',$ $\bl' \rangle$ the result returned by $collapse_CS(u,v)$, where $u$ and $v$ are two nodes from $\EF$ which fulfill the usual conditions necessary for the application of $collapse$. Then, for every $(x, y) \in bl$: if for every $z \in path_{T_c}(x,y)$ ($x,y \in T_c)$, exists $z' \in \EF'$ such that $\overline{z'}=z$, then $(x', y') \in bl'$, where $x', y' \in \EF'$, $\overline{x'}=x$ and $\overline{y'}=y$.
\end{lemma}
\begin{proof}
Let $\EF$, $\EF'$, $x$, $y$, $x'$, and $y'$ be as defined in the lemma. The conditions for the two nodes $x'$ and $y'$ from $\EF'$ to form a blocking pair: $(x',y') \in \bl'$, are that $(\overline{x}, \overline{y}) \in \bl$ and $connpr_{G'}(x',y')=\emptyset$. The first condition is part of the prerequisites of the lemma, so it remains to be proved that $connpr_{G'}(x',y')=\emptyset$. Assume by contradiction that there exists a path in $G'$  from a $p(x')$ to a $q(y')$. Then according to lemma \ref{lemma:connection} there is a path $Pt$ in $\EF'$ from $x'$ to $y'$ such that for every $z\in P$ there exists a unary atom with argument $z$ in the path in $G'$ from $p(x')$ to $q(y')$. Any path in $\EF'$ from $x'$ to $y'$ includes the path in $T'_c$ (the tree from which both $x'$ to $y'$ make part) from $x'$ to $y'$. Assume $path_{T'_c}(x', y')=(x{_1}'=x', x{_2}', \ldots, x{_n}'=y')$: then $Pt$ contains the unary atoms $l{_1}', l{_2}', \ldots, l{_n}'$ with $args(l{_i}')=x_i'$, for $1 \leq i \leq n$ such that $(l{_i}', l_{i+1}') \in connpr_{G'}$, for every $1 \leq i <n$. Let $\overline{x_i'}=x_i$. As every node on the path $path_{T_c}(x,y)$ is the counterpart of some node in $path_{T'_c}(x',y')$ and every node in $path_{T'_c}(x',y')$ has the some counterpart in $path_{T_c}(x,y)$, one can conclude that $path_{T_c}(x,y)=(x_1, x_2, \ldots, x_n)$. Also, from the definition of $collapse$ one can see that the presence of unary atoms $l{_i}'$ with $args(l{_i}')=x_i'$ in $Pt / G'$ implies the presence of atoms $l_i$ with $args(l_i)=x_i$ and $pred(l_i)=pred({l_i}')$ in $G$, for every $1 \leq i \leq n$. Furthermore $(l{_i}', l_{i+1}') \in connpr_{G'}$ implies $(l{_i}, l_{i+1}) \in connpr_{G}$, for every $1 \leq i <n$. The latter results leads to: $(l_1, l_n) \in connpr_{G}$ with $args(l_1)=x_1=\overline{x_1}'=x$ and $args(l_n)=x_n=\overline{x_n}'=y$, or in other words to $(pred(l_1), pred(l_n))$ $ \in $ $connpr_{G}(x,y)$. This contradicts with the fact that $(x,y) \in \bl$, and thus $connpr_G(x,y)=\emptyset$.
\end{proof}

\begin{lemma}\label{lemma:shrinkabove}
Let $CS=\langle \EF,$ $\ct,$ $\st,$ $G,$ $\bl \rangle$, $\EF=\langle \F, \ES \rangle $ be a complete clash-free relaxed completion structure. If there are two nodes $u$ and $v$ in a tree $T_c$ in $F$ such that $u<_{T_c}v$, $\ct(u)=\ct(v)$, and there is no blocking node $x^{\prime}$, $x^{\prime} <_{T_c} v$, $collapse_{CS}(u,v)$ returns a complete clash-free relaxed completion structure.
\end{lemma}
\begin{proof}
We have to show that $CS'=collapse_{CS}(u,v)$ is complete, that is, no expansion rule further applies to this completion structure. We will consider every leaf node $x$ of $\EF'$ and show that no rule can be applied to further expand such a node. There are three possible cases as concerns the counterpart of $x$ in $\EF$, $\overline{x}$ (which at its turn is a leaf node in $\EF$):
\begin{itemize}
\item  $\overline{x}$ is  a blocked node in $CS$, which does not make part from the tree $T_c$ from which $u$ and $v$ make part. Let $T_d$ be the tree from which $\overline{x}$ makes part: then there is a node $y' \in T_d$ such that $(y',\overline{x}) \in \bl$. No node was eliminated from $T_d$ as a result of the transformation so for every $z \in path_{T_c}(\overline{x},y{'})$, exists $z' \in \EF'$ such that $\overline{z'}=z$. Thus lemma \ref{lemma:integrityblockingpath} can be applied: $(x,y) \in bl'$, where $y$ is the node in $\EF'$ for which $\overline{y}=y'$. So $x$ is a blocked node in $CS$.

\item  $\overline{x}$ is  a blocked node in $CS$ which makes part from the same tree $T_c$ from which $u$ and $v$ also make part: then there is a node $y' \in T_c$ such that $(y',\overline{x}) \in \bl$. Depending on the location of $y'$ in $T_c$ one can distinguish between the following situations :
\begin{itemize}
\item $y' \not >_{T_c} u$ (Figure \ref{figure:shrink1} a)): in this case $y'$ is on a branch which does not contain $u$ and $v$ (as it is also the case that $y' \not < u$ due to the fact that there is no blocking node $x^{\prime}$ such that $\roo \leq x^{\prime} <v$) and it is not eliminated as a result of applying the transformation, so the path from $\overline{x}$ to $y'$ in $T_c$ is preserved as a result of the transformation. Lemma \ref{lemma:integrityblockingpath} can be applied with the result that $(x,y) \in \bl$ where $y$ is the node in $\EF'$ for which $\overline{y}=y'$
\item $y' \geq_{T_c} u$ and $y' \not \geq v$ (Figure \ref{figure:shrink1} b)): in this case $y'$ is eliminated as a result of applying the transformation, but $\overline{x}$ is also eliminated which contradicts with the existence of $x$ in $CS'$. To see why $\overline{x}$ is also eliminated notice that $y' \not \leq v$ (as again this would contradict with the fact that there is no blocking node $x^{\prime}$ such that $\roo \leq x^{\prime} <v$) and $\overline{x} >y'$. This implies that $\overline{x} >u$ and $\overline{x} \not \leq v$ which suggests that $\overline{x}$ is one of the eliminated nodes, too.
\item $y' \geq v$ (Figure \ref{figure:shrink1} c)): in this case $y'$ is not eliminated as a result of applying the transformation, so the path from $\overline{x}$ to $y'$ in $T_c$ is preserved as a result of the transformation. Lemma \ref{lemma:integrityblockingpath} can be applied with the result that $(x,y) \in \bl$ where $y$ is the node in $\EF'$ for which $\overline{y}=y'$
\end{itemize}

\begin{center}
\begin{figure}[htbp]
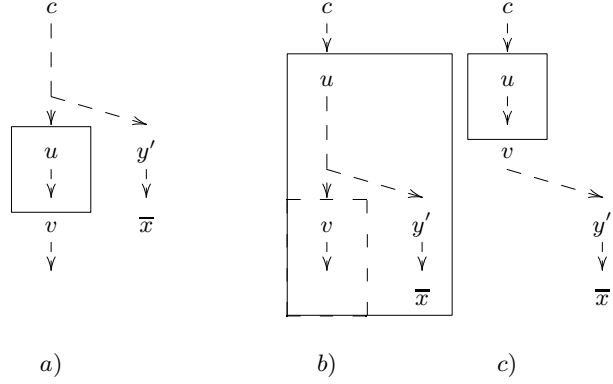

\vspace{3mm}
\Treek[0.5]{1.5}{ \K{$c$} \ARdash{dd}\\
\ARdash{dr}\\
\K{$u$} \ARdash{d} &\K{$y'$} \ARdash{d} \\
\K{$v$} \ARdash{d} &\K{$\overline{x}$}\\
&&\\
a)
\QSS{3,1}{4,1}
}
\Treek[0.5]{1.5}{
\K{$c$} \ARdash{d}\\
\K{$u$} \ARdash{dd} \\
\ARdash{dr}\\
\K{$v$} \ARdash{d} &\K{$y'$} \ARdash{d}\\
&\K{$\overline{x}$}\\
b)
\QS{2,1}{5,2}
\QS[--]{4,1}{5,1}
}
\Treek[0.5]{1.5}{ \K{$c$} \ARdash{d}\\
\K{$u$} \ARdash{d} \\
\K{$v$} \ARdash{dr} \\
&\K{$y'$} \ARdash{d} \\
&\K{$\overline{x}$}\\
c)
\QSS{2,1}{3,1}
}

\vspace{3mm}
\caption{Shrinking a completion structure by eliminating a subtree with a root above any blocking node (the eliminated part is highlighted with continuous line; the part highlighted with dashed line is still kept in)}
\label{figure:shrink1}
\end{figure}
\end{center}

So the conclusion of the analysis above is the existence of a node $y \in T'$ such that $(\overline{y}, \overline{x})\in \bl$. As $connpr_G(\overline{y},\overline{x})=\emptyset$, $connpr_{G'}(y,x)=\emptyset$ as the subtree $T[\overline{y}]$ can be found in $T'$ intact in the form of the subtree $T'[y]$: the eliminated nodes were not part of this subtree as, again, there is no blocking node $x^{\prime}$ in $T$, such that $\roo \leq x^{\prime} <v$.

\item $\overline{x}$ is not a blocked node in $CS$; as $CS$ is complete, no expansion rule can be applied to $\overline{x}$ in $CS$ and, by transfer neither to $x$ in $CS'$ (as they are two nodes which have equal contents which are justified in a similar way).
\end{itemize}
\end{proof}

\begin{lemma}\label{lemma:shrinkunder}
Let $CS=\langle EF,$ $\ct,$ $\st,$ $G,$ $\bl\rangle$ be a complete clash-free relaxed completion structure. If there are three nodes $z$, $u$, and $v$ in $T$ such that $z<u<v$ and there is no blocking node $x^{\prime}$ such that $z < x^{\prime} <v$, and $connpr_G(z,u) \subseteq connpr_G(z,v)$, $collapse_{CS}(u,v)$ returns a complete clash-free relaxed completion structure.
\end{lemma}

\begin{proof}
Like for the lemma above we show that any leaf node in the completion structure $CS'=collapse_{CS}(u,v)$ (or more precisely in the corresponding tree $T'$) cannot be further expanded. Again we consider every such leaf $x$ and we distinguish between three cases as concerns its counterpart in $T$, $\overline{x}$:
\begin{itemize}
\item $\overline{x}$ is  a blocked node in $CS$, which does not make part from the tree $T_c$ from which $u$ and $v$ make part.This case is similar with the first case in the previous lemma.
\item $\overline{x}$ is a blocked node in $CS$ which makes part from the same tree $T_c$ from which $u$ and $v$ make part: then there is a node $y' \in T_c$ such that $(\overline{x}, y') \in \bl$. Using a similar argument as for the previous lemma one concludes that there is a node $y \in T'$ such that $y'=\overline{y}$, or in other words $y'$ has not been eliminated as a result of applying the transformation. In the following we will show that $(y,x) \in bl'$ and $x$ is not further expanded. We will do this on a case-basis considering different locations of $\overline{y}$ and $\overline{x}$ in $T_c$ w.r.t. the nodes $z$, $u$, an $v$ (we consider only those cases in which after the transformation both $\overline{y}$ and $\overline{x}$ are maintained in the structure):
\begin{itemize}
\item $\overline{y} \leq_{T_c} z$ and there is a node $z'$ such that $z'<_{T_c} u$, $z' \geq_{T_c} \overline{y}$, and $\overline{x}>_{T_c} z'$ (Figure \ref{figure:shrink2} a)): in this case the transformation does not remove any node from $path_{T_c}(\overline{y},\overline{x})$ so lemma \ref{lemma:integrityblockingpath} can be applied with the result that $(y,x) \in \bl'
    $.
\item $\overline{y} >_{T_c} v$ (Figure \ref{figure:shrink2} b)): in this case no nodes from the subtree $T_c[\overline{y}]$ are removed during the transformation so using the same argument as above we obtain that $(y,x) \in \bl'$.
\item $\overline{y} \not >_{T_c} z$ and $\overline{y} \not \leq_{T_c} z$ (Figure \ref{figure:shrink2} c)): in this case $\overline{y}$ is not on the same path as $z$, $u$, and $v$ and again the subtree $T_c[\overline{y}]$ is copied intact into $T_c'$, so $(y,x) \in \bl'$.
\item $\overline{y} \leq_{T_c} z$ and $\overline{x} \geq_{T_c} v$: in this case $\overline{y}$, $z$, $u$, $v$ and $\overline{x}$ are all on the same path in $T_c$. Assume by contradiction that $connpr_{G'}(y,x) \neq \emptyset$, or in other words there is a path in $G'$ from a $p(y)$ to some $q(x)$. By lemma \ref{lemma:connection} one obtains that there must be a path $Pt$ between $y$ and $x$ in $\EF{'}$:  note that every such path contains $path_{T'_{c}}(y,x)$. From the same lemma and the previous observation one obtains that there exists a set of unary atoms $l_1, l_2, \ldots, l_n$ in $G'$ with arguments $x_1, x_2, \ldots x_n$, where $path_{T'_{c}}(y,x)=(x_1=y, x_2, \ldots x_n=y)$ such that $(l_i, l_{i+1}) \in connpr_{G'}$, for $1 \leq i <n$. Note that $(l_i, l_{i+1}) \in connpr_{G'}$, for $1 \leq i <n$ implies that $(l_i, l_j) \in connpr_{G'}$, for $1 \leq i <j \leq n$.

    Observe that the counterpart of $z$ from $T_c$ in $T_c'$ is still $z$ and the counterpart of $v$ from $T_c$ in $T_c'$ is $u$, or in other words $\overline{z}=z$ and $\overline{u}=v$. So, $z,u \in path_{T_c'}(x,y)$, or in other words exists $1 \leq j <k \leq n$ such that $x_j=z$ and $x_k=u$. As $(l_1, l_j),$ $(l_j, l_k),$ $(l_k, l_n) \in$ $connpr_{G'}$: $(pred(l_1),$ $ pred(l_j)) \in$ $ connpr_{G'}(y,z)$, $(pred(l_j), pred(l_k)) \in$ $connpr_{G'}(z,u)$, and $(pred(l_k), pred(l_n)) \in$ $connpr_{G'}(u,x)$. 

By definition of $collapse$: $connpr_{G'}(y,u)=$ $connpr_{G}(\overline{y},u)$, $connpr_{G'}(z,u)=$ $connpr_{G}(z,u)$ and $connpr_{G'}(u,y)=$ $connpr_{G}(v,\overline{x})$, so: $(pred(l_1),$ $ pred(l_j)) \in$ $ connpr_{G}(\overline{y},z)$, $(pred(l_j), pred(l_k)) \in connpr_(z,u)$, and $(pred(l_k),$ $ pred(l_n)) \in$ $ connpr_{G'}(v,\overline{x})$. From the lemma condition $connpr_(z,u) \subseteq$ $ connpr_(z,v)$, thus $(pred(l_j), pred(l_k))$ $ \in connpr_{G'}(z,v)$. 

Finally, $(pred(l_1), pred(l_j)) \in $ $ connpr_{G}(\overline{y},z)$, $(pred(l_j),$ $ pred(l_k)) \in$ $ connpr_{G}$ $(z,v)$, and $(pred(l_k), pred(l_n)) \in connpr_{G'}(v,\overline{x})$ implies $(pred(l_1),$ $ pred(l_n)) \in$ $ connpr_{G}(\overline{y},\overline{x})$, which is a contradiction with the fact that $connpr_G(\overline{y}, \overline{x})$ $ = \emptyset$ as $(\overline{y},\overline{x}) \in \bl$. Thus, our assumption is false: $connpr_{G'}$ $(y,x) $ $= \emptyset$, and $(y,x) \in bl'$.
\end{itemize}
\item $\overline{x}$ is not a blocked node in $CS$ (Figure \ref{figure:shrink2} d)); using a similar argument as for the previous lemma one can show that no expansion rule applies to $x$ in $CS'$.
\end{itemize}
\end{proof}

\begin{center}
\begin{figure}[htbp]
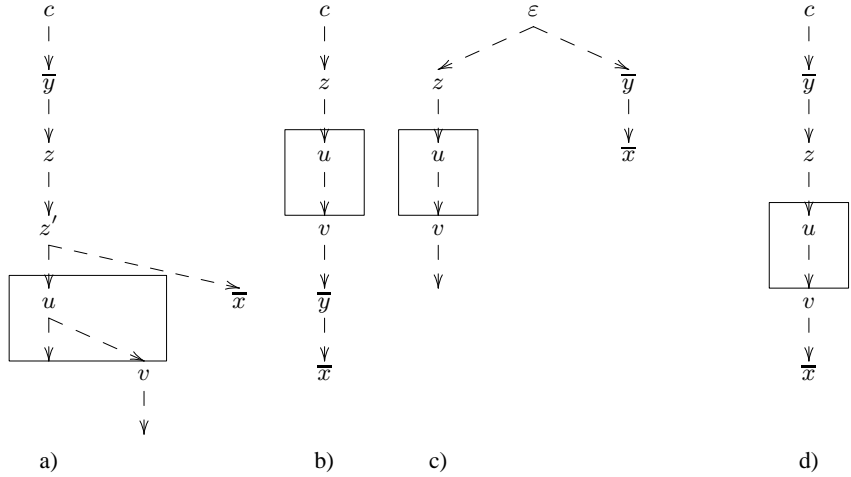

\vspace{3mm}
\Treek[0.5]{1.5}{
\K{$c$} \ARdashc{d}\\
\K{$\overline{y}$}\ARdashc{d}\\
\K{$z$}\ARdashc{d}\\
\K{$z'$}\ARdashc{d}\ARdashc{drr}\\
\K{$u$} \ARdashc{d} \ARdashc{dr}&&\K{$\overline{x}$}\\
&\K{$v$} \ARdashc{d}\\
\Kk{-2}{a)}&\\
\QSSS{5,1}{6,2}
}
\Treek[0.5]{1.5}{
\K{$c$} \ARdashc{d}\\
\K{$z$} \ARdashc{d} \\
\K{$u$} \ARdashc{d} \\
\K{$v$} \ARdashc{d} \\
\K{$\overline{y}$} \ARdashc{d} \\
\K{$\overline{x}$} \\
\Kk{-2}{b)}\\
\QSS{3,1}{4,1}
}
\hspace{3mm}
\Treek[0.5]{1.5}{ &\K{$\roo$} \ARdashc{dl}\ARdashc{dr}\\
\K{$z$} \ARdashc{d} &&\K{$\overline{y}$} \ARdashc{d}\\
\K{$u$} \ARdashc{d} &&\K{$\overline{x}$} \\
\K{$v$}\ARdashc{d}&&&\\
\\
\\
\Kk{-2}{c)}
\QSS{3,1}{4,1}
}
\Treek[0.5]{1.5}{
\K{$c$} \ARdashc{d}\\
\K{$\overline{y}$} \ARdashc{d} \\
\K{$z$} \ARdashc{d} \\
\K{$u$} \ARdashc{d} \\
\K{$v$} \ARdashc{d} \\
\K{$\overline{x}$} \\
\Kk{-2}{d)}\\
\QSS{4,1}{5,1}
}
\vspace{3mm}
\caption{Shrinking a completion structure by eliminating a subtree with a root below a blocking node (the eliminated part is highlighted)}
\label{figure:shrink2}
\end{figure}
\end{center}

Now, we will describe a sequence of transformations on a relaxed clash-free complete completion structure $CS=\langle \EF,$ $\ct,$ $\st,$ $G,$ $\bl $  $\rangle$, $\EF=\langle \F, \ES \rangle$, which returns a complete clash-free completion structure. The transformations which have to be applied to $CS$ are the following (the order in which they are applied is irrelevant):

\begin{itemize}
\item for every two nodes $u$ and $v$ in a tree $T_c \in F$ such that $c<_{T_c}u<_{T_c}v$, $\ct(u)=\ct(v)$, and there is no blocking node $x$, $c \leq_{T_c} x <_{T_c}v$, $collapse_{CS}(u,v)$ (we will call such a transformation a transformation of type 1) ;
\item for every two nodes $u$, and $v$ in a  tree $T_c \in F$ for which there exists a node $z$ in $T_c$ such that $z<_{T_c}u<_{T_c}v$ and there is no blocking node $x$ such that $z <_{T_c} x <_{T_c}v$, and $connpr_G(z,u) \subseteq connpr_(z,v)$, $collapse_{CS}(u,v)$ (we will call such a transformation a transformation of type 2).
\end{itemize}

That the resulted completion structure is complete follows directly from Lemma \ref{lemma:shrinkabove} and Lemma \ref{lemma:shrinkunder}.  We still have to prove the following claim:
\begin{claim}
Let $CS=\langle \EF,$  $\ct,$ $\st,$ $G,$ $\bl \rangle$ be a complete relaxed completion structure to which no transformation of the form described above can be further applied. Then every branch of $CS$ has at most $k=2^p(2^{p^2}-1)+3$ nodes with $p=|\upreds{P}|$.
\end{claim}

We will analyze every branch of every tree $T_c$ at a time. Consider the current branch is $IB$ and that it contains the blocking nodes $x_1, x_2, \ldots x_n$. From Corollary \ref{corollary:limitblockingnodes} we know that $n \leq 2^p$, where $p=|\upreds{P}|$. The last node of the branch will be denoted with $end$ (Figure \ref{figure:branch}). We split the branch $IB$ in $n+1$ paths and count the maximum number of nodes with a certain content in each of these paths. In order to do this need an additional lemma which is defined next.

\begin{lemma}\label{lemma:countequalnodessegment}
Let $IB$ be a branch in a tree $T_c$ as depicted in Figure \ref{figure:branch}. For a given $s \in 2^{\upreds{P}}$:
 \begin{itemize}
 \item
 for any $1 \leq i <n$, there can be at most $2^{p^2}$ nodes in $path_{T_c}(x_i, x_{i+1})$ with content equal to $s$, in case there is no node $x \in T_c$ such that $c<_{T_c}x\leq_{T_c}x_i$ and $\ct(x)=s$
\item for any $1 \leq i <n$, there can be at most $2^{p^2}-1$ nodes in $path_{T_c}(x_i, x_{i+1})$ with content equal to $s$, except for $x_i$, in case there is a node $x \in T_c$ such that $c<_{T_c}x\leq_{T_c}x_i$ and $\ct(x)=s$
\item there can be at most $2^{p^2}$ nodes in $path_{T_c}(x_n, end)$ with content equal to $s$, except for $x_n$.
 \end{itemize}
\end{lemma}
\begin{proof}

We will prove that for any $1 \leq i <n$, there can be at most $2^{p^2}$ nodes in $path_{T_c}(x_i, x_{i+1})$ with content equal to $s$ in case there is no node $x \in T_c$ such that $c<_{T_c}x\leq_{T_c}x_i$ and $\ct(x)=s$. Assume by contradiction that there are at least $2^{p^2}+1$ nodes in $path_{T_c}(x_i, x_{i+1})$ with content equal to $s$. Let's call these node $y_1, y_2, \ldots, y_m$, where $m > 2^{p^2}$. It is necessary that $connpr_G(y_1, y_i) \supset connpr_G(y_1, y_{i+1})$ for every $1 < i <m$, otherwise a transformation of type 2 could be further applied to $CS$. As $connpr_G(x,y)\subseteq\upreds{P}\times\upreds{P}$ and $|2^{\upreds{P}\times\upreds{P}}|=2^{p^2}$, and there at least $2^{p^2}$ distinct values for $connpr_G(y_1, y_i)$, when $1 < i <m$, there must be an $1 < i <m$ such that $connpr_G(y_1, y_i)$ $=\emptyset$. But in this case $(y_1, y_i) \in \bl$ (as the two nodes also have equal content) which contradicts with the fact that $y_i \neq end$.
The other cases are proved similarly.
\end{proof}

Now we will proceed to the actual counting. Let $s \in 2^{\upreds{P}}$ be a possible content value for any node in $IB$. We will count the maximum number of nodes with content $s$ in $IB$ - in order to do this we have to distinguish between three different cases as regards $s$:

\begin{center}
\begin{figure}[htbp]
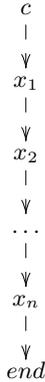

\vspace{3mm}
\Treek[0.5]{1.5}{
\K{$c$} \ARdashc{d}\\
\K{$x_1$}\ARdashc{d}\\
\K{$x_2$}\ARdashc{d}\\
\K{$\ldots$}\ARdashc{d}\\
\K{$x_n$} \ARdashc{d} \\
\K{$end$} \\
}
\caption{A random branch $IB$ in the resulted complete clash-free relaxed completion structure: $x_1$, \ldots, $x_n$ are blocking nodes}
\label{figure:branch}
\end{figure}
\end{center}

\begin{itemize}
\item there is no node $x \in T_c$ with $c<_{T_c}x<_{T_c}x_1$ such that $\ct(x)=s$, and there is no $1 \leq i \leq n$ such that $\ct(x_i)=s$. In this case there is maximum 1 node with content equal to $s$ in $path_{T_c}(c,x_1)$ (the root), maximum $2^{p^2}$ nodes in each $path_{T_c}(x_i,x_i+1)$ and maximum $2^{p^2}$ nodes in $path_{T_c}(x_n,end)$ (according to lemma \ref{lemma:countequalnodessegment}); for the last path there cannot be $2^{p^2}+1$ nodes as that would mean that $end$  is a blocked node with content equal to $s$, so there would be a blocking node with content equal to $s$, which contradicts with the fact the hypothesis there is no blocking node with content equal to $s$). Also there are at most $2^p-1$ blocking nodes (if there would be $2^p$ such nodes, the maximum indicated by corollary \ref{corollary:limitblockingnodes} there would remain no valid value for $s$). Summing all up, in this case there are at most $2^{p^2}(2^p-1)+1$ nodes with content equal to $s$.
\item there is no node $x$ such that $c<_{T_c}x<_{T_c}x_1$ such that $\ct(x)=s$ but there is a node $x_i$, $1 \leq i \leq n$ such that $\ct(x_i)=s$. In this case there is no node $x$ such that $c<_{T_c}x<_{T_c}x_i$ which has content equal to $s$ (lemma \ref{lemma:blockingnodefirst}), and thus $path_{T_c}(c,x_1)$ maximum 1 node with content equal to $s$ (the root). $path_{T_c}(x_i,x_{i+1})$ has maximum $2^{p^2}$ nodes, every path $(x_j, x_{j+1})$, where $i < j <n$ has maximum $2^{p^2}-1$ nodes, and the path $(x_n, end)$ has maximum $2^{p^2}$ nodes (according to lemma \ref{lemma:countequalnodessegment}). Summing all up, in this case there are at most $(2^{p^2}-1)(n-i+1)+3$ nodes with content equal to $s$, where $n$ is the number of blocking nodes. There are at most $2^p$ blocking nodes (corollary \ref{corollary:limitblockingnodes}), so the maximum of the expression is met when $i=1$ and $n=2^p$ and is  $2^p(2^{p^2}-1)+3$.
\item there is a node $x$ such that $c<_{T_c}x<_{T_c}x_1$ and $\ct(x)=s$. In this case $\ct(x_i) \neq s$, for every $1 \leq i \leq n$ (lemma \ref{lemma:blockingnodefirst}). The counting is as follows: $path_{T_c}(c,x_1)$ has maximum 1 node with content equal to $s$ ($x$), otherwise a transformation of type 1 could be applied, $path_{T_c}(x_i,x_{i+1})$ has maximum $2^{p^2}-1$ nodes, $1 \leq i <n$ and the path $(x_n, end)$ has maximum $2^{p^2}$ nodes (according to lemma \ref{lemma:countequalnodessegment}). Also there are at most $2^p-1$ blocking nodes (if there would be $2^p$ such nodes, the maximum indicated by corollary \ref{corollary:limitblockingnodes} there would remain no valid value for $s$). Summing all up, in this case there are at most $(2^{p^2}-1)(2^p-1)+1$ nodes with content equal to $s$.
\end{itemize}

From the three cases the maximum of number of nodes with content equal to a given $s$ in any branch $IB$ of a tree $T_c \in F$ is $2^p(2^{p^2}-1)+3$, which is exactly $k$.

At this point we have a complete relaxed clash-free completion structure with at most $k$ nodes on any branch, thus a complete clash-free completion structure for $p$ w.r.t. $P$.
\end{proof}

\end{document}